\documentclass[11pt]{article}

\usepackage[margin=1in]{geometry}
\PassOptionsToPackage{numbers, compress}{natbib}

\usepackage{natbib}
\usepackage{amsmath}
\usepackage{algorithm}
\usepackage{algpseudocode}
\usepackage{float}
\usepackage{placeins}
\usepackage{tikz}
\usetikzlibrary{positioning}

\DeclareMathOperator{\E}{\mathbb{E}}

\usepackage[utf8]{inputenc} 
\usepackage[T1]{fontenc}    
\usepackage{hyperref}       
\usepackage{url}            
\usepackage{booktabs}       
\usepackage{amsfonts}       
\usepackage{amsmath}        
\usepackage{nicefrac}       
\usepackage{microtype}      
\usepackage{xcolor}         
\usepackage{float}
\usepackage{graphicx}
\usepackage{caption}
\usepackage{multirow}

\newtheorem{theorem}{Theorem}[section]
\newtheorem{proposition}[theorem]{Proposition}

\newtheorem{example}[theorem]{Example}
\newtheorem{remark}[theorem]{Remark}
\newtheorem{proof}[theorem]{Proof}

\title{Causal Survival Forests with Negative Controls}

\author{%
	\begin{tabular}{@{}ccc@{}}
		Zijun Gao%
		\thanks{
			Department of Data Sciences and Operations, USC Marshall School of
			Business, University of Southern California.
			\texttt{zijungao@marshall.usc.edu}}
		&
		Kyoungeui Hong%
		\thanks{Department of Computer Science, University of Southern
			California. \texttt{hongkyou@usc.edu}}
		&
		Leyi Ma%
		\thanks{Department of Computer Science, University of Southern California.
			\texttt{rma99604@usc.edu}}
		\\[0.75em]
		Qianli Wu%
		\thanks{Department of Computer Science, University of Southern
			California. \texttt{ewu96023@usc.edu}}
		&
		Zachary Izzo%
		\thanks{NEC Labs America.
			\texttt{zle.izzo@gmail.com}}
		&
		Ruishan Liu%
		\thanks{Department of Computer Science, University of Southern California.
			\texttt{ruishanl@usc.edu}}
	\end{tabular}
}
\begin{document}

	\maketitle
	
	\begingroup
	\renewcommand{\thefootnote}{}
	\footnotetext{The first four authors contributed equally. The last two authors are co-corresponding authors.}
	\addtocounter{footnote}{-1}
	\endgroup

\begin{abstract}
We study heterogeneous treatment-effect (HTE) estimation in observational survival studies commonly associated with both censored outcomes and unmeasured confounding.
We integrate causal survival forests (CSF) with negative controls (NC) from proximal causal inference and introduce Negative Control Causal Survival Forests (NC-CSF), a flexible nonparametric HTE learner for survival analysis. 
Our approach uses a loss that incorporates proxy variables and Neyman orthogonalization to train the random forest, thereby mitigating bias from unobserved confounding and gaining robustness to nuisance estimation.
Through extensive simulations spanning varying levels of confounding, proxy relevance, and censoring mechanisms, we demonstrate that NC-CSF substantially reduces bias and estimation error relative to existing baselines. 
We further demonstrate the practical utility of our method on various clinical datasets, where it confirms several existing findings and also reveals new interpretable patterns of treatment-effect heterogeneity. 
To facilitate practical use, we provide an end-to-end Python implementation of NC-CSF that carefully handles implementation details such as nuisance estimation and clipping.
\end{abstract}


\section{Introduction}





Time-to-event outcomes are common in treatment-effect analysis, especially in clinical trials and biomedical studies \cite{kalbfleisch2002statistical, imbens2015causal}. These survival outcomes are frequently right-censored, for example because follow-up ends before the event is observed or participants leave the study early. 
In recent years, treatment-effect heterogeneity for survival outcomes has received increasing attention \citep{henderson2020individualized,tabib2020nonparametric,hu2021estimating,cui2023estimating,hu2022flexible}, where the causal effect of treatment may vary across subgroups defined by baseline covariates. 
To model such heterogeneity flexibly in survival settings, machine-learning-assisted estimators have been developed. Among them, causal survival forests \citep{cui2023estimating,wager2018estimation,athey2019generalized} provide a promising framework, as they can estimate heterogeneous treatment effects in a flexible and nonparametric manner. 

Many existing methods for heterogeneous treatment-effect analysis with survival outcomes rely on an ignorability assumption, that is, they assume there is no unobserved confounding. However, this assumption may be unrealistic in many observational survival datasets, where unobserved confounding can induce bias in treatment-effect estimates when it is not properly adjusted for \citep{imbens2015causal,huang2020sensitivity}.
To address unobserved confounding, proximal causal inference \citep{miao2018identifying,tchetgen2020proximal,cui2023semiparametric,zhang2023proximal,tchetgen2023longitudinal,liu2025regression} offers a promising framework by leveraging negative-control treatments and negative-control outcomes, that is, auxiliary variables that are informative about latent confounders while satisfying exclusion-type restrictions with respect to the primary treatment and outcome. 
Negative-control variable assumptions may be more suitable in some biomedical applications than the classical no-unmeasured-confounding assumption, for example, pre-treatment measurements may share latent risk factors while remaining unaffected by treatment.
This literature has also recently been extended to survival settings with censored data, primarily for average treatment effects and under parametric modeling assumptions \citep{li2025regression}.

Motivated by the practical need for flexible nonparametric modeling of heterogeneous treatment effects, as well as by the appeal of addressing unobserved confounding via negative controls, we connect proximal causal inference with causal survival forests and develop \emph{negative-control causal survival forests} (NC-CSF).
In particular, our contributions are as follows. 
\begin{enumerate}
    \item \textit{Method}.
    We design a loss for random forests to learn heterogeneous treatment effects with censored survival outcomes that explicitly incorporates negative controls to adjust for unobserved confounding and employs Neyman orthogonalization to achieve robustness to nuisance-function estimation.

    \item \textit{Theory}. We provide theoretical justification for the proposed approach by showing how the negative-control-based loss mitigates bias from unobserved confounding and by establishing a local doubly robust property.

    \item \textit{Applications}. We evaluate the proposed method on both simulated and real datasets. In particular, on an AIDS clinical trial dataset, our method suggests that the fitted RMST treatment effect is negative for younger individuals and positive for older individuals, consistent with existing treatment recommendations, whereas existing methods yield uniformly positive estimates across ages.

    \item \textit{Packages}. We provide a ready-to-use Python package for the proposed method, where we have carefully addressed several practical issues required for stable finite-sample performance, such as nuisance estimation and variable clipping.
\end{enumerate}

\paragraph{Organization.} Section~\ref{sec:formulation} introduces the problem setup and background. Section~\ref{sec:method} presents our proposed NC-CSF estimator and provides its theoretical justification. Section~\ref{sec:simulations} studies NC-CSF's performance in simulations.
Section~\ref{sec:real.data} illustrates the method on real clinical data. 
We conclude with a discussion in Section~\ref{sec:discussion}. Proofs and additional numerical results are deferred to the appendix.

\section{Background}\label{sec:formulation}

We first formulate the problem of heterogeneous treatment-effect estimation in the presence of right censoring and unmeasured confounding. We then review the two key ingredients underlying our proposal: causal survival forests and negative controls.

\subsection{Problem Formulation}

We follow the potential-outcome framework \citep{rubin1974estimating,imbens2015causal} and consider an observational survival study with $n$ units. Each unit is associated with a binary treatment $A_i \in \{0,1\}$, observed baseline covariates $X_i$, and unobserved confounders $U_i$. 
Let $T_i(a)$ and $C_i(a)$, for $a \in \{0,1\}$, denote the potential event time and potential censoring time, respectively. The realized event and censoring time are given by $T_i = T_i(A_i)$ and $C_i = C_i(A_i)$, and the observed survival time and censoring indicator are $Y_i = \min(T_i,C_i)$ and $\Delta_i = I\{T_i \le C_i\}$. 
Let $h > 0$ be a fixed horizon, and we define the horizon-adjusted follow-up time and indicator by $Y_i^{(h)}=\min(Y_i,h)$ and $\Delta_i^{(h)}=\Delta_i\,\mathbf{1}\{Y_i\le h\}+\mathbf{1}\{Y_i>h\}$.

We consider two sets of causal estimands derived from the potential event times. The first is the restricted mean survival time (RMST), defined by $\min\{T_i(a),h\}$, and the second is the survival probability at horizon $h$, defined by $I\{T_i(a)>h\}$. Let $G_i(a;h)$ denote either of these two functionals, and use $G_i(h)$ to denote $G_i(A_i;h)$. The associated heterogeneous treatment-effect estimands are
\begin{align*}
&\tau
:=
\mathbb{E}\left[ G_i(1;h)-G_i(0;h) \right],\\
&\tau(x)
:=
\mathbb{E}\left[ G_i(1;h)-G_i(0;h) \mid X_i=x \right], \\
&\tau^w(x)
:=
\mathbb{E}\left[ w(X_i,U_i)(G_i(1;h)-G_i(0;h)) \mid X_i=x \right], \\
&\tau(x,u)
:=
\mathbb{E}\left[G_i(1;h)-G_i(0;h) \mid X_i=x,U_i = u \right],
\end{align*}
where $w(X_i,U_i) \ge 0$ denotes some weight function such that $\mathbb{E}[w(X_i,U_i) \mid X_i] = 1$ (we define a specific weight function in Section~\ref{sec:theory}).
Here, $\tau$ is the average treatment effect, $\tau(x)$ is the treatment effect conditional on the observed covariates, $\tau(x,u)$ is the treatment effect conditional on the full confounding state, and \(\tau^w(x)\)
denotes a weighted average of \(\tau(x,U_i)\) over the conditional
distribution of \(U_i\) given \(X_i=x\).

We state the assumptions needed for our proposal.
We assume the stable unit treatment value assumption (SUTVA), which says there is no interference between units and no hidden versions of treatment. We also assume overlap given both $X_i$ and $U_i$, namely
\begin{align*}
0 < P(A_i=1 \mid X_i,U_i) < 1, \quad \text{a.s.}
\end{align*}
Importantly, we do \emph{not} assume ignorability given the observed covariates $X_i$, and we allow treatment assignment to depend on the latent confounder $U_i$.





\subsection{Causal Survival Forest Preliminaries}\label{sec:CSF}
To review the causal survival forests (CSF) in \citep{cui2023estimating} under no unmeasured confounding, we first introduce some new notation.
 Let $e(X_i)=\mathbb{P}(A_i=1\mid X_i)$ denote the propensity score, and let $m(X_i)=E[G_i(h)\mid X_i]$ denote the marginal mean function. Let $S_C(t\mid X_i,A_i)=\mathbb{P}(C_i>t\mid X_i,A_i)$ and $\Lambda_C(t\mid X_i,A_i)=\int_0^t \lambda_C(u\mid X_i,A_i)\,du$ denote the conditional censoring survival function and cumulative hazard, respectively, where $\lambda_C(t\mid X_i,A_i)$ is the conditional censoring hazard function defined by
$\lambda_C(t\mid X_i,A_i)
=
\lim_{\Delta t\to 0}
{
\mathbb{P}(t \le C_i < t+\Delta t \mid C_i \ge t,\; X_i,A_i)
}/{
\Delta t
}$.

In CSF, the following loss is used for node splitting and leaf fitting within each tree,
\begin{align*}
\ell^{\mathrm{CSF}}
&=
\sum_{i=1}^n
\Bigl(
\Gamma_i^{\mathrm{CSF}}(h)-m(X_i)-\tau(X_i)\bigl(A_i-e(X_i)\bigr)
\Bigr)^2, \\
\Gamma_i^{\mathrm{CSF}}(h)
&=
\frac{
\Delta_i^{(h)} G_i(h)
+
\bigl(1-\Delta_i^{(h)}\bigr)\rho_{A_i}\!\bigl(Y_i^{(h)},X_i;h\bigr)
}{
S_C\!\bigl(Y_i^{(h)}\mid X_i,A_i\bigr)
}
-
\int_0^{Y_i^{(h)}}
\frac{
\rho_{A_i}(t, X_i;h)
}{
S_C\!\bigl(t\mid X_i,A_i\bigr)
}
\,d\Lambda_C\!\bigl(t\mid X_i,A_i\bigr),
\end{align*}
where $\Gamma_i^{\mathrm{CSF}}(h)$ is the censoring-corrected outcome and $\rho_a(t, X_i;h)=E[G_i(a;h)\mid T_i(a)>t,X_i]$ denotes the expected outcome conditional on the unit surviving to time $t$.
Results for uncensored data can be recovered as a special case of the censored setting by assuming $C_i(0)=C_i(1)=\infty$, a.s.

A key property of this construction is its robustness to nuisance-function estimation (here the nuisance functions include $\rho_a$, $S_C$, $\Lambda_C$, $m$, and $e$). 
In particular, the associated score function is Neyman orthogonal, meaning that its pathwise derivative regarding the nuisance functions is zero at the truth, and only second-order nuisance estimation errors remain locally.
In Section~\ref{sec:method}, we extend the loss of CSF to incorporate negative controls to address unobserved confounders.

\subsection{Negative Controls and Bridge Functions Preliminaries}

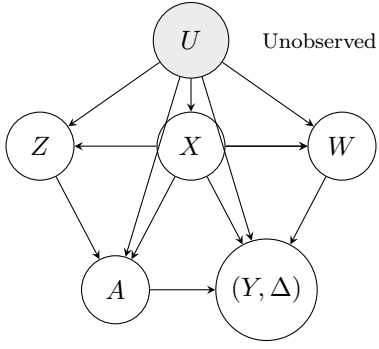
\begin{figure}
    \centering
\begin{tikzpicture}[
  >=stealth,
  every node/.style={font=\small},
  obs/.style={draw, circle, minimum size=9mm},
  latent/.style={draw, circle, minimum size=10mm, fill=gray!15},
  note/.style={font=\scriptsize}
]
\node[latent] (U) at (0,2.2) {$U$};
\node[obs] (Z) at (-2,0.8) {$Z$};
\node[obs] (X) at (0,0.8) {$X$};
\node[obs] (W) at (2,0.8) {$W$};
\node[obs] (A) at (-1,-1.1) {$A$};
\node[obs] (Y) at (1,-1.1) {$(Y,\Delta)$};

\draw[->] (U) -- (Z);
\draw[->] (U) -- (X);
\draw[->] (U) -- (W);
\draw[->] (U) -- (A);
\draw[->] (U) -- (Y);

\draw[->] (X) -- (Z);
\draw[->] (X) -- (W);
\draw[->] (X) -- (A);
\draw[->] (X) -- (Y);

\draw[->] (Z) -- (A);
\draw[->] (W) -- (Y);
\draw[->] (A) -- (Y);

\draw[->] (X) -- (W);

\node[note, right=3mm of U] {Unobserved};
\end{tikzpicture}
\caption{Directed acyclic graph illustrating the assumptions of the negative-control treatment $Z$ and negative-control outcome $W$.}
\label{fig:diag}
\end{figure}

To review negative controls and bridge functions for latent confounders, we first introduce some additional notation. The standard proximal causal inference framework \citep{miao2018identifying,tchetgen2020proximal,cui2023semiparametric} assumes access to a negative-control outcome proxy $W_i$ and a negative-control treatment proxy $Z_i$. Intuitively, $W_i$ is associated with the latent confounder but is not affected by the treatment, while $Z_i$ is associated with treatment assignment but does not directly affect the outcome. 
The conditional independence structure is shown in Figure~\ref{fig:diag}; the formalized assumption adapted to our setting is in \eqref{eq:nc_censored}.

To remove bias induced by unobserved confounding, a common way of using negative controls is through bridge functions. In the standard setting of estimating the average treatment effect, one may consider bridge functions satisfying $\mathbb{E}\left[f_a(W_i,X_i)\mid X_i,U_i\right] = \mathbb{E}\left[Y_i(a)\mid X_i,U_i\right]$. When such functions exist, a key benefit of these bridge relations is that, although $U_i$ is not observed, we have $\mathbb{E}[f_a(W_i,X_i)] = \mathbb{E}\!\left[\mathbb{E}\!\left[f_a(W_i,X_i)\mid X_i,U_i\right]\right] = \mathbb{E}\!\left[\mathbb{E}\!\left[Y_i(a)\mid X_i,U_i\right]\right] = \mathbb{E}[Y_i(a)]$, so that unbiased estimation can still be achieved using observable variables alone. In existing work, estimators are typically formed by averaging bridge-function-based functionals over units, whereas we instead use the bridge functions to construct a loss to train random forests that allows more flexible modeling of $\tau(x)$ and facilitates the use of well-developed machine learning tools.




\section{Method}\label{sec:method}


\subsection{Negative Control and Bridge Functions for Survival Analysis}

In this paper, we assume access to both a negative-control outcome and a negative-control treatment, satisfying the conditional independence structure
\begin{align}\label{eq:nc_censored}
(A_i,Z_i)\perp (T_i(a),C_i(a),W_i)\mid U_i,X_i.
\end{align}
This means that, conditional on the full confounding state $(X_i,U_i)$, the treatment $A_i$ and negative-control treatment $Z_i$ are independent of the potential event time, the potential censoring time, and the negative-control outcome $W_i$.
Although the assumption~\eqref{eq:nc_censored} is not testable from the observed data alone, it may be plausible in biomedical studies by temporal ordering; for example, pre-treatment measurements can serve as negative-control outcomes as they may share latent risk factors \citep{tchetgen2020proximal,li2025regression}.

Given $W_i$ and $Z_i$, we consider bridge functions $f_a(w,x)$, $a\in\{0,1\}$, and $q(z,x)$ satisfying
\begin{align}\label{eq:bridge.function}
\begin{split}
\mathbb{E}\left[f_a(W_i,X_i)\mid X_i,U_i\right] &= \mathbb{E}\left[G_i(a;h)\mid X_i,U_i\right], \\
\mathbb{E}\left[q(Z_i,X_i)\mid X_i,U_i\right] &= \mathbb{P}(A_i=1\mid X_i,U_i).
\end{split}
\end{align}
These bridge relations imply that, although $U_i$ is unobserved, the latent conditional outcome and treatment mechanisms can be represented through observable functions of $(W_i,X_i)$ and $(Z_i,X_i)$, respectively. 
We note that our specification of $q$ differs mildly from that in \cite{miao2018identifying}, where the conditional expectation of the bridge function is specified to match $1 / \mathbb{P}(A_i = a \mid X_i, U_i)$.


Finally, for censoring, we assume the standard conditional independent censoring assumption \citep{kalbfleisch2002statistical,cui2023estimating} given the observed variables including the negative controls:
\begin{align}\label{eq:censoring}
C_i(a)\perp T_i(a)\mid X_i,W_i,Z_i,A_i=a, \quad a \in \{0,1\}.
\end{align}
We next introduce some additional notation for censored survival outcomes in the presence of negative-control treatment and outcome variables. Let $m(x,w,z)=E[G_i(h)\mid X_i=x,W_i=w,Z_i=z]$ denote the marginal mean function. Define $S_C(t\mid X_i,W_i,Z_i,A_i)$, 
$\rho_a(t\mid X_i,W_i,Z_i)$, $\Lambda_C(t\mid X_i,W_i,Z_i,A_i)$, $\lambda_C(t\mid X_i,W_i,Z_i,A_i)$
analogously to Section~\ref{sec:CSF}.

\subsection{Oracle NC-CSF}

We construct the loss for NC-CSF using both bridge functions $q$ and $f_a$, and then fit a random forest by minimizing it. 
We refer to this as the oracle NC-CSF, since we assume that all true nuisance functions, including the bridge functions, are provided.
Below we focus on the loss construction.

Analogous to the residualization in \citep{robinson1988root,nie2021quasi}, we introduce residualized treatment and outcome,
\begin{align}\label{eq:residual}
\begin{split}    
A_i^{\mathrm{res}} &= A_i-q_i, \\
Y_i^{\mathrm{res}}
&=
\frac{
\Delta_i^{(h)} \bigl(Y_i^{(h)}-m_i\bigr)
+
\bigl(1-\Delta_i^{(h)}\bigr)
\bigl(\rho_{A_i}(Y_i^{(h)},X_i,W_i,Z_i;h)-m_i\bigr)
}{
S_C(Y_i^{(h)} \mid X_i,W_i,Z_i,A_i)
}
\\
&\quad
-
\int_0^{Y_i^{(h)}}
\frac{d\Lambda_C(t \mid X_i,W_i,Z_i,A_i)}
{S_C(t \mid X_i,W_i,Z_i,A_i)}
\bigl(\rho_{A_i}(t,X_i,W_i,Z_i;h)-m_i\bigr)\,dt.
\end{split}
\end{align}
Here, $A_i^{\mathrm{res}}$ is the bridge-adjusted treatment residual satisfying $E[A_i^{\mathrm{res}}\mid X_i,U_i]=0$, and $Y_i^{\mathrm{res}}$ denotes the censoring-corrected residualized outcome. In the absence of censoring, $Y_i^{\mathrm{res}}$ simplifies to $Y_i^{(h)}-m_i$, as illustrated in Example~\ref{exam:uncensored_rmst}.
Finally, we define the NC-CSF loss based on the residuals again motivated by \citep{robinson1988root,nie2021quasi},
\begin{align}\label{eq:loss}
\ell^{\mathrm{NC\text{-}CSF}}
=
\sum_{i=1}^n
\left(
Y_i^{\mathrm{res}}
-
\tau(X_i)\,A_i^{\mathrm{res}}
\right)^2.
\end{align}

\begin{example}[Uncensored case]\label{exam:uncensored_rmst}
If there is no censoring, then $C_i=\infty$ almost surely and hence $\Delta_i=1$ for all $i$, $\Delta_i^{(h)}=1$, $S_C(t\mid X_i,W_i,Z_i,A_i)=1$, $\Lambda_C(t\mid X_i,W_i,Z_i,A_i)=0$, and thus $Y_i^{\mathrm{res}}=Y_i^{(h)}-m_i$.
\end{example}

\subsection{Empirical NC-CSF}

In this section, we discuss how to estimate the nuisance quantities and plug them into the oracle loss for fitting, yielding the (empirical) NC-CSF. 
Because \(U_i\) is unobserved,
the latent bridge equations in \eqref{eq:bridge.function} cannot be
evaluated directly from the observed data. We therefore use the observable conditional moment equations stated in Proposition~\ref{prop:bridge_observed_moment} below, which do not involve \(U_i\) and are implied by Eq.~\eqref{eq:nc_censored} and Eq.~\eqref{eq:bridge.function}.

\begin{proposition}
\label{prop:bridge_observed_moment}
Suppose $f$ and $q$ satisfy \eqref{eq:bridge.function}, then, for each $a$,
\begin{align}\label{eq:bridge_observed_prop_combined}
\begin{split}    
\mathbb{E}\!\left[A_i-q(Z_i,X_i)\mid X_i,W_i\right] = 0,\quad 
\mathbb{E}\!\left[G_i(a;h)-f_a(W_i,X_i)\mid X_i,Z_i,A_i=a\right] = 0.
\end{split}
\end{align}
Conversely, if the following completeness conditions hold: $\mathbb{E}[r(X_i,U_i)\mid X_i,W_i]=0$ a.s. implies $r(X_i,U_i)=0$ a.s. and $\mathbb{E}[s(X_i,U_i)\mid X_i,Z_i,A_i=a]=0$ a.s. implies $s(X_i,U_i)=0$ a.s. for any square-integrable measurable function $r$, $s$, then \eqref{eq:bridge_observed_prop_combined} implies \eqref{eq:bridge.function}.
\end{proposition}

When censoring is present, \(G_i(a;h)\) may be unobserved. 
Under \eqref{eq:nc_censored} and \eqref{eq:censoring}, it can be replaced by an inverse censoring probability weighted pseudo-outcome that satisfies the same conditional moment equation (details in the Appendix).
To learn survival-related nuisance functions, under \eqref{eq:censoring}, those nuisance functions are not required to satisfy conditions involving the unobserved confounder $U_i$, and thus can be estimated using standard survival-analysis procedures, such as Cox models. 
Finally, to avoid overfitting in the nuisance stage, we employ cross-fitting \cite{chernozhukov2018double}.
See Algorithm~\ref{alg:censored_final_training} for a brief summary. Additional implementation details and a more detailed version of Algorithm~\ref{alg:censored_final_training} are deferred to the Appendix.

\begin{algorithm}[H]
\caption{Censored NC-CSF (Brief)}
\label{alg:censored_final_training}
\begin{algorithmic}[1]
\Require Data $\{(X_i,W_i,Z_i,A_i,Y_i,\Delta_i)\}_{i=1}^n$, causal estimand, horizon $h$, number of folds $K$
\State Split the sample into $K$ folds for cross-fitting.
\For{$k=1,\ldots,K$}
    \State Fit the censoring model, survival models, treatment and outcome bridges on the training folds.
    \State Construct $A_i^{\mathrm{res}}$, $Y_i^{\mathrm{res}}$ in \eqref{eq:residual} on the held-out fold.
\EndFor
\State Fit the final causal forest on $\{(X_i,A_i^{\mathrm{res}},Y_i^{\mathrm{res}})\}_{i=1}^n$ with the loss \eqref{eq:loss}.
\State \Return $\widehat\tau(\cdot)$
\end{algorithmic}
\end{algorithm}

\subsection{Theories}\label{sec:theory}

\begin{proposition}[Validity of the NC-CSF score]\label{prop:score}
Assume \eqref{eq:nc_censored} and \eqref{eq:bridge.function}. Then the score has conditional mean zero given $X_i=x$ at
\begin{align*}
\tau^w(x)
=
\frac{
\mathbb{E}\!\left[(A_i-q(Z_i,X_i))^2 \tau(X_i,U_i)\mid X_i=x\right]
}{
\mathbb{E}\!\left[(A_i-q(Z_i,X_i))^2\mid X_i=x\right]
},
\end{align*}
where $\tau(x,u)=\mathbb{E}[Y(1)-Y(0)\mid X=x,U=u]$. 

In particular, when the latent variable $U$ does not modify the treatment effect beyond $X$, that is, $\tau(x,u)=\tau(x)$, we have $\tau^w(x)=\tau(x)$.
\end{proposition}

\begin{remark}[Interpretation of the weighted target]\label{rmk:tau.w}

We treat $\tau^w(x)$ as the general NC-CSF estimand. It remains a heterogeneous treatment-effect target over the observed covariates even when the latent confounder also modifies treatment effects. By
\begin{align}\label{eq:weighted_cate_difference}
\tau^w(x)-\tau(x)
=
\frac{\operatorname{Cov}(\omega(x,U_i),\tau(x,U_i)\mid X_i=x)}
{\mathbb{E}[\omega(x,U_i)\mid X_i=x]},
\end{align}
$\tau^w(x)$ equals the standard CATE $\tau(x)$ when $U$ does not modify the effect beyond $X$, or when $\omega(x,U)$ and $\tau(x,U)$ are conditionally uncorrelated given $X=x$.

The weighting is analogous to the overlap weighting induced by an R-learner with omitted effect modifiers \citep{nie2021quasi}.
In addition, among all positive latent weighting functions \(w'\), the choice \(w'(x,u)\propto\omega(x,u)\) minimizes the oracle conditional variance, which implies \(\tau^w(x)\) is variance-optimal within this class of weighted estimands.
\end{remark}

\begin{proposition}[Neyman orthogonality of the NC-CSF score]\label{prop:Neyman.orthogonality}
Let $\eta=(q,m,S_C,\rho_a,\Lambda_C)$ denote the set of nuisance functions. 
At the true nuisance value $\eta_0$, the pathwise derivative of the score of the NC-CSF loss regarding each nuisance component has conditional expectation zero given $X$.
\end{proposition}
Proposition~\ref{prop:Neyman.orthogonality} implies the loss is only second-order sensitive to the nuisance estimation error.

\begin{proposition}
\label{prop:leaf_oracle}
Let $\mathcal S$ be a fixed subgroup, or a subgroup learned on data independent of the sample used for within-subgroup estimation, and let $n_{\mathcal S}\to\infty$ denote its estimation-sample size. Define
\begin{align*}
\tau_{\mathcal S}^w
=
\frac{
\mathbb{E}[\{A_i-q(Z_i,X_i)\}^2\tau(X_i,U_i)\mid X_i\in\mathcal S]
}{
\mathbb{E}[\{A_i-q(Z_i,X_i)\}^2\mid X_i\in\mathcal S]
}.
\end{align*}
Assume the following:
\begin{enumerate}
    \item \textit{Boundedness:} the outcomes, covariates, and nuisance functions $q$, $m$, $S_C$, $\rho_a$, and $\Lambda_C$ are uniformly bounded, and $S_C$ is bounded away from zero on the time range of interest;
    \item \textit{Residual overlap:} $\mathbb{E}[\{A_i-q(Z_i,X_i)\}^2\mid X_i\in\mathcal S]$ is bounded away from zero;
    \item \textit{Nuisance estimation:} the cross-fitted estimators $\widehat q$, $\widehat m$, $\widehat S_C$, $\widehat\rho_a$, and $\widehat\Lambda_C$ are consistent and each has $L_2$ error $o_p(n_{\mathcal S}^{-1/4})$ on the subgroup.
\end{enumerate}
Then the cross-fitted NC-CSF estimator is asymptotically equivalent to its oracle-nuisance counterpart and
\begin{align*}
\sqrt{n_{\mathcal S}}\,(\widehat\tau_{\mathcal S}-\tau_{\mathcal S}^w)
\overset{d}{\longrightarrow}
N(0,\sigma_{\mathcal S}^2),
\end{align*}
where
\begin{align*}
\sigma_{\mathcal S}^2
=
\frac{
\operatorname{Var}\!\left(
A_i^{\mathrm{res}}
(Y_i^{\mathrm{res}}-\tau_{\mathcal S}^w A_i^{\mathrm{res}})
\mid X_i\in\mathcal S
\right)
}{
\mathbb{E}[(A_i^{\mathrm{res}})^2\mid X_i\in\mathcal S]^2
}.
\end{align*}
\end{proposition}

In Proposition~\ref{prop:leaf_oracle}, if $\tau(x,u)$ is constant on $\mathcal S$, then $\tau_{\mathcal S}^w$ equals that shared subgroup effect. However, no homogeneity assumption is needed for the asymptotic normality around the weighted within-subgroup target.

\begin{remark}\label{rmk:splitting}
Proposition~\ref{prop:leaf_oracle} considers the estimation and inference within a fixed or honest subgroup, following the sample-splitting principle used by honest causal trees and forests \citep{athey2016recursive,wager2018estimation}. It does not prove that a recursively learned forest recovers a unique oracle partition. Exact partition recovery generally requires additional separation and structural conditions, such as minimum jumps between adjacent regions \citep{madridpadilla2021lattice}.
\end{remark}

\section{Simulation}\label{sec:simulations}

\subsection{Methods for Comparison}

We distinguish a matched ablation from external competitors. This distinction is important because the methods differ in software, nuisance estimators, and forest engines.

\textbf{Matched ablations and oracle references.}
In the uncensored setting, \emph{Uncensored Baseline} applies \texttt{CausalForestDML} \citep{battocchi2019econml} directly to $[X,W,Z]$, while \emph{Uncensored NC-CSF} adds the bridge-adjusted residualization. In the censored setting, \emph{Censored Baseline} is NC-CSF with the bridge layer removed: it retains the same preprocessing, censoring-robust residual construction, event and censoring nuisances, final features, train/test splits, and forest family. It is therefore the comparison that isolates the contribution of the negative-control bridge. The oracle variants have access to $U$ and serve only as synthetic upper-bound references.

\textbf{External survival and HTE learners.}
We additionally evaluate the R implementation of causal survival forests in \texttt{grf} \citep{cui2023estimating}, random survival forests used as T- and S-learners \citep{ishwaran2008random}, and a discrete-time survival BART S-learner \citep{sparapani2016bart}. Every external learner receives $[X,W,Z]$ as ordinary pre-treatment covariates; NC-CSF differs by assigning $W$ and $Z$ their negative-control roles. Thus no external method is disadvantaged by lack of feature access. The no-negative-control arm in our simulation is the matched Censored Baseline, not the \texttt{grf} implementation. We also report population-level AIPW and TMLE comparisons in Appendix~\ref{sec:marginal_baselines}, while noting that they target a marginal ATE rather than an individual CATE.

\textbf{Implementation and metrics.}
NC-CSF uses logistic regression for the treatment bridge, ExtraTrees \citep{geurts2006extratrees} for outcome-bridge regressions, Cox models \citep{cox1972regression} for censored nuisances, and EconML's \texttt{CausalForestDML} for the final forest. All benchmark methods use identical training/test splits, horizons, and target values. We summarize RMSE and MAE over the common $12\times22=264$ censored configurations; pseudocode for the matched estimators appears in Appendix~\ref{alg:appendix_uncensored_baseline}--\ref{alg:appendix_oracle_comparator}.

\begin{table}[H]
\captionsetup{labelfont={color=black},textfont={color=black}}
\centering
\small
\setlength{\tabcolsep}{4.5pt}
\renewcommand{\arraystretch}{1.0}
\caption{Average error over the $12\times22=264$ censored benchmark configurations. External learners receive $[X,W,Z]$ as ordinary covariates; lower is better.}
\label{tab:censored-summary}
\begin{tabular}{lcccc}
\toprule
\textbf{Model}
& \multicolumn{2}{c}{\textbf{RMST}}
& \multicolumn{2}{c}{\textbf{Survival Probability}} \\
\cmidrule(lr){2-3}\cmidrule(lr){4-5}
& \textbf{RMSE} & \textbf{MAE} & \textbf{RMSE} & \textbf{MAE} \\
\midrule
NC-CSF Oracle     & 0.0516 & 0.0380 & 0.0361 & 0.0275 \\
NC-CSF            & 0.0994 & 0.0768 & 0.0613 & 0.0491 \\
\texttt{grf} R-CSF & 0.1211 & 0.0971 & 0.0686 & 0.0581 \\
Censored Baseline & 0.1662 & 0.1247 & 0.0788 & 0.0642 \\
RSF T-learner     & 0.1773 & 0.1416 & 0.0796 & 0.0721 \\
RSF S-learner     & 0.1637 & 0.1274 & 0.0889 & 0.0638 \\
Survival BART     & 0.1762 & 0.1318 & 0.0699 & 0.0657 \\
\bottomrule
\end{tabular}
\end{table}

NC-CSF has the lowest error among all feasible methods. Relative to the matched Censored Baseline, it reduces RMST RMSE by $40.2\%$ and survival-probability RMSE by $22.2\%$; relative to the strongest external tree-ensemble comparator, \texttt{grf} R-CSF, the reductions are $17.9\%$ and $10.6\%$. The component ablation in Appendix~\ref{tab:appendix_final_feature_ablation} shows that the bridge adjustment is the largest contributor: adding the bridge while holding the forest and censoring pipeline fixed reduces RMST RMSE from $0.1662$ to $0.1081$ and survival-probability RMSE from $0.0788$ to $0.0675$; nuisance-summary augmentation yields the remaining improvement to $0.0994$ and $0.0613$.

\subsection{Data-Generating Mechanism}\label{sec:data.generating.short}

We generate $n$ i.i.d.\ samples with covariates $X_i$ and latent confounder $U_i$, where $X_i \sim \mathcal{N}(0,I_p)$ and $U_i \sim \mathcal{N}(0,1)$. The negative-control proxies $Z_i$ and $W_i$ are constructed as noisy linear functions of $(X_i,U_i)$, so that both carry information about the latent confounder. Treatment assignment follows a Bernoulli model with logistic mean $\mathbb{P}(A_i=1\mid X_i,U_i)=\sigma(b_0+\gamma_u U_i+\alpha^\top X_i)$. Potential event times are generated from a Weibull proportional hazards model, with $T_i(a)\mid X_i,U_i \sim \mathrm{Weibull}(k,\lambda \exp(-\eta_i(a)/k))$ and linear predictor $\eta_i(a)=\beta_t^\top X_i+\beta_u U_i+\tau a$. Because the proxy construction is linear-Gaussian, the conditional law of $U_i$ given $(X_i,Z_i,W_i)$ is available in closed form, which yields an explicit expression for the observed-data conditional treatment effect $\mathbb{E}[T_i(1)-T_i(0)\mid X_i,Z_i,W_i]$ after integrating out $U_i$. 
Appendix~\ref{sec:proxy_stress} further replaces the proxy mechanism with nonlinear, heavy-tailed, noninjective, and binary-coarsened alternatives.

We evaluate performance across regimes varying proxy quality, confounding strength, treatment effect size, and covariate dimensionality, under both linear and nonlinear settings, which results in $12\times22$ configurations (Table~\ref{tab:synthetic-configurations} in Appendix).

\subsection{Results with Uncensored Data}

We evaluate estimation accuracy using root mean squared error, 
$\mathrm{RMSE}=\sqrt{n^{-1}\sum_{i=1}^n(\widehat{\tau}(X_i)-\tau_i)^2}$, 
and mean absolute error, 
$\mathrm{MAE}=n^{-1}\sum_{i=1}^n|\widehat{\tau}(X_i)-\tau_i|$.
NC-CSF closely matches the oracle across all $n$ and covariate dimensions, while the baseline exhibits higher error, particularly in small-sample and high-dimensional regimes. Table~\ref{tab:uncensored-summary} in Appendix further shows that NC-CSF improves over the Uncensored Baseline in both RMSE and MAE and remains closer to the oracle; additional results appear in Appendix Figure~\ref{fig:appendix_mae_uncensored}, Figure~\ref{fig:appendix_rmse_uncensored}, and Table~\ref{tab:mc_uncertainty}.

\begin{figure}[tbp]
\centering

\begin{minipage}[t]{0.45\linewidth}
    \raggedright
    \includegraphics[width=\linewidth]{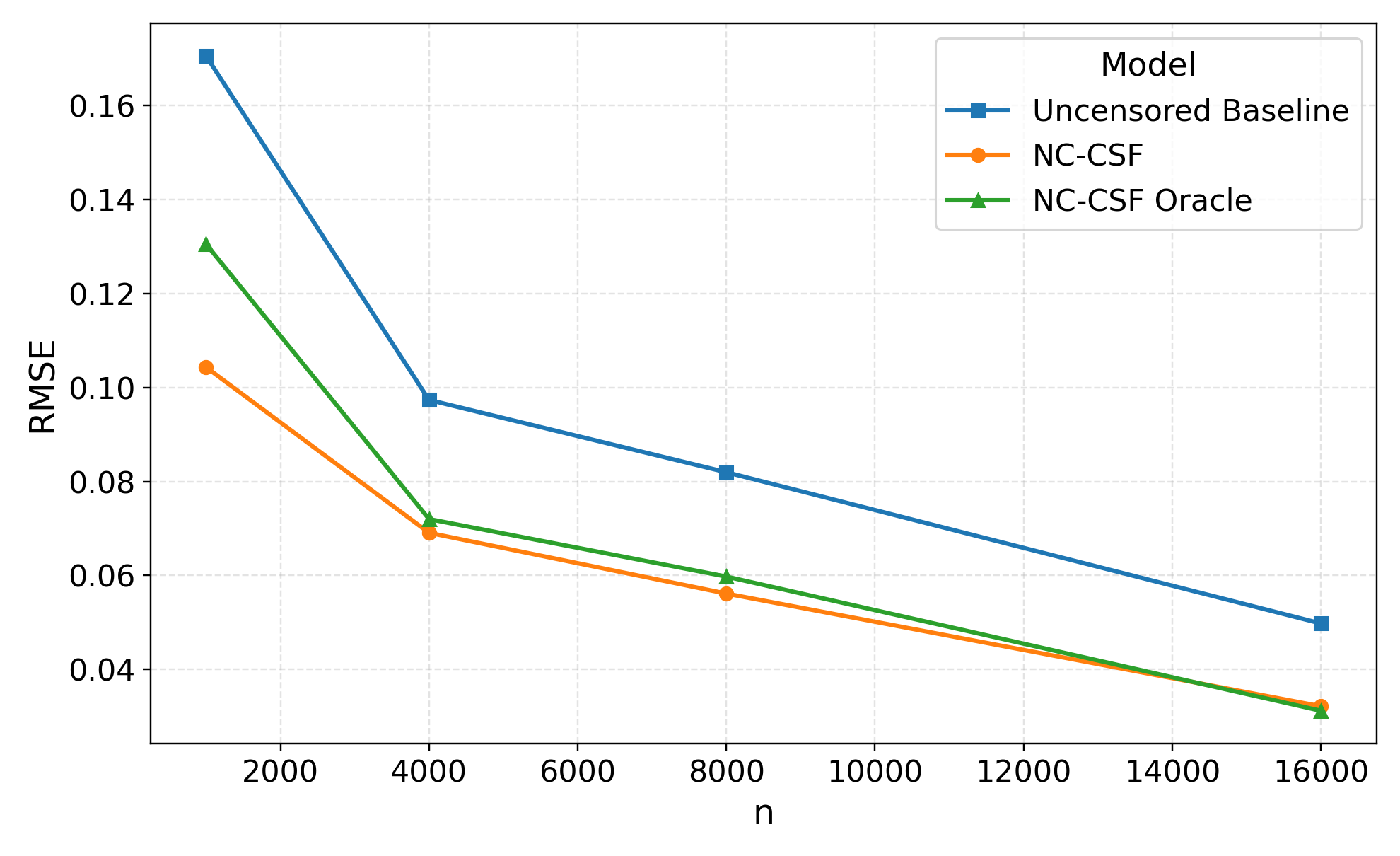}
\end{minipage}
\hfill
\begin{minipage}[t]{0.45\linewidth}
    \raggedright
    \includegraphics[width=\linewidth]{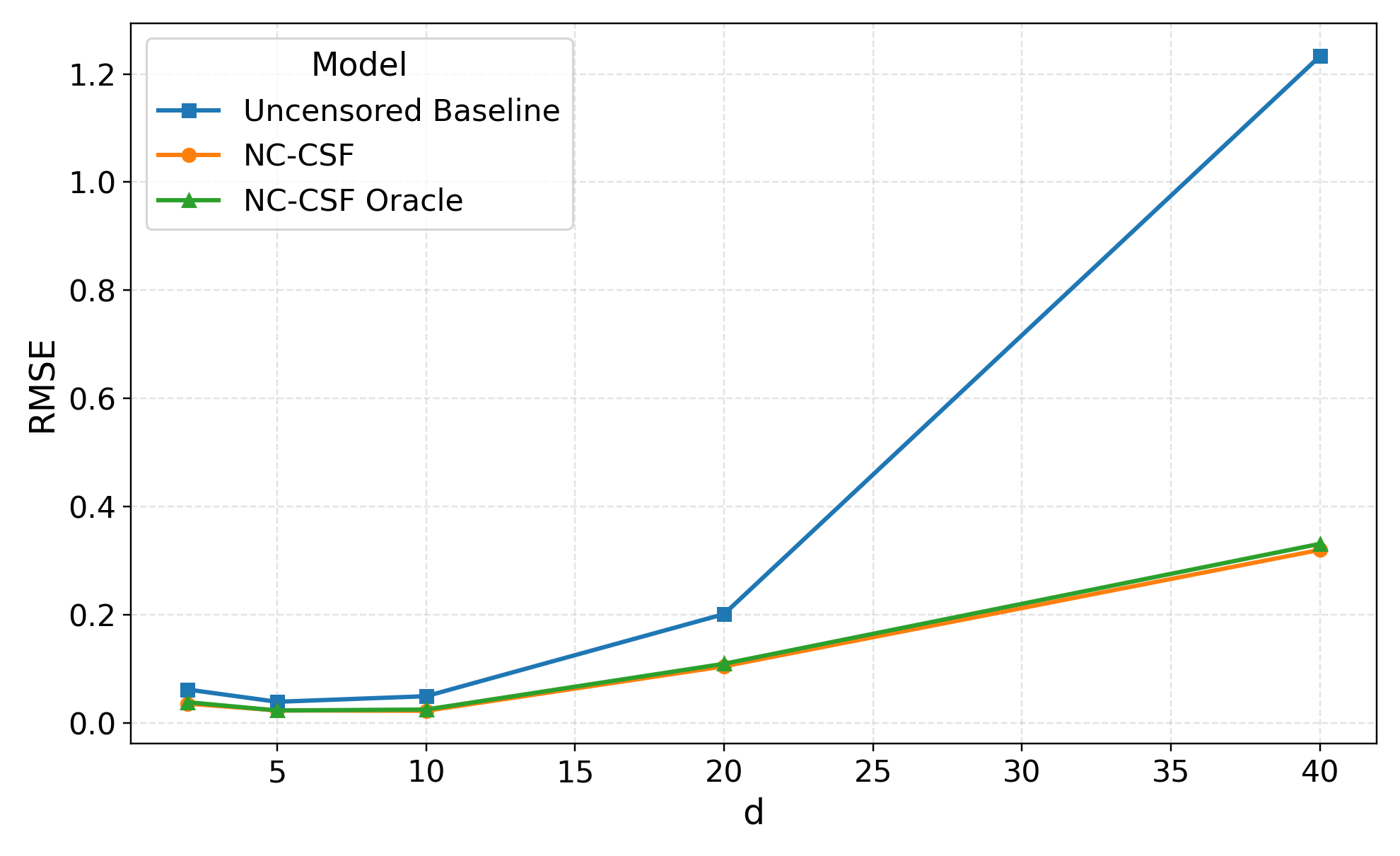}
\end{minipage}
\caption{
\textbf{Left:} RMSE vs.\ sample size $n$ under Case 10 (nonlinear treatment and outcome, informative proxies, weak confounding). 
\textbf{Right:} RMSE vs.\ covariate dimension $d$ under Case 12 (nonlinear, treatment/outcome, non-informative proxies, weak confounding). 
The associated Monte Carlo standard errors and confidence intervals are reported in Appendix Table~\ref{tab:mc_uncertainty}.
}
\label{fig:uncensored-main}
\vspace{-0.5cm}
\end{figure}

\subsection{Results with Censored Data}

Following the uncensored experiments, we evaluate performance under right censoring. We fix the censoring rate at 35\% across all configurations and calibrate the Weibull scale parameter accordingly to achieve the desired censoring level. Table~\ref{tab:censored-summary} summarizes results for both RMST (Restricted Mean Survival Time) and survival probability. These results indicate that the benefits of proxy-based adjustment become more pronounced in the presence of censoring, where standard approaches are more susceptible to bias and instability.

\begin{figure}[tbp]
\centering
\includegraphics[width=0.8\linewidth]{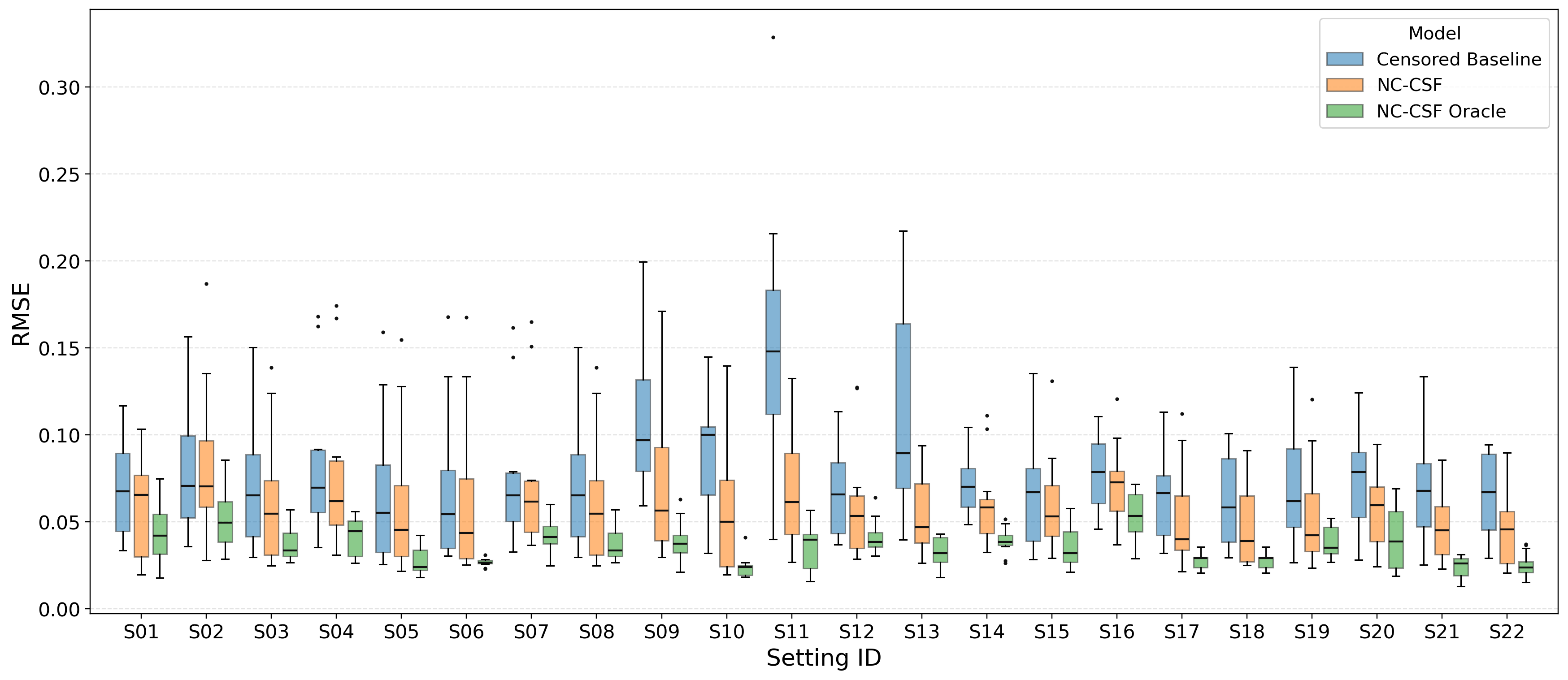}
\caption{RMSE for survival probability under censoring across settings S01--S22. Boxes show variation across the 12 base cases. NC-CSF achieves lower and more stable RMSE than the baseline and remains close to the oracle, with the largest gains in harder settings (S10--S11, S17--S22).}
\label{fig:censored-rmse-main}
\end{figure}

To assess heterogeneity, we compute subgroup bias on an independent test set of size 20{,}000 for each setting--case pair. Subgroups are predefined using quantile bins of selected covariates, yielding 20 groups per configuration: three bins for $X_0$, three for $X_1$, five for the true CATE, and nines for the $X_0 \times X_1$ interaction bins.

After training on an independent dataset, we predict CATEs on the test set and compute, for each subgroup $g$, the bias $\mathrm{Bias}(g)=|g|^{-1}\sum_{i\in g}\{\widehat{\tau}(X_i)-\tau_i\}$, excluding groups below a minimum size. These plots highlight local CATE accuracy beyond aggregate ATE performance. Figure~\ref{fig:subgroup_bias} shows one example; additional results appear in Appendix~\ref{fig:appendix_mae_rmst_censored}--\ref{fig:appendix_subbias_survprob_x0x1}.

\begin{figure}[tbp]
\centering
\includegraphics[width=0.8\linewidth]{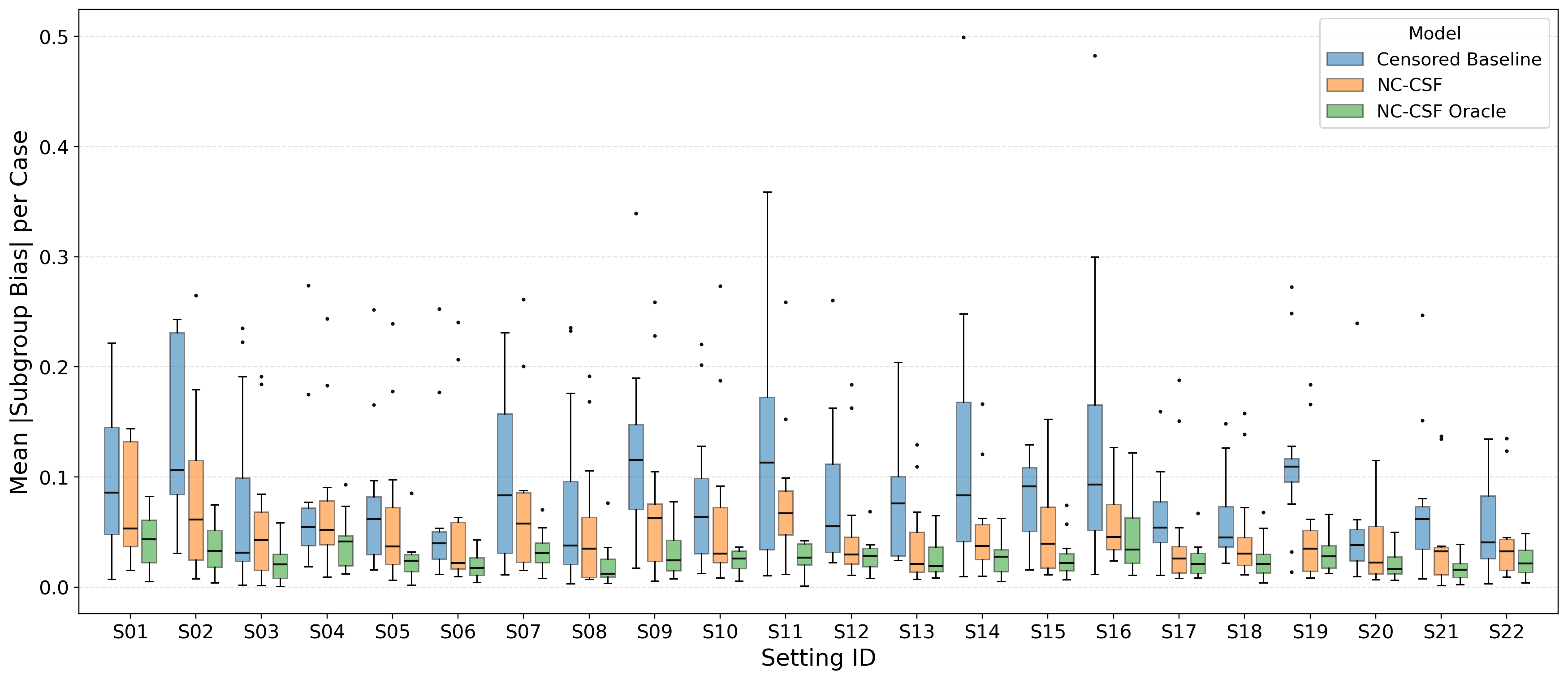}
\caption{Subgroup bias for RMST estimation across settings S01--S22 under censoring, using $X_1$-based subgroups. Each box summarizes the distribution of subgroup-level biases across the 12 base cases within a setting. NC-CSF consistently reduces subgroup bias relative to the baseline and remains closer to the oracle, indicating improved heterogeneous treatment effect estimation.}
\label{fig:subgroup_bias}
\end{figure}

\section{Real Data}\label{sec:real.data}

\subsection{ACTG175 Trial}
\label{sec:actg175}

We apply NC-CSF to the ACTG175 HIV trial \cite{hammer1996trial}, using the same ddI versus ZDV+ddI subset and RMST target as \cite{cui2023estimating}. The analysis contains 1{,}083 subjects, with \(A=0\) denoting ddI monotherapy (561 controls) and \(A=1\) denoting ZDV+ddI (522 treated). Following \cite{cui2023estimating}, we analyze this two-arm subset as observational and use RMST at \(h=1{,}000\) days, which avoids the late follow-up region dominated by censoring. For the negative-control specification, we set \(W=(\texttt{cd420},\texttt{cd820})\), \(Z=(\texttt{race},\texttt{preanti})\), and use 12 baseline variables as \(X\). Since the negative-control restrictions are not directly testable, these proxy choices should be interpreted as exploratory modeling choices. For comparison, we re-ran the causal survival forest implementation in the R \texttt{grf} package (R-CSF) on the same subset, covariates, and horizon, rather than reusing the original published figure.

\begin{figure}[!htbp]
\centering
\begin{minipage}[t]{0.48\linewidth}
    \centering
    \includegraphics[width=\linewidth]{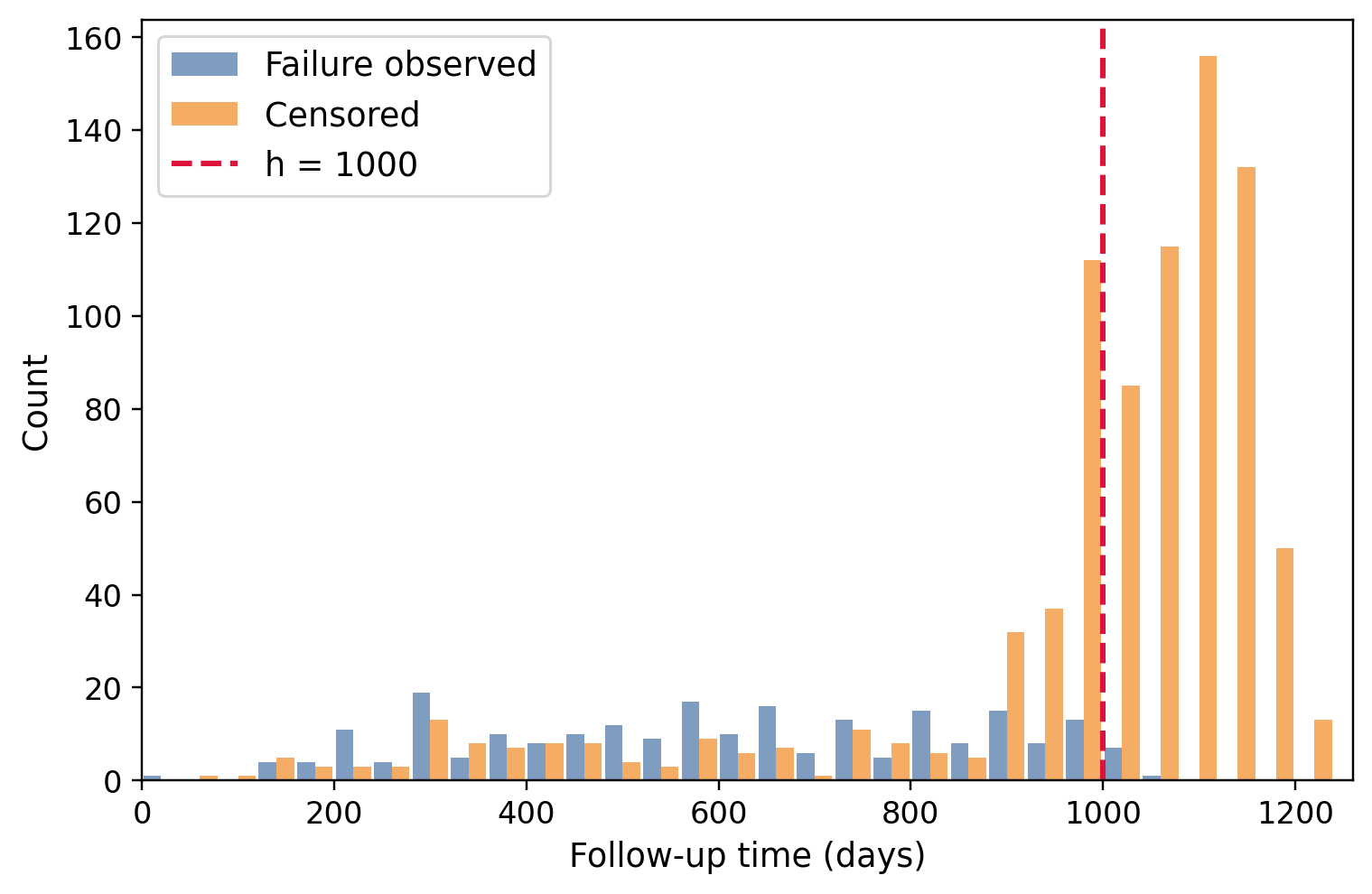}
\end{minipage}
\hfill
\begin{minipage}[t]{0.48\linewidth}
    \centering
        \includegraphics[width=\linewidth]{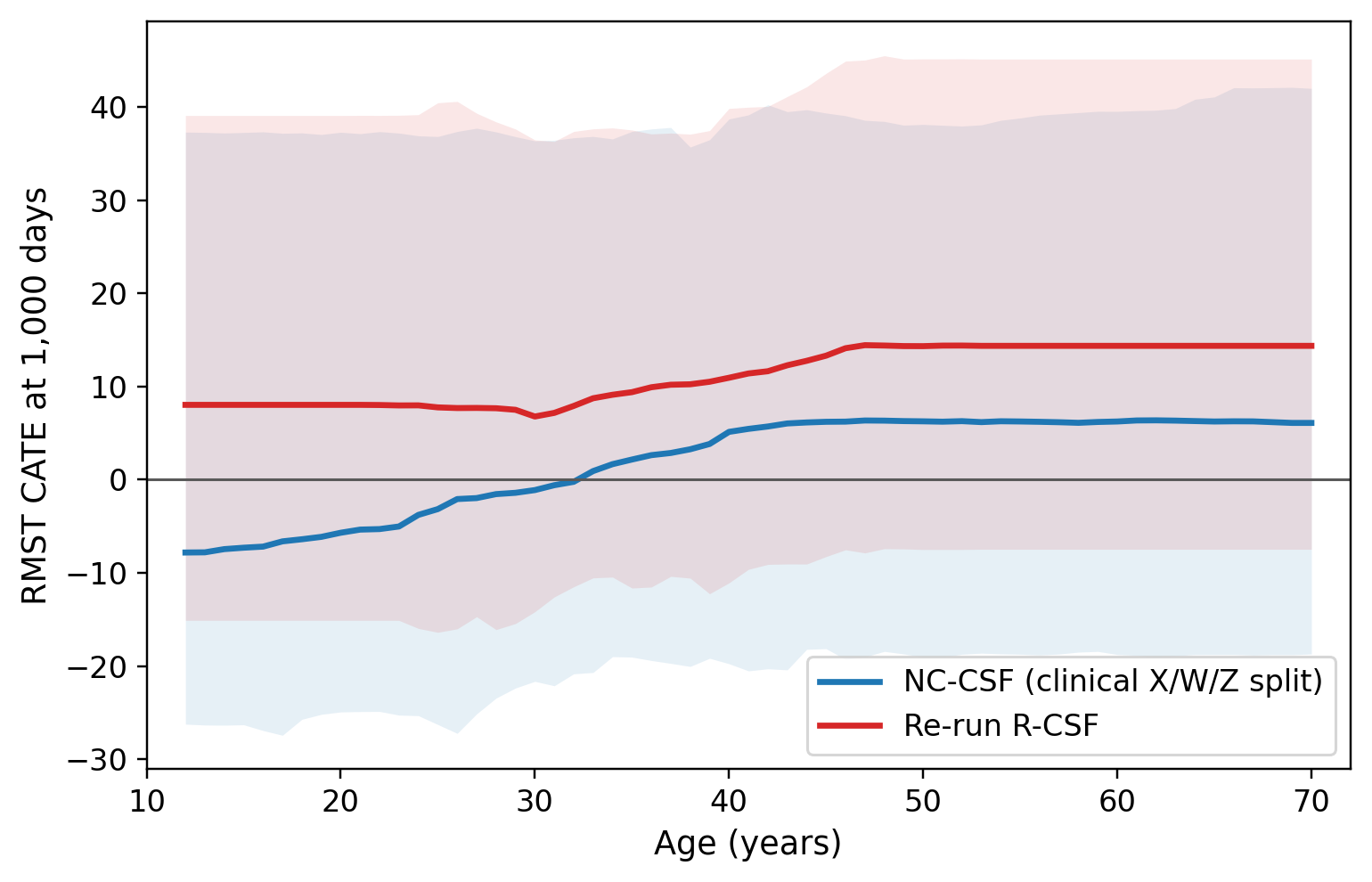}
\end{minipage}
\caption{\textbf{Left:} ACTG175 follow-up times with RMST horizon \(h=1{,}000\). \textbf{Right:} age-specific RMST effect estimates from NC-CSF and our re-run R-CSF baseline, holding other covariates fixed. Shaded bands are pointwise percentile-bootstrap intervals.}
\label{fig:actg175_followup_age}
\vspace{-0.25cm}
\end{figure}

The right panel of Figure~\ref{fig:actg175_followup_age} shows the main fitted pattern. NC-CSF estimates a weaker, slightly negative fitted RMST effect of ZDV+ddI for younger reference profiles: \(-5.69\) days at age 20 and \(-1.11\) days at age 30. The estimate becomes positive by middle age, reaching \(5.13\) days at age 40 and peaking at \(6.37\) days around age 62. By contrast, the re-run R-CSF curve is positive across the displayed age range and peaks at \(14.43\) days around age 47. Thus NC-CSF suggests a more age-dependent and more conservative incremental benefit from adding ZDV to ddI.
We interpret the negative estimates for younger profiles as unfavorable incremental RMST benefit from ZDV+ddI relative to ddI alone, not as direct evidence of ZDV toxicity. This direction is consistent with the ACTG175 individualized-treatment-rule analysis of \cite{lu2013variable}, where age entered positively in a rule favoring ZDV+ddI over ddI, although that study used CD4 count at $20\pm5$ weeks rather than RMST. The pointwise intervals are wide because 1{,}083 subjects from a trial powered for a marginal contrast provide limited information for an age-resolved 1{,}000-day RMST effect; the re-run R-CSF intervals are similarly wide. We therefore use the paired difference between curves fitted to the same subset and seeds as an exploratory comparison, not as a treatment rule. All three proxy variants in Appendix~\ref{sec:appendix_actg175_splits} preserve the pattern of smaller fitted benefit at younger ages.

\subsection{Right Heart Catheterization}
\label{sec:rhc}

As a second real-data check, we use the right heart catheterization (RHC) data from the SUPPORT study of seriously ill hospitalized adults \cite{connors1996rhc}. The dataset records whether a patient received RHC within 24 hours of admission, 30-day survival, and 72 baseline variables covering demographics, comorbidities, admission diagnoses, and physiological and laboratory measurements. The treatment indicates RHC within 24 hours of admission and the outcome is binary 30-day survival. Following the negative-control analysis of \cite{kallus2021negative}, we set \(Z=(\texttt{pafi1},\texttt{paco21})\), \(W=(\texttt{ph1},\texttt{hema1})\), and use the remaining 68 encoded baseline variables as \(X\). The sample contains 5{,}735 patients, with 2{,}184 treated and 3{,}551 controls.

\begin{table}[H]
\vspace{-0.25cm}
\centering
\caption{RHC average treatment-effect estimates for binary 30-day survival. Negative values indicate lower 30-day survival under RHC. The literature estimates are reported in \cite{kallus2021negative}. Our estimate is the average of the uncensored final model's subject-level effect predictions.}
\label{tab:rhc_results}
\setlength{\tabcolsep}{5pt}
\vspace{0.25cm}
\renewcommand{\arraystretch}{1.1}
\begin{tabular}{lccc}
\toprule
\textbf{Estimator} & \textbf{ATE estimate} & \textbf{SE} & \textbf{95\% CI} \\
\midrule
Final Non-Censored Model & -0.0433 & 0.0110 & [-0.0649, -0.0217] \\
\midrule
Kallus et al.\ DR & -0.0219 & 0.0051 & [-0.0319, -0.0119] \\
Kallus et al.\ DR(sta) & -0.0271 & 0.0043 & [-0.0355, -0.0188] \\
CPSMT20 linear bridge & -0.0607 & 0.0055 & [-0.0714, -0.0500] \\
VV15 no-unmeasured-confounding & -0.0612 & 0.0141 & [-0.0889, -0.0335] \\
\bottomrule
\end{tabular}
\end{table}
\vspace{-0.5cm}

All estimates in Table~\ref{tab:rhc_results} imply lower survival under RHC. Our estimate, \(-0.0433\), lies between the flexible minimax bridge estimates of \cite{kallus2021negative} and the more negative linear/no-hidden-confounding estimates. Since RHC is observational and lacks ground truth, we interpret this result as a consistency check rather than a definitive accuracy comparison.

\subsection{MIMIC-IV pharmacologic positive and inert controls}
\label{sec:mimic}

We analyze MIMIC-IV \citep{johnson2023mimiciv} to study whether combining tacrolimus with a CYP3A inhibitor increases the risk of acute kidney injury (AKI) within seven days. This interaction is well documented \citep{fda2023prograf,birdwell2015cpic,naesens2009calcineurin}. We define survival as remaining free of AKI, so negative estimates indicate harm. Pretreatment laboratory values are used as proxies for underlying illness severity, while administrative variables capture differences in prescribing patterns. The analysis includes \(2{,}257\) records in the positive-control group with CYP3A inhibition and \(1{,}398\) records in the inert-control group, where the added drug has no CYP3A activity.

\begin{table}[H]
\captionsetup{labelfont={color=black},textfont={color=black}}
\centering
\setlength{\tabcolsep}{4pt}
\renewcommand{\arraystretch}{1.0}
\caption{MIMIC-IV average treatment-effect estimates for 7-day AKI-free survival and restricted mean AKI-free time in the positive-control group comparing tacrolimus with versus without a concomitant CYP3A inhibitor. Negative values indicate lower AKI-free survival or fewer AKI-free days under concomitant CYP3A inhibition. The reported estimates are averages of the corresponding methods' subject-level effect predictions.
}
\label{tab:mimic_primary}
\begin{tabular}{lccccc}
\toprule
& \multicolumn{2}{c}{\textbf{AKI-free survival at $7$ days}} & & \multicolumn{2}{c}{\textbf{AKI-free RMST through $7$ days}} \\
\cmidrule(lr){2-3}\cmidrule(lr){5-6}
\textbf{Method} & \textbf{Estimate} & \textbf{95\% CI} & \textbf{E-value} & \textbf{Estimate} & \textbf{95\% CI} \\
\midrule
RSF-S  & $-0.0033$ & $[-0.0128,+0.0018]$ & 1.15 & $-0.0205$ & $[-0.0773,+0.0017]$ \\
BART   & $-0.0462$ & $[-0.1529,+0.0311]$ & --   & $-0.2365$ & $[-0.7806,+0.1800]$ \\
RSF-T  & $-0.0758$ & $[-0.1245,-0.0401]$ & 2.17 & $-0.4279$ & $[-0.6396,-0.2223]$ \\
CSF    & $-0.1176$ & $[-0.1939,-0.0758]$ & 2.69 & $-0.6232$ & $[-0.9129,-0.4440]$ \\
\textbf{NC-CSF} & $\mathbf{-0.2410}$ & $\mathbf{[-0.2857,-0.1816]}$ & $\mathbf{4.40}$ & $\mathbf{-1.6248}$ & $\mathbf{[-2.1263,-1.2023]}$ \\
\bottomrule
\end{tabular}
\end{table}

\begin{table}[H]
\captionsetup{labelfont={color=black},textfont={color=black}}
\centering
\setlength{\tabcolsep}{4pt}
\renewcommand{\arraystretch}{1.0}
\caption{Results for the inert group: the added drug has no CYP3A activity. See the caption of Table~\ref{tab:mimic_primary} for details.}
\label{tab:mimic_inert}
\begin{tabular}{lcccc}
\toprule
& \multicolumn{2}{c}{\textbf{AKI-free survival at $7$ days}} & \multicolumn{2}{c}{\textbf{AKI-free RMST through $7$ days}} \\
\cmidrule(lr){2-3}\cmidrule(lr){4-5}
\textbf{Method} & \textbf{Estimate} & \textbf{95\% CI} & \textbf{Estimate} & \textbf{95\% CI} \\
\midrule
RSF-S  & $+0.0053$ & $[+0.0001,+0.0189]$ & $+0.0224$ & $[-0.0051,+0.0986]$ \\
BART   & $+0.0092$ & $[-0.0694,+0.0965]$ & $+0.0473$ & $[-0.3791,+0.5054]$ \\
RSF-T  & $+0.0147$ & $[-0.0187,+0.0492]$ & $-0.0221$ & $[-0.1486,+0.0990]$ \\
CSF    & $+0.0013$ & $[-0.0414,+0.0466]$ & $+0.0056$ & $[-0.1574,+0.2161]$ \\
\textbf{NC-CSF} & $\mathbf{-0.0323}$ & $\mathbf{[-0.1403,+0.0171]}$ & $\mathbf{+0.0265}$ & $\mathbf{[-0.2453,+0.4034]}$ \\
\bottomrule
\end{tabular}
\end{table}

In the positive-control group, NC-CSF gives the strongest evidence of harm: a \(0.2410\) reduction in seven-day AKI-free survival and \(1.6248\) fewer AKI-free days, more than twice the corresponding CSF estimates. RSF-S and BART do not clearly detect harm, while RSF-T and CSF give smaller effects. In the inert group, NC-CSF finds no clear effect on either outcome, supporting its ability to detect the known interaction without producing a signal under the inert substitution.

Among the four prespecified comorbidity subgroups, only prior liver disease has a direct biological link through hepatic CYP3A metabolism. NC-CSF estimates a larger harm for patients with prior liver disease than for those without it (\(-0.3532\) versus \(-0.2131\)), whereas the other three subgroup comparisons show no clear differences. Appendix~\ref{sec:mimic_subgroups} reports the liver-disease results.

\section{Discussion}\label{sec:discussion}

In this paper, we provide a practical approach to survival analysis with flexible heterogeneous treatment-effect modeling in the presence of unobserved confounding. Empirically, the proposed method performs favorably in extensive experiments and yields interesting findings in real data applications. We also provide an off-the-shelf Python package to improve practical usability.

We discuss limitations. First, $\tau^w(x)$ is the general estimand and need not equal the conventional CATE when latent treatment-effect heterogeneity covaries with bridge-residualized treatment variation. Second, negative-control validity is application-specific and cannot be verified from the observed law alone. Third, completeness does not resolve finite-sample ill-posedness: orthogonality limits downstream sensitivity to sufficiently accurate nuisance estimates but does not create accurate bridge estimates when the moment operator is poorly conditioned. Regularized bridge learning and diagnostics for proxy strength therefore remain practically relevant.

Several directions remain for future work. 
From the causal perspective, since the resulting target is weighted in general, it would be useful to further study the interpretation and behavior of different estimands.
On the theoretical side, one direction is to analyze the recursive splitting behavior of both the oracle forest and its feasible counterpart with estimated nuisance functions.

\newcommand{\markrevisionbib}[1]{\expandafter\def\csname revisionbib@#1\endcsname{}}
\markrevisionbib{tchetgen2020proximal}
\markrevisionbib{lipsitch2010negative}
\markrevisionbib{shi2020selective}
\markrevisionbib{ishwaran2008random}
\markrevisionbib{sparapani2016bart}
\markrevisionbib{athey2016recursive}
\markrevisionbib{madridpadilla2021lattice}
\markrevisionbib{johnson2023mimiciv}
\markrevisionbib{fda2023prograf}
\markrevisionbib{birdwell2015cpic}
\markrevisionbib{naesens2009calcineurin}
\let\originalbibitem\bibitem
\renewcommand{\bibitem}[2][]{%
  \ifcsname revisionbib@#2\endcsname\color{black}\else\color{black}\fi
  \if\relax\detokenize{#1}\relax
    \originalbibitem{#2}%
  \else
    \originalbibitem[#1]{#2}%
  \fi
}

\bibliographystyle{plain}
\bibliography{ref}

\appendix

\section{Additional Algorithms}

This section provides pseudocode for the estimators used in the simulation comparisons. Algorithm~\ref{alg:uncensored_final_training} describes the uncensored NC-CSF pipeline, while Algorithms~\ref{alg:appendix_uncensored_baseline} and~\ref{alg:appendix_censored_baseline} give the corresponding observed-data baselines, which use $(X,W,Z)$ directly and remove the negative-control bridge layer. For the censored setting, the matched baseline keeps the same preprocessing, survival nuisance estimation, censoring correction, residual construction, and final causal-forest step as Censored NC-CSF, isolating the effect of proxy-based bridge adjustment. Algorithm~\ref{alg:censored_final_training_detailed} extends Algorithm~\ref{alg:censored_final_training} with implementation details, and Algorithm~\ref{alg:appendix_oracle_comparator} gives the synthetic-only oracle comparator, which has access to $U$ and DGP-specific nuisance functions and serves only as an upper-bound reference. 
We additionally compare against \texttt{grf::causal\_survival\_forest}, RSF T-/S-learners, and survival BART in the common simulation benchmark; the \texttt{grf} implementation is also used in the ACTG175 analysis.

\begin{algorithm}
\caption{Censored NC-CSF Training}
\label{alg:censored_final_training_detailed}
\begin{algorithmic}[1]
\Require Data $\{(X_i,W_i,Z_i,A_i,Y_i,\Delta_i)\}_{i=1}^n$; target type; horizon $h$; number of folds $K$; clipping levels $\epsilon_q,\alpha_{\rm br},\alpha_{\rm res}$
\Ensure Fitted effect predictor $\widehat\tau(\cdot)$

\State Partition $\{1,\ldots,n\}$ into cross-fitting folds $\mathcal I_1,\ldots,\mathcal I_K$.

\For{$k=1,\ldots,K$}
    \State Let $\mathcal I_{-k}=\{1,\ldots,n\}\setminus\mathcal I_k$ and set $(Y_i^\star,\Delta_i^\star)=(Y_i^{(h)},\Delta_i^{(h)})$ for RMST, and $(Y_i^\star,\Delta_i^\star)=(Y_i,\Delta_i)$ for survival probability.

    \State Fit a censoring Cox model on $\mathcal I_{-k}$ with covariates $(X,W,Z,A)$,
    times $T^\star$, and censoring indicators $1-\Delta^\star$, yielding the censoring survival
    $\widehat S_C^{(-k)}$ and cumulative censoring hazard increments
    $d\widehat\Lambda_C^{(-k)}$.

    \State Construct target-specific bridge pseudo-outcomes $\widetilde Y_i^{\rm bridge}$ on $\mathcal I_{-k}$,
    and clip them to empirical $(\alpha_{\rm br},1-\alpha_{\rm br})$ quantiles.

    \State Fit the treatment bridge $\widehat q^{(-k)}$ using logistic regression on $\mathcal I_{-k}$,
    with labels $A$ and bridge covariates $(X,Z)$.

    \For{$a\in\{0,1\}$}
        \State Let $\mathcal I_{-k,a}=\{i\in\mathcal I_{-k}:A_i=a\}$.

        \State Fit the outcome bridge $\widehat f_a^{(-k)}$ using a random forest regressor on $\mathcal I_{-k,a}$,
        with labels $\widetilde Y^{\rm bridge}$ and bridge covariates $(X,W)$.

        \State Fit event-survival model $\widehat S_a^{(-k)}$ using a Cox model on $\mathcal I_{-k,a}$,
        with covariates $(X,W,Z)$, times $T^\star$, and event indicators $\Delta^\star$.
    \EndFor

    \State Construct continuation functions $\widehat\rho_a^{(-k)}(t;h)$ from $\widehat S_a^{(-k)}$, $a\in\{0,1\}$:
    \Statex \hspace{\algorithmicindent}
    $\widehat\rho_a^{(-k)}(t;h)=
    t+\int_t^h
    \frac{\widehat S_a^{(-k)}(u\mid X,W,Z)}
         {\widehat S_a^{(-k)}(t\mid X,W,Z)}\,du$
    for RMST, and
    $\widehat\rho_a^{(-k)}(t;h)=
    \frac{\widehat S_a^{(-k)}(h\mid X,W,Z)}
         {\widehat S_a^{(-k)}(t\mid X,W,Z)}$
    for survival probability.

    \For{$i\in\mathcal I_k$}
        \State Predict
        $\widehat q_i=\textsc{Clip}\{\widehat q^{(-k)}(X_i,Z_i),\epsilon_q,1-\epsilon_q\}$,
        $\widehat f_{a,i}=\widehat f_a^{(-k)}(X_i,W_i)$ for $a\in\{0,1\}$.

        \State Form
        $\widehat m_i=\widehat q_i\widehat f_{1,i}+(1-\widehat q_i)\widehat f_{0,i}$,
        and $A_i^{\rm res}=A_i-\widehat q_i$.

        \State Compute $\widetilde Y_i^{\rm res}$ from Eq.~\eqref{eq:survival_residual}
        using $\widehat m_i$, $\widehat\rho_{A_i}^{(-k)}$,
        $\widehat S_C^{(-k)}$, and $d\widehat\Lambda_C^{(-k)}$.

        \State Form
        $X_i^{\rm final}
        =(X_i,W_i,Z_i,\widehat q_i,\widehat f_{1,i},\widehat f_{0,i},
        \widehat m_i,\widehat r_{1,i},\widehat r_{0,i},\widehat r_{1,i}-\widehat r_{0,i})$,
        where $\widehat r_{a,i}$ is the target-specific event-survival summary.
    \EndFor

    \State Clip $\{\widetilde Y_i^{\rm res}:i\in\mathcal I_k\}$ to empirical
    $(\alpha_{\rm res},1-\alpha_{\rm res})$ quantiles, yielding $Y_i^{\rm res}$.
\EndFor

\State Fit the final causal forest $\widehat{\mathcal F}$ on the cross-fitted training triples
$\{(X_i^{\rm final},A_i^{\rm res},Y_i^{\rm res})\}_{i=1}^n$.

\State Refit only the nuisance models on the full training sample. For a new point $(x,w,z)$,
construct $X_{\rm new}^{\rm final}$ using the same formula above and return
$\widehat\tau(x)=\widehat{\mathcal F}(X_{\rm new}^{\rm final})$.

\State \Return $\widehat\tau(\cdot)$.
\end{algorithmic}
\end{algorithm}

\begin{algorithm}
\caption{Uncensored NC-CSF Training}
\label{alg:uncensored_final_training}
\begin{algorithmic}[1]
\Require Data $\{(X_i,W_i,Z_i,A_i,Y_i^{\rm unc})\}_{i=1}^n$; number of folds $K$; clipping levels $\epsilon_q,\alpha_y,\alpha_{\rm res}$
\Ensure Fitted effect predictor $\widehat\tau(\cdot)$

\State Partition $\{1,\ldots,n\}$ into cross-fitting folds $\mathcal I_1,\ldots,\mathcal I_K$.

\For{$k=1,\ldots,K$}
    \State Let $\mathcal I_{-k}=\{1,\ldots,n\}\setminus\mathcal I_k$.

    \State Clip $\{Y_i^{\rm unc}:i\in\mathcal I_{-k}\}$ to empirical
    $(\alpha_y,1-\alpha_y)$ quantiles, yielding $Y_i^{\rm unc,clip}$.

    \State Fit the treatment bridge $\widehat q^{(-k)}$ using logistic regression on $\mathcal I_{-k}$,
    with labels $A$ and bridge covariates $(X,Z)$.

    \For{$a\in\{0,1\}$}
        \State Let $\mathcal I_{-k,a}=\{i\in\mathcal I_{-k}:A_i=a\}$.

        \State Fit the outcome bridge $\widehat f_a^{(-k)}$ using a random forest regressor on $\mathcal I_{-k,a}$,
        with labels $Y^{\rm unc,clip}$ and bridge covariates $(X,W)$.
    \EndFor

    \For{$i\in\mathcal I_k$}
        \State Predict
        $\widehat q_i=\textsc{Clip}\{\widehat q^{(-k)}(X_i,Z_i),\epsilon_q,1-\epsilon_q\}$,
        $\widehat f_{a,i}=\widehat f_a^{(-k)}(X_i,W_i)$ for $a\in\{0,1\}$.

        \State Form
        $\widehat m_i=\widehat q_i\widehat f_{1,i}+(1-\widehat q_i)\widehat f_{0,i}$,
        $A_i^{\rm res}=A_i-\widehat q_i$, and
        $\widetilde Y_i^{\rm res}=Y_i^{\rm unc}-\widehat m_i$.

        \State Form
        $X_i^{\rm final}
        =(X_i,W_i,Z_i,\widehat q_i,\widehat f_{1,i},\widehat f_{0,i},\widehat m_i)$.
    \EndFor

    \State Clip $\{\widetilde Y_i^{\rm res}:i\in\mathcal I_k\}$ to empirical
    $(\alpha_{\rm res},1-\alpha_{\rm res})$ quantiles, yielding $Y_i^{\rm res}$.
\EndFor

\State Fit the final causal forest $\widehat{\mathcal F}$ on the cross-fitted training triples
$\{(X_i^{\rm final},A_i^{\rm res},Y_i^{\rm res})\}_{i=1}^n$.

\State Refit only the nuisance models on the full training sample. For a new point $(x,w,z)$,
construct $X_{\rm new}^{\rm final}$ using the same formula above and return
$\widehat\tau(x)=\widehat{\mathcal F}(X_{\rm new}^{\rm final})$.

\State \Return $\widehat\tau(\cdot)$.
\end{algorithmic}
\end{algorithm}

\begin{algorithm}
\small
\caption{Uncensored Baseline}
\label{alg:appendix_uncensored_baseline}
\begin{algorithmic}[1]
\Require Training data $\{(X_i,W_i,Z_i,A_i,Y_i^{\rm unc})\}_{i=1}^n$; test point $(x,w,z)$
\Ensure Baseline CATE estimate $\widehat\tau_{\rm base}^{\rm unc}(x,w,z)$

\State Form raw covariates $X_i^{\rm raw}=(X_i,W_i,Z_i)$ for all training samples.

\State Fit a causal forest using outcomes $Y_i^{\rm unc}$, treatments $A_i$, and covariates $X_i^{\rm raw}$.

\State Form the test covariate vector $x^{\rm raw}=(x,w,z)$.

\State \Return $\widehat\tau_{\rm base}^{\rm unc}(x,w,z)$ by evaluating the fitted causal forest at $x^{\rm raw}$.
\end{algorithmic}
\end{algorithm}

\begin{algorithm}
\small
\caption{Censored Baseline}
\label{alg:appendix_censored_baseline}
\begin{algorithmic}[1]
\Require Training data $\{(X_i,W_i,Z_i,A_i,Y_i,\Delta_i)\}_{i=1}^n$; target type; horizon \(h\) if needed; test point $(x,w,z)$
\Ensure $\widehat\tau_{\mathrm{base}}^{\mathrm{cens}}(x,w,z)$
\State Form the raw observed feature stack
\[
X_i^{\mathrm{raw}}=[X_i,W_i,Z_i].
\]
\State Initialize the censored baseline with the same target, horizon, and final forest family as the proposed censored model.
\For{each nuisance fold}
    \State Fit an observed-data propensity model on \([X,W,Z]\).
    \State Fit treated and control Cox event-survival models using the target-preprocessed time/event outcome.
    \State Fit a Cox censoring model on \((Y_i,1-\Delta_i)\) with features \([X_i,W_i,Z_i,A_i]\).
    \State Construct the residual pair \((Y_i^{\mathrm{res}},A_i^{\mathrm{res}})\) from these observed-data nuisances.
\EndFor
\State Fit the final causal forest on $\{(X_i^{\mathrm{raw}},A_i^{\mathrm{res}},Y_i^{\mathrm{res}})\}_{i=1}^n$.
\State Form the test feature vector \(x^{\mathrm{raw}}=[x,w,z]\).
\State \Return $\widehat\tau_{\mathrm{base}}^{\mathrm{cens}}(x,w,z)=\texttt{effect}(x^{\mathrm{raw}})$.
\end{algorithmic}
\end{algorithm}

\begin{algorithm}
\small
\caption{NC-CSF Oracle Comparator}
\label{alg:appendix_oracle_comparator}
\begin{algorithmic}[1]
\Require Synthetic training data with oracle access to latent confounders $U$ and DGP nuisances; setting type $\in\{\mathrm{uncensored},\mathrm{censored}\}$; test point $(x,u)$
\Ensure Oracle CATE estimate $\widehat\tau_{\rm oracle}(x,u)$

\State Use $(X,U)$ in place of proxy-based bridge inputs and obtain oracle nuisances from the DGP.

\If{uncensored setting}
    \State For each training unit, obtain
    $\widehat q_i^{\rm oracle}=q_0(X_i,U_i)$ and
    $\widehat f_{a,i}^{\rm oracle}=f_a(X_i,U_i)$ for $a\in\{0,1\}$.

    \State Form
    $\widehat m_i^{\rm oracle}
    =\widehat q_i^{\rm oracle}\widehat f_{1,i}^{\rm oracle}
    +(1-\widehat q_i^{\rm oracle})\widehat f_{0,i}^{\rm oracle}$,
    $A_i^{\rm res}=A_i-\widehat q_i^{\rm oracle}$, and
    $Y_i^{\rm res}=Y_i^{\rm unc}-\widehat m_i^{\rm oracle}$,
    with the same clipping rules as Uncensored NC-CSF.

    \State Form $X_i^{\rm final}$ by augmenting $(X_i,U_i)$ with oracle nuisance summaries.

\Else
    \State Set $(Y_i^\star,\Delta_i^\star)=(Y_i^{(h)},\Delta_i^{(h)})$ for RMST, and
    $(Y_i^\star,\Delta_i^\star)=(Y_i,\Delta_i)$ for survival probability.

    \State For each training unit, obtain oracle bridge nuisances
    $\widehat q_i^{\rm oracle}=q_0(X_i,U_i)$ and
    $\widehat f_{a,i}^{\rm oracle}=f_a(X_i,U_i;h)$, and oracle survival nuisances
    $\widehat S_a^{\rm oracle}(t\mid X_i,U_i)$,
    $S_C^{\rm oracle}(t\mid X_i,U_i,A_i)$, and
    $d\Lambda_C^{\rm oracle}(t\mid X_i,U_i,A_i)$.

    \State Form
    $\widehat m_i^{\rm oracle}
    =\widehat q_i^{\rm oracle}\widehat f_{1,i}^{\rm oracle}
    +(1-\widehat q_i^{\rm oracle})\widehat f_{0,i}^{\rm oracle}$
    and $A_i^{\rm res}=A_i-\widehat q_i^{\rm oracle}$.

    \State Compute $Y_i^{\rm res}$ using the oracle version of Eq.~\eqref{eq:survival_residual}.

    \State Form $X_i^{\rm final}$ by augmenting $(X_i,U_i)$ with oracle bridge and survival nuisance summaries.
\EndIf

\State Fit a causal forest to
$\{(X_i^{\rm final},A_i^{\rm res},Y_i^{\rm res})\}_{i=1}^n$.

\State Construct the corresponding oracle final feature vector $X_{\rm new}^{\rm final}$ for test point $(x,u)$.

\State \Return the fitted causal forest prediction at $X_{\rm new}^{\rm final}$.
\end{algorithmic}
\end{algorithm}

\section{Bridge Estimation and Operator Stability}
\label{sec:bridge_stability}

Define the observable conditional-expectation operators
\begin{align*}
(\mathcal T_q q)(x,w)
&=\mathbb{E}[q(Z,X)\mid X=x,W=w],\\
(\mathcal T_{f,a}f_a)(x,z)
&=\mathbb{E}[f_a(W,X)\mid X=x,Z=z,A=a].
\end{align*}
The bridge moment equations set these operators equal to $\mathbb{E}[A\mid X=x,W=w]$ and $\mathbb{E}[G(a;h)\mid X=x,Z=z,A=a]$, respectively. Estimating their inverse solutions is the conditional-moment problem studied in proximal causal learning \citep{kallus2021negative,cui2023semiparametric}. Completeness makes the relevant operators injective on the chosen function classes, which supports identification, but injectivity alone does not bound the inverse. A sequence of bridge candidates can have small observable moment error while remaining far from the target bridge when an operator is nearly singular.

For a finite-dimensional bridge $q_\theta$, let $M(\theta)$ denote the vector of observable moment restrictions. If the model is correctly specified, $M(\theta_0)=0$, and the smallest singular value of $\nabla_\theta M(\theta_0)$ is bounded below, standard M-estimation arguments yield consistency and a regular rate. The logistic treatment-bridge fit used in the implementation should be understood in this working-model sense. Analogously, a nonparametric moment learner requires a stability inequality of the form
\begin{align*}
\|h-h_0\|_{L_2}
\leq
\kappa^{-1}\|\mathcal T(h-h_0)\|_{\mathcal H},
\end{align*}
for a suitable moment norm $\|\cdot\|_{\mathcal H}$. A small $\kappa$ corresponds to a poorly conditioned inverse and amplifies conditional-moment estimation error.

Ridge regularization replaces inversion by a stabilized problem. Schematically, if $\widehat h_\lambda$ minimizes an empirical moment loss plus $\lambda\|h\|^2$, its error separates into a stabilized estimation term and a regularization-bias term. Under an appropriate source condition, moment consistency, and a sequence $\lambda\downarrow0$ chosen slowly enough to control stochastic amplification, both terms can vanish. Proposition~\ref{prop:Neyman.orthogonality} then controls propagation of the remaining bridge error into the effect estimate, provided the rates in Proposition~\ref{prop:leaf_oracle} hold. It does not replace these bridge-estimation conditions.

In practice, informative diagnostics include cross-fitted conditional-moment residuals, sensitivity to regularization and clipping, and stability across scientifically defensible proxy routings. The weak-proxy and proxy-map experiments in Appendix~\ref{sec:proxy_stress} quantify the effect of poor conditioning and information loss on the final estimator.

\section{Proofs}

\begin{proof}[Proof of Proposition~\ref{prop:bridge_observed_moment}]

We first prove \eqref{eq:bridge_observed_prop_combined}. 
For the treatment bridge,
\begin{align*}
\mathbb{E}\!\left[A_i-q(Z_i,X_i)\mid X_i,U_i,W_i\right]
&=
\mathbb{E}\!\left[A_i\mid X_i,U_i,W_i\right]
-
\mathbb{E}\!\left[q(Z_i,X_i)\mid X_i,U_i,W_i\right] \\
&=
\mathbb{P}(A_i=1\mid X_i,U_i)
-
\mathbb{E}\!\left[q(Z_i,X_i)\mid X_i,U_i\right] \\
&= 0,
\end{align*}
where the second equality uses $(A_i,Z_i)\perp W_i\mid X_i,U_i$, and the third follows from \eqref{eq:bridge.function}. Taking conditional expectation with respect to $(X_i,W_i)$ yields
\begin{align*}
\mathbb{E}\!\left[A_i-q(Z_i,X_i)\mid X_i,W_i\right]=0.
\end{align*}

Similarly, for the outcome bridge,
\begin{align*}
&\mathbb{E}\!\left[G_i(a;h)-f_a(W_i,X_i)\mid X_i,U_i,Z_i,A_i=a\right] \\
&\qquad=
\mathbb{E}\!\left[G_i(a;h)\mid X_i,U_i,Z_i,A_i=a\right]
-
\mathbb{E}\!\left[f_a(W_i,X_i)\mid X_i,U_i,Z_i,A_i=a\right] \\
&\qquad=
\mathbb{E}\!\left[G_i(a;h)\mid X_i,U_i\right]
-
\mathbb{E}\!\left[f_a(W_i,X_i)\mid X_i,U_i\right] \\
&\qquad=0,
\end{align*}
where the third equality uses $(A_i,Z_i)\perp (G_i(a;h),W_i)\mid X_i,U_i$, and the last follows from \eqref{eq:bridge.function}. Taking conditional expectation with respect to $(X_i,Z_i,A_i=a)$ yields
\begin{align*}
\mathbb{E}\!\left[G_i(a;h)-f_a(W_i,X_i)\mid X_i,Z_i,A_i=a\right]=0.
\end{align*}

We now prove the converse. Suppose \eqref{eq:bridge_observed_prop_combined} holds. Define
\begin{align*}
r(X_i,U_i)
&=
\mathbb{P}(A_i=1\mid X_i,U_i)-\mathbb{E}\!\left[q(Z_i,X_i)\mid X_i,U_i\right], \\
s(X_i,U_i)
&=
\mathbb{E}\!\left[G_i(a;h)\mid X_i,U_i\right]-\mathbb{E}\!\left[f_a(W_i,X_i)\mid X_i,U_i\right].
\end{align*}
By iterated expectation and $(A_i,Z_i)\perp W_i\mid X_i,U_i$,
\begin{align*}
\mathbb{E}\!\left[r(X_i,U_i)\mid X_i,W_i\right]
&=
\mathbb{E}\!\left[
\mathbb{E}\!\left[A_i-q(Z_i,X_i)\mid X_i,U_i,W_i\right]
\middle| X_i,W_i
\right] \\
&=
\mathbb{E}\!\left[A_i-q(Z_i,X_i)\mid X_i,W_i\right] \\
&=0.
\end{align*}
By the first completeness condition, this implies $r(X_i,U_i)=0$ a.s., that is,
\begin{align*}
\mathbb{E}\!\left[q(Z_i,X_i)\mid X_i,U_i\right]
=
\mathbb{P}(A_i=1\mid X_i,U_i).
\end{align*}

Similarly, using $(A_i,Z_i)\perp (G_i(a;h),W_i)\mid X_i,U_i$,
\begin{align*}
\mathbb{E}\!\left[s(X_i,U_i)\mid X_i,Z_i,A_i=a\right]
&=
\mathbb{E}\!\left[
\mathbb{E}\!\left[G_i(a;h)-f_a(W_i,X_i)\mid X_i,U_i,Z_i,A_i=a\right]
\middle| X_i,Z_i,A_i=a
\right] \\
&=
\mathbb{E}\!\left[G_i(a;h)-f_a(W_i,X_i)\mid X_i,Z_i,A_i=a\right] \\
&=0.
\end{align*}
Similarly for the second completeness condition.
\end{proof}

\begin{proof}[Proof of Proposition~\ref{prop:score}]\label{proof:prop:tau.w}
We first prove the uncensored case, and then note that the censored case follows by the same argument using the censoring-corrected pseudo-outcome as in the proof of Proposition~\ref{prop:Neyman.orthogonality}.

For the uncensored case, we aim to show
\begin{align*}
\mathbb{E}\!\left[
(A_i-q_i)(G_i-m_i-\tau^w(X_i)(A_i-q_i))
\middle| X_i=x
\right]=0.
\end{align*}
In fact, we can decompose $(A_i-q_i)(G_i-m_i)$ as
\begin{align*}
A_i(A_i-q_i)(Y_i(1)-f_{1,i})
+
(1-A_i)(A_i-q_i)(Y_i(0)-f_{0,i})
+
(A_i-q_i)^2(f_{1,i}-f_{0,i}).
\end{align*}
Taking conditional expectation given $(X_i,U_i)$, the first two terms are zero by \eqref{eq:nc_censored} and \eqref{eq:bridge.function} as
\begin{align*}
\mathbb{E}\!\left[A_i(A_i-q_i)(Y_i(1)-f_{1,i})\mid X_i,U_i\right]
&=
\mathbb{E}\!\left[A_i(A_i-q_i)\mid X_i,U_i\right]
\mathbb{E}\!\left[Y_i(1)-f_{1,i}\mid X_i,U_i\right]
=0,
\end{align*}
and similarly for the second term. Therefore,
\begin{align*}
\mathbb{E}\!\left[(A_i-q_i)(G_i-m_i)\mid X_i,U_i\right]
&=
\mathbb{E}\!\left[(A_i-q_i)^2\mid X_i,U_i\right]
\mathbb{E}\!\left[f_{1,i}-f_{0,i}\mid X_i,U_i\right] \\
&=
\mathbb{E}\!\left[(A_i-q_i)^2 \tau(X_i,U_i)\mid X_i,U_i\right].
\end{align*}
Taking conditional expectation given $X_i=x$ yields
\begin{align*}
\mathbb{E}\!\left[(A_i-q_i)(G_i-m_i)\mid X_i=x\right]
=
\mathbb{E}\!\left[(A_i-q_i)^2\tau(X_i,U_i)\mid X_i=x\right].
\end{align*}
Therefore, the NC-CSF loss satisfies
\begin{align*}
&\mathbb{E}\!\left[
(A_i-q_i)(G_i-m_i-\tau(X_i)(A_i-q_i))
\middle| X_i=x
\right] \\
&\qquad=
\mathbb{E}\!\left[(A_i-q_i)^2\tau(X_i,U_i)\mid X_i=x\right]
-
\tau(x)\mathbb{E}\!\left[(A_i-q_i)^2\mid X_i=x\right],
\end{align*}
which is zero at
\begin{align*}
\tau^w(x)
=
\frac{
\mathbb{E}\!\left[(A_i-q_i)^2\tau(X_i,U_i)\mid X_i=x\right]
}{
\mathbb{E}\!\left[(A_i-q_i)^2\mid X_i=x\right]
}.
\end{align*}
Finally, if $\tau(x,u)=\tau(x)$, then the numerator of $\tau^w(x)$ simplifies to
$\tau(x)\mathbb{E}\!\left[(A_i-q_i)^2\mid X_i=x\right]$,
and as a result $\tau^w(x)=\tau(x)$.
\end{proof}

\begin{proof}[Proof of Proposition~\ref{prop:Neyman.orthogonality}]
Let $V_i=(X_i,W_i,Z_i,A_i)$. Under \eqref{eq:censoring}, $\mathbb{E}[\Gamma_i\mid V_i]=\mathbb{E}[G_i\mid V_i]$ according to \cite{cui2023estimating}. Therefore,
\begin{align*}
\mathbb{E}\!\left[Y^{\mathrm{res}}\mid X\right]
=
\mathbb{E}\!\left[
\mathbb{E}\!\left[Y^{\mathrm{res}}\mid V\right]
\middle| X
\right] 
=
\mathbb{E}\!\left[
\mathbb{E}\!\left[G-m\mid V\right]
\middle| X
\right] 
=
\mathbb{E}\!\left[G-m\mid X\right].
\end{align*}
Next,
\begin{align*}
&\mathbb{E}\!\left[Y^{\mathrm{res}}\mid X\right]
=
\mathbb{E}\!\left[G-m\mid X\right] \\
=&
\mathbb{E}\!\left[
\mathbb{E}\!\left[
A G(1;h) + (1-A)G(0;h) - q(Z,X)f_1(W,X) - \bigl(1-q(Z,X)\bigr)f_0(W,X)
\mid X,U
\right]
\middle| X
\right].
\end{align*}
By \eqref{eq:nc_censored} and \eqref{eq:bridge.function},
\begin{align*}
\mathbb{E}[A G(1;h)\mid X,U]
&=
\mathbb{E}[A\mid X,U]\cdot \mathbb{E}[G(1;h)\mid X,U] \\
&=
\mathbb{E}[q(Z,X)\mid X,U]\cdot \mathbb{E}[f_1(W,X)\mid X,U] \\
&=
\mathbb{E}[q(Z,X)f_1(W,X)\mid X,U],
\end{align*}
and similarly
\begin{align*}
\mathbb{E}[(1-A)G(0;h)\mid X,U]
=
\mathbb{E}[(1-q(Z,X))f_0(W,X)\mid X,U].
\end{align*}
Therefore, $\mathbb{E}[Y^{\mathrm{res}}_{\eta_0}\mid X]=0$.
Again by \eqref{eq:bridge.function} and the tower property,
\begin{align*}
\mathbb{E}\!\left[A^{\mathrm{res}}_{\eta_0}\mid X\right]
&=
\mathbb{E}\!\left[
\mathbb{E}\!\left[A-q(Z,X)\mid X,U\right]
\middle| X
\right]
=0.
\end{align*}

For the bridge $q$, since $A^{\mathrm{res}}=A-q$, the derivative of the score with respect to $q$ is proportional to
\begin{align*}
-\bigl(Y^{\mathrm{res}}_{\eta_0}-2\tau_0(X)A^{\mathrm{res}}_{\eta_0}\bigr).
\end{align*}
Taking conditional expectation given $X$ and using the two centering identities above yields zero. For $m$, the derivative of the score is proportional to
$-A^{\mathrm{res}}_{\eta_0}$,
whose conditional expectation given $X$ is also zero by the above argument.

For the survival-related nuisance functions $S_C$, $\rho$, define the censoring counting process and at-risk process by $N_{C,i}(t)=I(Y_i\le t,\Delta_i=0)$ and $R_i(t)=I(Y_i\ge t)$. Under \eqref{eq:censoring},
\begin{align*}
M_{C,i}(t\mid V_i)
=
N_{C,i}(t)-\int_0^t R_i(u)\,d\Lambda_C(u\mid V_i)
\end{align*}
is a martingale. 
For any perturbation $\delta(t;h)$ of $\rho(t;h)$, the first-order derivative of the score with respect to $\rho$ is proportional to
\begin{align*}
A_i^{\mathrm{res}}
\int_0^h \frac{\delta_{A_i}(t;h)}{S_C(t\mid V_i)}\,dM_{C,i}(t\mid V_i),
\end{align*}
which has mean zero conditional on $X_i$. The same argument applies to perturbations of $S_C$ and $\Lambda_C$.

\end{proof}

\begin{proof}[Proof of Proposition~\ref{prop:leaf_oracle}]
Let $\widetilde\tau_{\mathcal S}$ be the oracle estimator using the true nuisance functions. The argument in the proof of Proposition~\ref{prop:score}, now averaging over $X\in\mathcal S$, gives
\begin{align*}
&\mathbb{E}\!\left[
A_i^{\mathrm{res}}
(Y_i^{\mathrm{res}}-tA_i^{\mathrm{res}})
\mid X_i\in\mathcal S
\right] 
=
\mathbb{E}[(A_i^{\mathrm{res}})^2\tau(X_i,U_i)\mid X_i\in\mathcal S]
-t\,\mathbb{E}[(A_i^{\mathrm{res}})^2\mid X_i\in\mathcal S],
\end{align*}
whose unique zero is $t=\tau_{\mathcal S}^w$. Thus the oracle estimator is a ratio of empirical moments centered at $\tau_{\mathcal S}^w$. Boundedness, residual overlap, and $n_{\mathcal S}\to\infty$ yield the displayed central limit theorem and variance by the multivariate CLT and delta method.

For the feasible estimator, cross-fitting makes each observation independent of the nuisance estimates used to construct its score, conditional on the training folds. Proposition~\ref{prop:Neyman.orthogonality} removes the first-order Gateaux derivative at the truth. The remaining score difference consists of products of nuisance errors and is $o_p(n_{\mathcal S}^{-1/2})$ under the stated $o_p(n_{\mathcal S}^{-1/4})$ rates and boundedness conditions. Hence
\begin{align*}
\widehat\tau_{\mathcal S}-\widetilde\tau_{\mathcal S}
=o_p(n_{\mathcal S}^{-1/2}),
\end{align*}
and Slutsky's theorem gives the result. If $\tau(x,u)$ is constant on $\mathcal S$, the weighted target reduces to that common value.
\end{proof}

\section{Additional Numerical Details and Results}

\subsection{Numerical Experiment Details}\label{sec:simulation}
We describe more implementation details of NC-CSF using our proposed loss with estimated nuisance functions.

\subsubsection{Uncensored NC-CSF}

In the uncensored setting, the nuisance layer consists of
\[
q_i=q(Z_i,X_i), 
\qquad 
f_{1,i}=f(W_i,1,X_i), 
\qquad 
f_{0,i}=f(W_i,0,X_i),
\]
together with the bridge mixture
\[
m_i
=
q(Z_i,X_i)f(W_i,1,X_i)
+
\bigl(1-q(Z_i,X_i)\bigr)f(W_i,0,X_i).
\]
We then form the orthogonal residuals
\begin{equation}
A_i^{\mathrm{res}} = A_i-q_i,
\qquad
Y_i^{\mathrm{res}} = Y_i^{\mathrm{unc}}-m_i.
\label{eq:uncensored_residuals}
\end{equation}
The final-stage feature vector augments the raw observed covariates with nuisance summaries:
\begin{equation}
X_i^{\mathrm{final}}
=
[X_i,\;W_i,\;Z_i,\;q_i,\;f_{1,i},\;f_{0,i},\;m_i].
\label{eq:uncensored_final_features}
\end{equation}
The final learner is a causal forest fit on the bridge-adjusted residualized data. In our implementation, we use the \texttt{CausalForestDML} interface in EconML with externally constructed residuals and final-stage features.

\subsubsection{Censored NC-CSF}

We now extend the estimator to right-censored outcomes. As in the uncensored case, we estimate the treatment bridge \(q\), the treatment-specific outcome bridges \(f(W,1,X)\) and \(f(W,0,X)\), and the bridge mixture \(m\), and define
\begin{equation}
A_i^{\mathrm{res}} = A_i-\widehat q_i.
\label{eq:censored_treatment_residual}
\end{equation}

For censored outcomes, the outcome bridges are trained using a target-specific IPCW bridge pseudo-outcome \(\widetilde Y_i^{\mathrm{bridge}}\), used only for nuisance estimation. For the RMST target, using the truncated time/event pair \((\widetilde Y_i,\widetilde\Delta_i)\), we use
\begin{equation}
\widetilde Y_i^{\mathrm{bridge}}
=
\frac{\widetilde Y_i\widetilde\Delta_i}
{\widehat S_C(\widetilde Y_i\mid X_i,W_i,Z_i,A_i)}.
\label{eq:rmst_bridge_pseudo_outcome}
\end{equation}
For the survival-probability target, let
\begin{equation}
\widetilde Y_i^{(h)}=\min(Y_i,h), 
\qquad
\widetilde\Delta_i^{(h)}
=
\Delta_i\mathbf{1}\{Y_i\le h\}
+
\mathbf{1}\{Y_i>h\},
\label{eq:horizon_adjusted_pair}
\end{equation}
and define
\begin{equation}
\widetilde Y_i^{\mathrm{bridge}}
=
\frac{\mathbf{1}\{Y_i>h\}}
{\widehat S_C(\widetilde Y_i^{(h)}\mid X_i,W_i,Z_i,A_i)}.
\label{eq:surv_bridge_pseudo_outcome}
\end{equation}

We also fit arm-specific event-survival curves
\[
\widehat S_a(t\mid X_i,W_i,Z_i)
\approx
\widehat{\mathbb{P}}(T_i>t\mid X_i,W_i,Z_i,A_i=a),
\qquad a\in\{0,1\},
\]
using separate Cox models within treatment arms. For RMST, we use
\begin{equation}
T_{i,\mathrm{nuis}}=\min(Y_i,h), 
\qquad
\Delta_{i,\mathrm{nuis}}
=
\Delta_i\mathbf{1}\{Y_i<h\}
+
\mathbf{1}\{Y_i\ge h\}.
\label{eq:cox_pair_rmst}
\end{equation}
For survival probability, the Cox duration/event pair is \((Y_i,\Delta_i)\).

These event-survival curves define target-specific continuation terms. For RMST,
\begin{equation}
\rho_a^{\mathrm{RMST}}(t,X_i;h)
=
t+
\int_t^h
\frac{\widehat S_a(u\mid X_i,W_i,Z_i)}
{\widehat S_a(t\mid X_i,W_i,Z_i)}
\,du,
\qquad
0\le t\le h,\quad a\in\{0,1\}.
\label{eq:rho_rmst_continuation}
\end{equation}
For survival probability,
\begin{equation}
\rho_a(t,X_i;h)
=
\frac{\widehat S_a(h\mid X_i,W_i,Z_i)}
{\widehat S_a(t\mid X_i,W_i,Z_i)},
\qquad
0\le t\le h,\quad a\in\{0,1\}.
\label{eq:rho_continuation}
\end{equation}

Because \(U_i\) is unobserved, our main specification uses a proxy-adjusted working model for censoring,
\begin{equation}
\widehat S_C(t\mid X_i,W_i,Z_i,A_i),
\label{eq:conditional_sc}
\end{equation}
estimated by a censoring Cox model using \((Y_i,1-\Delta_i)\) and covariates \((X_i,W_i,Z_i,A_i)\).

The outcome residual plugs the bridge-adjusted regression term into the three-term censoring-robust score of \cite{cui2023estimating}. For survival probability,
\begin{align}
Y_i^{\mathrm{res}}
&=
\frac{
\widetilde \Delta_i^{(h)} \bigl(\mathbf{1}\{Y_i>h\}-\widehat m_i\bigr)
+
\bigl(1-\widetilde \Delta_i^{(h)}\bigr)
\bigl(\rho_{A_i}(\widetilde Y_i^{(h)},X_i;h)-\widehat m_i\bigr)
}{
\widehat S_C(\widetilde Y_i^{(h)} \mid X_i,W_i,Z_i,A_i)
}
\label{eq:survival_residual}
\\
&\quad
-
\int_0^{\widetilde Y_i^{(h)}}
\frac{d\widehat \Lambda_C(t \mid X_i,W_i,Z_i,A_i)}
{\widehat S_C(t \mid X_i,W_i,Z_i,A_i)}
\bigl(\rho_{A_i}(t,X_i;h)-\widehat m_i\bigr)\,dt.
\nonumber
\end{align}
For RMST, we replace \(\mathbf{1}\{Y_i>h\}\) with \(\widetilde Y_i\widetilde\Delta_i\), and replace \(\rho_a(t,x;h)\) with \(\rho_a^{\mathrm{RMST}}(t,x;h)\).

The final-stage feature vector is
\begin{equation}
X_i^{\mathrm{final}}
=
[X_i,\;W_i,\;Z_i,\;\widehat q_i,\;\widehat f_{1,i},\;\widehat f_{0,i},\;\widehat m_i,\;
\widehat S_{1,i}(h),\;\widehat S_{0,i}(h),\;\widehat S_{1,i}(h)-\widehat S_{0,i}(h)].
\label{eq:censored_final_features}
\end{equation}
Here \(\widehat S_{a,i}(h)\) denotes the target-specific scalar summary derived from the fitted event-survival curve under arm \(a\). Clipping thresholds are specified in Algorithm~\ref{alg:censored_final_training}.

\subsection{Data-Generating Mechanism}\label{sec:data.generating}

We generate $n$ i.i.d.\ samples with covariates $X_i \sim \mathcal{N}(0, I_p)$ and latent confounder $U_i \sim \mathcal{N}(0,1)$. To model negative controls, we construct proxies $Z_i$ and $W_i$ as noisy linear functions of $(X_i,U_i)$:
\[
Z_i = a_z U_i + b_z^\top X_i + \varepsilon_{z,i}, 
\quad
W_i = a_w U_i + b_w^\top X_i + \varepsilon_{w,i},
\]
with independent Gaussian noise.

Treatment assignment follows
\[
A_i \sim \mathrm{Bernoulli}\big(\sigma(b_0 + \gamma_u U_i + \alpha^\top X_i)\big),
\]
introducing unobserved confounding via $U_i$. Potential event times follow a Weibull proportional hazards model:
\[
T_i(a)\mid X_i,U_i \sim \mathrm{Weibull}\!\big(k, \lambda e^{-\eta_i(a)/k}\big), 
\quad
\eta_i(a)=\beta_t^\top X_i+\beta_u U_i+\tau a.
\]

The conditional treatment effect under full observability is
\begin{equation}
\E[T_i(1)-T_i(0)\mid X_i,U_i]
= \lambda \Gamma(1+1/k)\exp\!\left(-\tfrac{1}{k}(\beta_t^\top X_i+\beta_u U_i)\right)\big(e^{-\tau/k}-1\big).
\end{equation}

Although $U_i$ is unobserved, its conditional distribution given $(X_i, Z_i,W_i)$ is available due to the linear-Gaussian structure. Defining residualized proxies $\tilde Z_i=Z_i-b_z^\top X_i$ and $\tilde W_i=W_i-b_w^\top X_i$, we have
\[
U_i \mid X_i, Z_i,W_i \sim \mathcal{N}\!\big(\mu_{U\mid Z,W},\,\sigma^2_{U\mid Z,W}\big),
\]
with closed-form expressions for the mean and variance. Integrating over this distribution yields the observed-data treatment effect:
\begin{equation}
\begin{aligned}
&\mathbb{E}[T_i(1)-T_i(0)\mid X_i,Z_i,W_i] \\
&= \lambda\,\Gamma\!\left(1+\frac{1}{k}\right)
\exp\!\left(-\frac{1}{k}\beta_t^\top X_i
-\frac{\beta_u}{k}\mu_{U\mid Z,W}
+
\frac{1}{2}\frac{\beta_u^2}{k^2}\sigma^2_{U\mid Z,W}
\right)
\left(e^{-\tau/k}-1\right).
\end{aligned}
\end{equation}

We evaluate performance across regimes varying proxy quality $(\sigma_z,\sigma_w)$, confounding strength $(\gamma_u,\beta_u)$, treatment effect size $\tau_{\mathrm{log\mbox{-}hr}}$, and dimensionality $(p_x,p_z,p_w)$, under both linear and nonlinear models. We fix $a_z=a_w=1.5$, and it results into $12\times22$ configurations.

\subsection{Proxy-map stress tests beyond linear Gaussian proxies}
\label{sec:proxy_stress}

To isolate the proxy mechanism, we replace the construction by
\begin{align*}
Z &= a_Z\widetilde g(U)+b_Z^\top X+\varepsilon_Z,
&
W &= a_W\widetilde g(U)+b_W^\top X+\varepsilon_W,
\end{align*}
where $\widetilde g$ is standardized to have mean zero and variance one under $U\sim N(0,1)$. Treatment, event times, censoring, $X$, $U$, and the target effects are held fixed across arms. P0 is the submitted linear-Gaussian proxy model. P1 uses $g(u)=\tanh(u)$; P2 uses $g(u)=u$ with $t_3$ noise scaled to preserve variance; P3 uses the noninjective map $g(u)=u+0.5u^2$; and P4 coarsens each noisy proxy to one bit. P1 and P2 remove joint Gaussianity while preserving an essentially invertible signal. P3 and P4 remove exact invertibility or fine-grained proxy information.

\begin{table}[H]
\captionsetup{labelfont={color=black},textfont={color=black}}
\centering
\small
\setlength{\tabcolsep}{5pt}
\renewcommand{\arraystretch}{1.08}
\caption{Proxy-map stress tests over all $12\times22$ configurations. Entries are average RMSE; gain is the relative reduction from the matched Censored Baseline.}
\label{tab:proxy_stress}
\begin{tabular}{lcc}
\toprule
\textbf{Proxy arm} & \textbf{RMST: Baseline $\rightarrow$ NC-CSF} & \textbf{Survival: Baseline $\rightarrow$ NC-CSF} \\
\midrule
P0 (submitted) & $0.1662\rightarrow0.0994$ (40.2\%) & $0.0788\rightarrow0.0613$ (22.2\%) \\
P1 ($\tanh$) & $0.1641\rightarrow0.1012$ (38.3\%) & $0.0791\rightarrow0.0620$ (21.6\%) \\
P2 ($t_3$ noise) & $0.1621\rightarrow0.0993$ (38.8\%) & $0.0798\rightarrow0.0603$ (24.5\%) \\
P3 (quadratic, noninjective) & $0.1795\rightarrow0.1159$ (35.4\%) & $0.0856\rightarrow0.0693$ (19.1\%) \\
P4 (binary coarsening) & $0.1769\rightarrow0.1079$ (39.0\%) & $0.0830\rightarrow0.0661$ (20.4\%) \\
\bottomrule
\end{tabular}
\end{table}

NC-CSF improves on the matched baseline in every arm, with RMST reductions of 35.4--40.2\% and survival-probability reductions of 19.1--24.5\%. P3 and P4 nevertheless degrade the bridge-based method in the predicted direction. In an $n$-sweep from 500 to 16{,}000, excess RMST RMSE relative to P0 averages only $+0.9\%$ for P1 and $+0.8\%$ for P2 and changes sign across sample sizes, whereas it is positive at every sample size and averages $+12.8\%$ for P3 and $+9.2\%$ for P4. Thus more data repair an identified but difficult moment system more readily than proxy information lost through a noninjective or coarsened map.

P3 degrades NC-CSF RMST accuracy by 16.6\% relative to P0, compared with 8.1\% for the no-bridge ablation. This approximately twofold differential degradation localizes the damage to the bridge component rather than to the forest or censoring correction. The bridge is degraded, not disabled: $\operatorname{corr}\{g(U),U\}=0.816$, and the linear $R^2$ of $U$ on $(Z,W,X)$ falls from 0.733 under P0 to 0.489 under P3 and 0.503 under P4.

\subsection{Paired Monte Carlo uncertainty}
\label{sec:mc_uncertainty}

We reran the representative sample-size and dimension sweeps 30 times. NC-CSF and the matched baseline use the same training draw within each replication, and performance is evaluated on an independent test set of 10{,}000 observations against the observed-data target $\mathbb{E}[T(1)-T(0)\mid X,Z,W]$. Table~\ref{tab:mc_uncertainty} reports mean RMSE plus or minus Monte Carlo standard error, the paired NC-CSF-minus-baseline difference with a pointwise 95\% interval, and the number of NC-CSF wins.

\begin{table}[H]
\captionsetup{labelfont={color=black},textfont={color=black}}
\centering
\footnotesize
\setlength{\tabcolsep}{3.5pt}
\renewcommand{\arraystretch}{1.08}
\caption{Monte Carlo uncertainty over 30 paired replications.}
\label{tab:mc_uncertainty}
\begin{tabular}{lrrrrrr}
\toprule
\textbf{Setting} & \textbf{$n$/$d$} & \textbf{Baseline RMSE} & \textbf{NC-CSF RMSE} & \textbf{NC-CSF $-$ baseline [95\% CI]} & \textbf{Reduction} & \textbf{Wins} \\
\midrule
Case 10 ($n$) & 1{,}000  & $0.1404\pm0.0057$ & $0.1079\pm0.0058$ & $-0.0325\;[-0.0411,-0.0240]$ & 23.2\% & 29/30 \\
Case 10 ($n$) & 4{,}000  & $0.1282\pm0.0060$ & $0.0892\pm0.0027$ & $-0.0390\;[-0.0490,-0.0291]$ & 30.5\% & 30/30 \\
Case 10 ($n$) & 8{,}000  & $0.1158\pm0.0044$ & $0.0810\pm0.0028$ & $-0.0348\;[-0.0421,-0.0275]$ & 30.1\% & 30/30 \\
Case 10 ($n$) & 16{,}000 & $0.1031\pm0.0025$ & $0.0793\pm0.0024$ & $-0.0238\;[-0.0298,-0.0179]$ & 23.1\% & 29/30 \\
Case 12 ($d$) & 2  & $0.0588\pm0.0027$ & $0.0308\pm0.0023$ & $-0.0280\;[-0.0330,-0.0229]$ & 47.6\% & 29/30 \\
Case 12 ($d$) & 5  & $0.0321\pm0.0008$ & $0.0191\pm0.0009$ & $-0.0130\;[-0.0149,-0.0112]$ & 40.5\% & 30/30 \\
Case 12 ($d$) & 10 & $0.0508\pm0.0015$ & $0.0347\pm0.0027$ & $-0.0161\;[-0.0209,-0.0113]$ & 31.7\% & 28/30 \\
Case 12 ($d$) & 20 & $0.2207\pm0.0245$ & $0.1256\pm0.0069$ & $-0.0952\;[-0.1443,-0.0460]$ & 43.1\% & 30/30 \\
Case 12 ($d$) & 40 & $0.4359\pm0.0313$ & $0.2816\pm0.0103$ & $-0.1543\;[-0.2205,-0.0881]$ & 35.4\% & 30/30 \\
\bottomrule
\end{tabular}
\end{table}

Every paired interval lies below zero, and NC-CSF wins 28--30 of 30 replications in each setting. These intervals quantify Monte Carlo uncertainty in mean test RMSE, not uncertainty for an individual CATE.

\subsection{Marginal ATE learners}
\label{sec:marginal_baselines}

AIPW and TMLE target a marginal average effect and therefore are not directly comparable with individual-level CATE learners. For a transparent auxiliary comparison, we estimate the population-average survival-probability effect at the same horizon, assign that scalar estimate to every test individual, and evaluate it against the true individual effects from the simulation.

\begin{table}[H]
\captionsetup{labelfont={color=black},textfont={color=black}}
\centering
\small
\caption{Population-level learners evaluated as constant individual-effect predictors.}
\label{tab:marginal_baselines}
\begin{tabular}{lcc}
\toprule
\textbf{Model} & \textbf{Survival-probability RMSE} & \textbf{Survival-probability MAE} \\
\midrule
AIPW & 0.1127 & 0.0879 \\
TMLE & 0.1063 & 0.0875 \\
\bottomrule
\end{tabular}
\end{table}

The table is included to situate the benchmark, not as a head-to-head HTE comparison: a constant marginal prediction cannot represent genuine individual heterogeneity.

\subsection{ACTG175 Proxy-Split Sensitivity}
\label{sec:appendix_actg175_splits}

The main ACTG175 analysis uses the clinical \(X/W/Z\) split described in Table~\ref{tab:appendix_actg175_splits}. Because ACTG175 does not include variables that were designed as negative-control proxies, we also fit two additional clinically motivated proxy splits. These specifications should not be interpreted as definitive proximal identification designs; rather, they test how much the fitted age profile changes when immune-progression, disease-severity, and treatment-selection variables are routed through the proxy channels.

All sensitivity analyses used the same ACTG175 subset, treatment coding, RMST target, and horizon as the main analysis: \(n=1{,}083\), with 561 ddI controls, 522 ZDV+ddI treated subjects, and \(h=1{,}000\) days. For a fair comparison, the three NC-CSF proxy-split specifications and the re-run R-CSF baseline were fit with 10{,}000 trees in the final forest stage. Bootstrap uncertainty summaries used 200 nonparametric bootstrap replications, with model seed 42, bootstrap seed 2026, and the same fixed age-reference grid across all specifications.

\begin{table}[!htbp]
\centering
\footnotesize
\setlength{\tabcolsep}{4pt}
\renewcommand{\arraystretch}{1.1}
\caption{ACTG175 proxy-split sensitivity specifications. Variables not listed under \(W\) or \(Z\) are kept in \(X\).}
\begin{tabular}{p{0.18\linewidth}p{0.23\linewidth}p{0.20\linewidth}p{0.30\linewidth}}
\toprule
\textbf{Specification} & \textbf{\(W\) variables} & \textbf{\(Z\) variables} & \textbf{Rationale} \\
\midrule
Clinical split & CD4 count at \(20\pm5\) weeks, CD8 count at \(20\pm5\) weeks & Race, prior antiretroviral history & Follow-up immune markers are outcome-side progression proxies; race and antiretroviral history may proxy treatment-selection patterns in the analyzed subset. \\
Rec A & Karnofsky score, symptomatic status, baseline CD4, baseline CD8, CD4 at \(20\pm5\) weeks, CD8 at \(20\pm5\) weeks & Prior antiretroviral history, intravenous drug use, haemophilia & \(W\) emphasizes disease severity and immune progression; \(Z\) emphasizes variables that may affect regimen selection or anticipated adherence risk. \\
Rec B & Baseline CD4, baseline CD8, CD4 at \(20\pm5\) weeks, CD8 at \(20\pm5\) weeks & Karnofsky score, symptomatic status, prior antiretroviral history, haemophilia & \(W\) isolates lab-based immune measurements; \(Z\) collects decision-relevant clinical status and history variables. \\
\bottomrule
\end{tabular}
\label{tab:appendix_actg175_splits}
\end{table}

\begin{figure}[!htbp]
\centering
\includegraphics[width=0.82\linewidth]{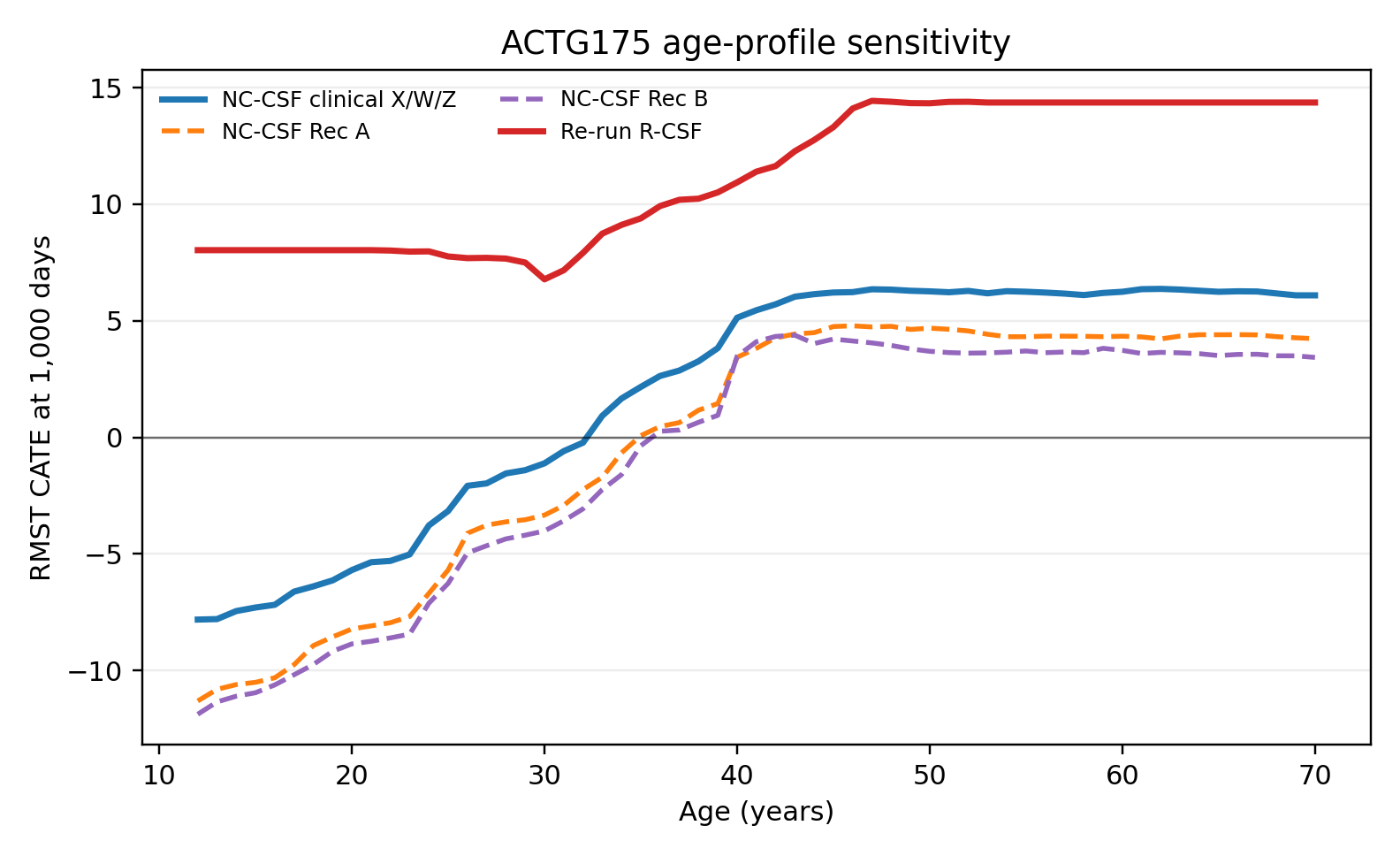}
\caption{Unified ACTG175 age-profile comparison for the main NC-CSF clinical split, two additional proxy-split sensitivity specifications, and our re-run R-CSF baseline. All curves target RMST at \(h=1{,}000\) days and use the same 10{,}000-tree final-stage setting.}
\label{fig:appendix_actg175_age_profiles}
\end{figure}

The sensitivity curves in Figure~\ref{fig:appendix_actg175_age_profiles} preserve the main qualitative pattern: estimated benefit is smaller for younger patients and larger for middle-aged or older patients. The proxy-split NC-CSF models have average fitted effects between 3.76 and 5.93 days, compared with 9.68 days for the re-run R-CSF baseline. Their peak ages range from 43 to 62 years, while their peak fitted effects range from 4.38 to 6.37 days. This indicates that the split choice changes the magnitude and smoothness of the age profile, but does not overturn the broad conclusion that age is the most visible heterogeneity dimension in this dataset.

\subsection{MIMIC-IV liver-disease subgroup analysis}
\label{sec:mimic_subgroups}

The liver-disease contrast was prespecified because hepatic CYP3A clearance is directly relevant to tacrolimus exposure. The cohort estimate $\widehat\tau^w$ for seven-day AKI-free survival is $-0.2410$. Table~\ref{tab:mimic_liver_subgroup} reports subgroup effects and each subgroup's difference from the cohort effect.

\begin{table}[H]
\captionsetup{labelfont={color=black},textfont={color=black}}
\centering
\scriptsize
\setlength{\tabcolsep}{4pt}
\renewcommand{\arraystretch}{1.0}
\caption{NC-CSF subgroup estimates by prior liver disease in the MIMIC-IV positive-control group.}
\label{tab:mimic_liver_subgroup}
\begin{tabular}{lrccccc}
\toprule
\textbf{Subgroup} & \textbf{$n$} & \textbf{$\widehat\tau^w$} & \textbf{95\% CI} & \textbf{$\Delta$ vs. cohort} & \textbf{95\% CI} & \textbf{E-value} \\
\midrule
Prior liver disease & 449 & $-0.3532$ & $[-0.3990,-0.2737]$ & $-0.1122$ & $[-0.1395,-0.0463]$ & 11.16 \\
No prior liver disease & 1{,}808 & $-0.2131$ & $[-0.2692,-0.1569]$ & $+0.0279$ & $[+0.0115,+0.0346]$ & 3.92 \\
\bottomrule
\end{tabular}
\end{table}

The harm estimate is $1.66$ times as large among patients with prior liver disease. Both the subgroup interval and the interval for its difference from the cohort exclude zero. By contrast, the three prespecified comorbidity axes without a direct metabolic link produce effect ratios of 0.85, 1.02, and 0.82 relative to the cohort, and every corresponding difference interval covers zero.

\subsection{Marginal versus Conditional Censoring}

Our main specification uses conditional censoring adjustment, modeling the censoring survival as 
$S_C(t\mid X_i,W_i,Z_i,A_i)$. We compare this with a marginal variant that replaces the subject-specific nuisance by a pooled curve $S_C(t)$. The conditional pseudo-outcome is
\begin{align}
Y_i^{\mathrm{res,cond}}
&=
\frac{
\widetilde \Delta_i (G_i-\widehat m_i)
+
(1-\widetilde \Delta_i)(\rho_{A_i}(\widetilde Y_i,X_i)-\widehat m_i)
}{
\widehat S_C(\widetilde Y_i \mid X_i,W_i,Z_i,A_i)
}
\label{eq:appendix_conditional_residual}
\\
&\quad
-
\int_0^{\widetilde Y_i}
\frac{d\widehat \Lambda_C(t \mid X_i,W_i,Z_i,A_i)}
{\widehat S_C(t \mid X_i,W_i,Z_i,A_i)}
(\rho_{A_i}(t,X_i)-\widehat m_i)\,dt.
\nonumber
\end{align}
The marginal variant replaces 
$\widehat S_C(\cdot\mid X_i,W_i,Z_i,A_i)$ and 
$d\widehat\Lambda_C(\cdot\mid X_i,W_i,Z_i,A_i)$ 
with $\widehat S_C(\cdot)$ and $d\widehat\Lambda_C(\cdot)$. Thus, the bridge layer and final forest are unchanged and only the censoring nuisance differs.

Table~\ref{tab:appendix_marginal_vs_conditional} compares both variants on the common $12\times22$ benchmark grid, using RMST at $h=3.0$ and survival probability at the observed-follow-up $q90$ horizon. Conditional and marginal NC-CSF perform nearly identically, while both variants outperform their corresponding baselines. Thus, the main gains come from the negative-control adjustment rather than the censoring model.

\begin{table}[H]
\centering
\small
\setlength{\tabcolsep}{5pt}
\renewcommand{\arraystretch}{1.1}
\caption{Marginal versus conditional censoring on the fixed-horizon benchmark. Lower RMSE and MAE indicate better performance.}
\begin{tabular}{llcc}
\toprule
\textbf{Target} & \textbf{Model} & \textbf{RMSE} & \textbf{MAE} \\
\midrule
\multirow{4}{*}{RMST ($h=3.0$)} 
& Conditional NC-CSF & 0.0994 & 0.0768 \\
& Marginal NC-CSF    & 0.0996 & 0.0763 \\
& Marginal Baseline  & 0.1214 & 0.0937 \\
& Conditional Baseline & 0.1662 & 0.1247 \\
\midrule
\multirow{4}{*}{Survival ($q90$)} 
& Conditional NC-CSF & 0.0613 & 0.0491 \\
& Marginal NC-CSF    & 0.0620 & 0.0495 \\
& Marginal Baseline  & 0.0676 & 0.0562 \\
& Conditional Baseline & 0.0788 & 0.0642 \\
\bottomrule
\end{tabular}
\label{tab:appendix_marginal_vs_conditional}
\end{table}

We use conditional censoring as the main specification because it preserves subject-level censoring heterogeneity, aligns naturally with the proxy-adjusted nuisance construction, and gives a slight RMST advantage in our benchmark. Since the gap between marginal and conditional censoring is small relative to the gain from NC-CSF over baseline estimators, we view the censoring specification as a secondary modeling choice rather than the primary driver of performance.

\subsection{Final-Stage Feature Ablation}

Table~\ref{tab:appendix_final_feature_ablation} evaluates two final-stage design choices: passing proxy-adjusted raw features to the final learner and augmenting them with nuisance summaries. The comparison uses the common $12 \times 22$ benchmark grid, with RMST at $h=3.0$ and survival probability at the $q90$ horizon. The \emph{NC-CSF (X)} and \emph{NC-CSF (X,W,Z)} variants omit nuisance summaries from the final learner.

\begin{table}[H]
\centering
\small
\setlength{\tabcolsep}{5pt}
\renewcommand{\arraystretch}{1.1}
\caption{Final-stage feature ablation on the fixed-horizon benchmark for RMST and survival probability. Lower RMSE and MAE indicate better performance.}
\begin{tabular}{llcc}
\toprule
\textbf{Target} & \textbf{Model} & \textbf{RMSE} & \textbf{MAE} \\
\midrule

\multirow{4}{*}{RMST ($h=3.0$)} 
& NC-CSF             & 0.0994 & 0.0768 \\
& NC-CSF (X,W,Z)     & 0.1081 & 0.0861 \\
& NC-CSF (X)         & 0.1147 & 0.0914 \\
& Censored Baseline  & 0.1662 & 0.1247 \\

\midrule

\multirow{4}{*}{Survival ($q90$)} 
& NC-CSF             & 0.0613 & 0.0491 \\
& NC-CSF (X,W,Z)     & 0.0675 & 0.0561 \\
& NC-CSF (X)         & 0.0717 & 0.0595 \\
& Censored Baseline  & 0.0788 & 0.0642 \\

\bottomrule
\end{tabular}
\label{tab:appendix_final_feature_ablation}
\end{table}

The bridge-adjusted feature set already gives a large gain over the baseline. Compared with the censored baseline, \emph{NC-CSF (X,W,Z)} reduces RMST RMSE from $0.1662$ to $0.1081$ ($35.0\%$) and RMST MAE from $0.1247$ to $0.0861$ ($30.9\%$). For survival probability, RMSE decreases from $0.0788$ to $0.0675$ ($14.3\%$), and MAE decreases from $0.0642$ to $0.0561$ ($12.6\%$).

Adding nuisance summaries yields a further improvement. Compared with \emph{NC-CSF (X,W,Z)}, the full \emph{NC-CSF} reduces RMST RMSE from $0.1081$ to $0.0994$ ($8.0\%$) and survival-probability RMSE from $0.0675$ to $0.0613$ ($9.2\%$). Overall, compared with the censored baseline, the full \emph{NC-CSF} reduces RMST RMSE by $40.2\%$ and survival-probability RMSE by $22.2\%$.

This isolates the bridge as the primary source of improvement. The forest family is identical across the matched comparisons; RSF-T and RSF-S omit residualization and give RMST RMSEs of 0.1773 and 0.1637, bracketing the 0.1662 of the residualized no-bridge baseline; and conditional versus marginal censoring changes NC-CSF only from 0.0994 to 0.0996 for RMST and from 0.0613 to 0.0620 for survival probability. Thus neither the forest, residualization alone, nor the precise censoring specification explains the bridge-adjusted gain.

Intuitively, raw $(X,W,Z)$ require the final learner to reconstruct the bridge-adjusted signal indirectly, whereas nuisance summaries such as $\widehat q$, $\widehat f_1$, and $\widehat f_0$ provide denoised, low-dimensional features aligned with the orthogonalized residualization step. This makes the final learner more effective at recovering treatment-effect heterogeneity.

\subsection{RMST Under the Default \texorpdfstring{$\max(Y)$}{max(Y)} Horizon}

The main text and fixed-horizon appendix tables evaluate RMST at a common benchmark horizon, such as $h=3.0$. The original synthetic benchmark also supports a default convention in which, when no explicit RMST horizon is supplied, the target is evaluated at the maximum observed follow-up within each benchmark instance:
\[
h_{\mathrm{cell}}=\max_{1\le i\le n_{\mathrm{cell}}} Y_i.
\]
This comparison serves as a robustness check, testing whether NC-CSF maintains its performance advantage when the RMST horizon is not fixed. 

The comparison includes the full $12\times22=264$ case-setting cells for each model. Under the default $\max(Y)$ convention, the realized RMST horizon varies across cells because it is determined by the maximum observed follow-up in each benchmark instance. We also report the corresponding survival-probability benchmark results from the same output grid for completeness.

\begin{table}[H]
\centering
\small
\setlength{\tabcolsep}{4pt}
\renewcommand{\arraystretch}{1.1}
\caption{Benchmark results under the default RMST horizon $h=\max_i Y_i$. Lower RMSE and MAE indicate better performance.}
\begin{tabular}{lcc}
\toprule
\textbf{Model} & \textbf{RMSE} & \textbf{MAE} \\
\midrule

\multicolumn{3}{c}{\textbf{Panel A: RMST ($h=\max_i Y_i$)}} \\
\midrule
NC-CSF Oracle & 0.0994 & 0.0559 \\
NC-CSF        & 0.1892 & 0.1155 \\
Baseline      & 0.3766 & 0.2609 \\

\midrule

\multicolumn{3}{c}{\textbf{Panel B: Survival (benchmark quantile horizon)}} \\
\midrule
NC-CSF Oracle & 0.0413 & 0.0337 \\
NC-CSF        & 0.0656 & 0.0556 \\
Baseline      & 0.0691 & 0.0587 \\

\bottomrule
\end{tabular}
\label{tab:appendix_maxy}
\end{table}

Table~\ref{tab:appendix_maxy} shows that the main qualitative ranking is preserved under the default $\max(Y)$ RMST convention. For RMST, NC-CSF reduces RMSE from $0.3766$ to $0.1892$ ($49.8\%$) and MAE from $0.2609$ to $0.1155$ ($55.7\%$) relative to the baseline. Thus, the advantage of NC-CSF is not specific to the fixed-horizon choice $h=3.0$.

For survival probability, the gain is smaller but still favorable: NC-CSF reduces RMSE from $0.0691$ to $0.0656$ ($5.0\%$) and MAE from $0.0587$ to $0.0556$ ($5.2\%$). We therefore interpret this appendix result as a robustness check: changing the RMST horizon convention from a fixed benchmark horizon to the default $\max(Y)$ rule does not overturn the main conclusion that NC-CSF improves upon the baseline.

The gains are larger for RMST than for horizon-specific survival probability, likely because RMST integrates survival over time and is therefore smoother. In contrast, survival probability depends on a single time point and is more sensitive to local errors in censoring correction and continuation-term estimation.

\subsection{Sensitivity to RMST and Survival Horizons}

Tables~\ref{tab:horizon_rmse} and~\ref{tab:horizon_bias} summarize average performance over all $12 \times 22$ case-setting combinations, separately for each horizon. For RMST, we use the four fixed horizons from the CSF synthetic benchmark, $h \in \{1.5, 2.0, 3.0, 15.0\}$. For survival probability, we sweep quantile horizons from $q60$ to $q90$. This comparison serves as a robustness check across time scales rather than a result tied to a single horizon choice.

\begin{table}[H]
\centering
\small
\setlength{\tabcolsep}{5pt}
\renewcommand{\arraystretch}{1.1}
\caption{RMSE comparison across horizon choices. Gain denotes the relative reduction from the baseline; lower RMSE is better.}
\begin{tabular}{lccc}
\toprule
\textbf{Horizon} & \textbf{NC-CSF} & \textbf{Baseline} & \textbf{Gain} \\
\midrule
\multicolumn{4}{c}{\textbf{RMST (fixed horizons)}} \\
\midrule
$h=1.5$  & 0.063 & 0.077 & 18.0\% \\
$h=2.0$  & 0.077 & 0.105 & 26.2\% \\
$h=3.0$  & 0.099 & 0.166 & 40.2\% \\
$h=15.0$ & 0.341 & 1.016 & 66.4\% \\

\midrule
\multicolumn{4}{c}{\textbf{Survival (quantile horizons)}} \\
\midrule
$q60$ & 0.066 & 0.069 & 5.0\% \\
$q70$ & 0.068 & 0.074 & 7.9\% \\
$q80$ & 0.068 & 0.080 & 15.1\% \\
$q90$ & 0.061 & 0.079 & 22.2\% \\
\bottomrule
\end{tabular}
\label{tab:horizon_rmse}
\end{table}

\begin{table}[H]
\centering
\small
\setlength{\tabcolsep}{5pt}
\renewcommand{\arraystretch}{1.1}
\caption{Absolute bias comparison across horizon choices. Gain denotes the relative reduction from the baseline; lower absolute bias is better.}
\begin{tabular}{lccc}
\toprule
\textbf{Horizon} & \textbf{NC-CSF} & \textbf{Baseline} & \textbf{Gain} \\
\midrule
\multicolumn{4}{c}{\textbf{RMST (fixed horizons)}} \\
\midrule
$h=1.5$  & 0.038 & 0.043 & 11.0\% \\
$h=2.0$  & 0.046 & 0.058 & 20.3\% \\
$h=3.0$  & 0.055 & 0.084 & 34.5\% \\
$h=15.0$ & 0.160 & 0.488 & 67.2\% \\

\midrule
\multicolumn{4}{c}{\textbf{Survival (quantile horizons)}} \\
\midrule
$q60$ & 0.036 & 0.036 & 0.9\% \\
$q70$ & 0.039 & 0.041 & 5.2\% \\
$q80$ & 0.038 & 0.043 & 12.1\% \\
$q90$ & 0.031 & 0.037 & 17.3\% \\
\bottomrule
\end{tabular}
\label{tab:horizon_bias}
\end{table}

The results show that NC-CSF achieves lower RMSE and smaller absolute bias than the baseline at every RMST horizon and every survival horizon. For RMST, the gain increases with the horizon and is largest at $h=15.0$, suggesting that bridge-adjusted nuisance learning is especially useful when long-horizon targets accumulate more censoring and survival-modeling error.

A similar pattern appears for survival probability. Gains are modest at $q60$ and $q70$ but become larger at $q80$ and $q90$, where censoring and continuation-term estimation are more influential. Overall, the advantage of NC-CSF is not confined to a single horizon choice; it persists across both the RMST horizon set and the survival quantile sweep, with the strongest improvements in more challenging long-horizon regimes.

\subsection{Compatibility with the Original CSF Synthetic Benchmark}\label{sec:appendix_csf_benchmark}

We additionally evaluate our estimator on the original synthetic simulation benchmark of Cui et al.~\cite{cui2023estimating}. This experiment is distinct from our main proximal benchmark: the original CSF simulation study does not include a latent unobserved confounder $U$ or negative-control proxy variables, and all confounding is induced through the observed covariates $X$. Thus, this appendix does not test the benefit of proximal adjustment. Instead, it checks whether the final censored learner used in NC-CSF remains stable in a standard fully observed causal survival setting.

We follow the Section~4 simulation protocol of Cui et al.~\cite{cui2023estimating}. The four data-generating settings are the original CSF synthetic settings, with independent covariates drawn from $[0,1]^p$ and $p=15$. Each method is trained on $n=2{,}000$ observations and evaluated on an independent test set of size $n_{\mathrm{test}}=2{,}000$, over $250$ simulation replications. We consider the same two estimands as the original paper: restricted mean survival time (RMST) and survival probability. The RMST horizons are $(1.5,2.0,15.0,3.0)$ for settings 1--4, respectively, and the survival probability horizons are $(0.8,1.2,10.0,2.0)$.

For a fair comparison, we re-run R-CSF locally using the direct \texttt{grf} implementation under the same simulation protocol. R-CSF is fit with $2{,}000$ trees, minimum node size $5$, subsampling fraction $0.5$, and the default split-variable rule $\min(p,\lceil\sqrt p\rceil+20)$. The nuisance survival forests use $500$ trees and minimum node size $15$, matching the original benchmark. Because this benchmark has no designated negative-control proxies, we evaluate the all-$X$ variant of NC-CSF, where all observed covariates are passed through the observed feature block $X$. Therefore, this experiment evaluates the compatibility of the NC-CSF final censored learner with the original CSF setting, rather than the benefit of negative-control adjustment.

\begin{table}[H]
\centering
\small
\setlength{\tabcolsep}{5pt}
\renewcommand{\arraystretch}{1.1}
\caption{Original CSF synthetic benchmark results. MSE is multiplied by 100. Excess MSE is computed relative to the better method in each replication.}
\begin{tabular}{llcc}
\toprule
\textbf{Setting} & \textbf{Metric} & \textbf{Re-run R-CSF} & \textbf{NC-CSF} \\
\midrule

\multicolumn{4}{l}{\textbf{Panel A: RMST}} \\
1 & MSE        & 0.27  & 0.25 \\
  & Excess MSE & 1.14  & 1.08 \\
2 & MSE        & 1.15  & 1.24 \\
  & Excess MSE & 1.03  & 1.16 \\
3 & MSE        & 13.37 & 14.67 \\
  & Excess MSE & 1.03  & 1.15 \\
4 & MSE        & 2.97  & 2.33 \\
  & Excess MSE & 1.38  & 1.02 \\

\midrule

\multicolumn{4}{l}{\textbf{Panel B: Survival probability}} \\
1 & MSE        & 0.16  & 0.15 \\
  & Excess MSE & 1.15  & 1.16 \\
2 & MSE        & 0.37  & 0.38 \\
  & Excess MSE & 1.05  & 1.11 \\
3 & MSE        & 0.15  & 0.16 \\
  & Excess MSE & 1.04  & 1.09 \\
4 & MSE        & 0.28  & 0.28 \\
  & Excess MSE & 1.14  & 1.09 \\

\bottomrule
\end{tabular}
\label{tab:appendix_csf_section4_mse}
\end{table}

\begin{table}[H]
\centering
\small
\setlength{\tabcolsep}{5pt}
\renewcommand{\arraystretch}{1.1}
\caption{Classification error on the original CSF synthetic benchmark. Classification error is defined as $1-(1/n_{\mathrm{test}})\sum_{i=1}^{n_{\mathrm{test}}}\mathbf{1}\{\mathrm{sign}(\widehat{\tau}(X_i))=\mathrm{sign}(\tau(X_i))\}$. For Setting~4, classification error is evaluated on the subgroup $X^{(1)} \geq 0.3$, following the original benchmark.}
\label{tab:appendix_csf_section4_classification}
\begin{tabular}{lcc}
\toprule
\textbf{Setting} & \textbf{Re-run R-CSF} & \textbf{NC-CSF} \\
\midrule

\multicolumn{3}{l}{\textbf{Panel A: RMST}} \\
1 & 0.22 & 0.22 \\
2 & 0.17 & 0.18 \\
3 & 0.09 & 0.09 \\
4 & 0.01 & 0.03 \\

\midrule

\multicolumn{3}{l}{\textbf{Panel B: Survival probability}} \\
1 & 0.22 & 0.25 \\
2 & 0.23 & 0.24 \\
3 & 0.10 & 0.11 \\
4 & 0.04 & 0.14 \\

\bottomrule
\end{tabular}
\end{table}

Tables~\ref{tab:appendix_csf_section4_mse} and~\ref{tab:appendix_csf_section4_classification} show that NC-CSF remains broadly competitive with the re-run R-CSF in the fully observed, no-latent-confounding benchmark. This is a demanding comparison for NC-CSF because the benchmark was designed for the original R-CSF framework and contains no negative-control structure. NC-CSF achieves similar MSE in most settings and improves RMST MSE in Settings~1 and~4, while R-CSF remains stronger in some classification-error comparisons, especially survival probability in Setting~4. We therefore interpret this experiment as a compatibility check: NC-CSF does not replace R-CSF in fully observed settings, but its final censored learner remains stable when applied outside the proximal benchmark.

\subsection{Detailed Synthetic Setup}

\begin{table}[H]
\centering
\caption{Synthetic experiment configuration grid. The left column reports the parameter ranges and the 12 base configurations; the right column reports the 22 dimensional settings. The full benchmark crosses the 12 cases with the 22 settings.}
\begin{minipage}[t]{0.62\textwidth}
\vspace{0pt}
\centering
\small
\setlength{\tabcolsep}{4pt}
\renewcommand{\arraystretch}{1.08}
{\footnotesize (a) Parameter ranges}

\begin{tabular}{ll}
\toprule
\textbf{Value} & \textbf{Range} \\
\midrule
$\mathrm{Corr}(Z,U)$ & $[0.1,\,0.9]$ \\
$\sigma_z,\sigma_w$ & $[0.73,\,14.93]$ \\
$(\gamma_u,\beta_u)$ & $[0,\,2]$ \\
$\tau_{\mathrm{log\mbox{-}hr}}$ & $[-0.7,\,0.7]$ \\
model & Linear / Nonlinear \\
$n$ & $[500,\,16000]$ \\
$d_x$ & $[2,\,40]$ \\
$d_z,\,d_w$ & $[1,\,20]$ \\
\bottomrule
\end{tabular}

\vspace{3pt}

\vspace{10pt}

\small
\setlength{\tabcolsep}{3pt}
\renewcommand{\arraystretch}{1.08}
{\footnotesize (b) Base configurations}
\begin{tabular}{cccccccc}
\toprule
\textbf{Case} & $\beta_u$ & $\gamma_u$ & $\sigma_z$ & $\sigma_w$ & $\tau_{\text{log-hr}}$ & \textbf{Linear T} & \textbf{Linear Y} \\
\midrule
C01 & 0.8  & 0.8  & 1.125 & 1.53 & -0.6  & Yes & Yes \\
C02 & 0.3  & 0.2  & 1.125 & 1.53 & -0.25 & Yes & Yes \\
C03 & 0.9  & 0.9  & 7.35  & 4.77 & -0.5  & Yes & Yes \\
C04 & 0.2  & 0.3  & 7.35  & 4.77 & -0.12 & Yes & Yes \\
C05 & 0.8  & 0.9  & 1.125 & 1.53 & -0.7  & Yes & No  \\
C06 & 0.2  & 0.2  & 1.125 & 1.53 & -0.3  & Yes & No  \\
C07 & 1.1  & 0.9  & 7.35  & 4.77 & -0.5  & Yes & No  \\
C08 & 0.15 & 0.25 & 7.35  & 4.77 & -0.12 & Yes & No  \\
C09 & 1.0  & 1.0  & 1.125 & 1.53 & -0.7  & No  & No  \\
C10 & 0.2  & 0.2  & 1.125 & 1.53 & -0.3  & No  & No  \\
C11 & 1.1  & 0.9  & 7.35  & 4.77 & -0.5  & No  & No  \\
C12 & 0.2  & 0.3  & 7.35  & 4.77 & -0.12 & No  & No  \\
\bottomrule
\end{tabular}

\vspace{3pt}
\end{minipage}
\hfill
\begin{minipage}[t]{0.32\textwidth}
\vspace{0pt}
\centering
\small
\setlength{\tabcolsep}{4pt}
\renewcommand{\arraystretch}{1.08}
{\footnotesize (c) Dimensional settings}
\begin{tabular}{ccccc}
\toprule
\textbf{Setting} & $n$ & $d_x$ & $d_w$ & $d_z$ \\
\midrule
S01 & 500   & 5  & 1  & 1  \\
S02 & 1000  & 5  & 1  & 1  \\
S03 & 2000  & 5  & 1  & 1  \\
S04 & 4000  & 5  & 1  & 1  \\
S05 & 8000  & 5  & 1  & 1  \\
S06 & 16000 & 5  & 1  & 1  \\
S07 & 2000  & 2  & 1  & 1  \\
S08 & 2000  & 5  & 1  & 1  \\
S09 & 2000  & 10 & 1  & 1  \\
S10 & 2000  & 20 & 1  & 1  \\
S11 & 2000  & 40 & 1  & 1  \\
S12 & 2000  & 5  & 2  & 2  \\
S13 & 2000  & 5  & 3  & 3  \\
S14 & 2000  & 5  & 5  & 5  \\
S15 & 2000  & 5  & 10 & 10 \\
S16 & 2000  & 5  & 20 & 20 \\
S17 & 2000  & 5  & 1  & 10 \\
S18 & 2000  & 5  & 10 & 1  \\
S19 & 2000  & 10 & 5  & 5  \\
S20 & 2000  & 20 & 5  & 5  \\
S21 & 2000  & 40 & 10 & 10 \\
S22 & 4000  & 20 & 10 & 10 \\
\bottomrule
\end{tabular}

\vspace{3pt}
\end{minipage}
\label{tab:synthetic-configurations}

\end{table}

The 12 cases isolate key problem dimensions under fixed proxy loading $a_z=a_w=1.5$. Cases C01–C04 consider linear treatment and outcome models, varying proxy quality via $(\sigma_z,\sigma_w)$ (informative vs.\ noisy proxies) and confounding strength via $(\beta_u,\gamma_u)$ (weak to strong). Cases C05–C08 introduce nonlinear outcome models, which capture misspecification in the survival process. Cases C09–C12 further incorporate nonlinear treatment assignment, representing the most challenging regimes. We focus on beneficial treatment effects (negative log-hazard ratios) to reflect common clinical settings where interventions are expected to improve survival. This restriction is without loss of generality: reversing the sign of $\tau_{\mathrm{log\mbox{-}hr}}$ yields the harmful-treatment counterpart, and all error metrics considered (e.g., RMSE, MAE, subgroup bias) are invariant to such sign reversal. Moreover, including both positive and negative effects within the same evaluation can lead to sign cancellation when averaging treatment effects, obscuring systematic estimation errors. Restricting to a single direction avoids this cancellation and provides a clearer assessment of model performance.

The experimental settings are designed to isolate the effect of different problem dimensions. Settings S01--S06 focus on sample size, varying $n$ from 500 to 16{,}000 while keeping covariate and proxy dimensions fixed. Settings S07--S11 vary the covariate dimensionality $d_x$ from 2 to 40 with fixed sample size. Settings S12--S16 jointly increase proxy dimensions $d_w$ and $d_z$, which capture regimes with richer proxy information. Settings S17--S22 are the stress test cases, which explore asymmetric and high-dimensional configurations.

\subsection{Additional Results and Graphs for Uncensored Cases}

\begin{table}[H]
\centering
\footnotesize
\setlength{\tabcolsep}{5pt}
\renewcommand{\arraystretch}{1.1}
\caption{Summary performance in the uncensored setting.}
\begin{tabular}{lcc}
\toprule
\textbf{Model} & \textbf{RMSE} & \textbf{MAE} \\
\midrule
NC-CSF Oracle       & 0.3347 & 0.1203 \\
NC-CSF              & 0.3518 & 0.1309 \\
Uncensored Baseline & 0.3988 & 0.1861 \\
\bottomrule
\end{tabular}
\label{tab:uncensored-summary}
\end{table}

\begin{figure}[H]
\centering

\begin{minipage}[t]{0.48\linewidth}
    \raggedright
    \includegraphics[width=\linewidth]{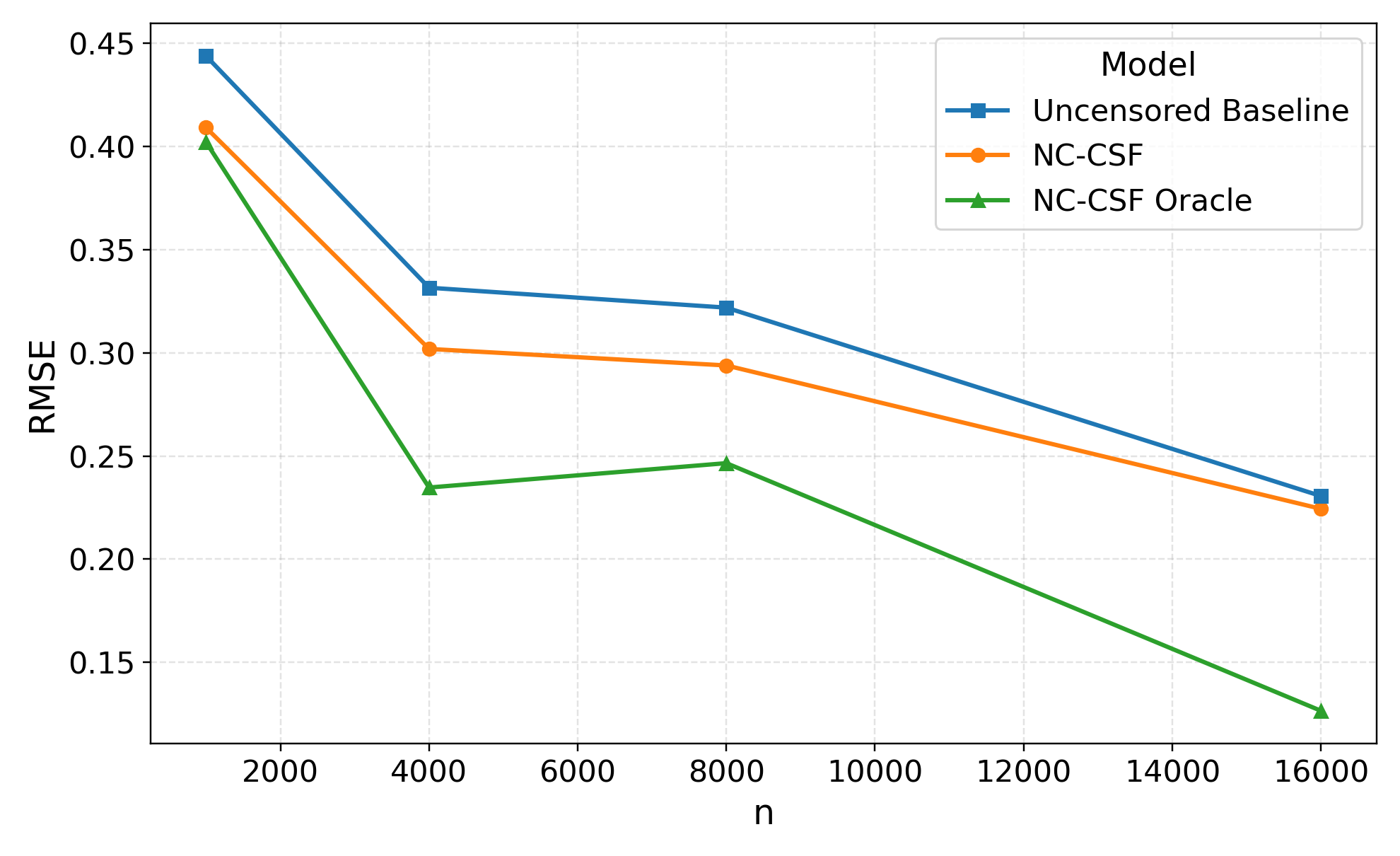}
\end{minipage}
\hfill
\begin{minipage}[t]{0.48\linewidth}
    \raggedright
    \includegraphics[width=\linewidth]{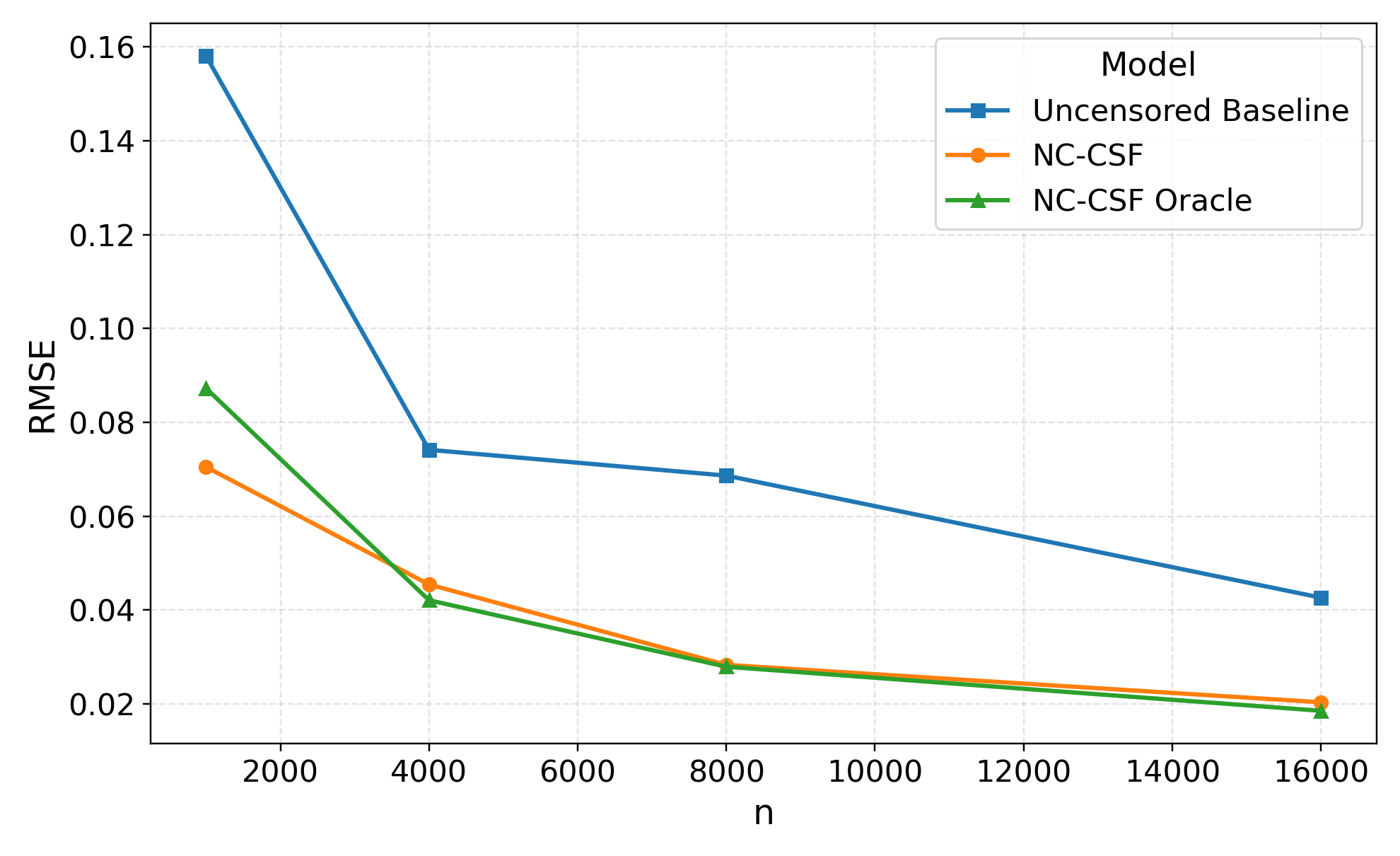}
\end{minipage}

\caption{
\textbf{Left:} RMSE as a function of sample size $n$ under Case 11.
\textbf{Right:} RMSE as a function of sample size $n$ under Case 12.
In both cases, NC-CSF consistently achieves lower RMSE.
}
\label{fig:appendix-uncensored-n-sweep}
\end{figure}

\begin{figure}[H]
\centering

\begin{minipage}[t]{0.48\linewidth}
    \raggedright
    \includegraphics[width=\linewidth]{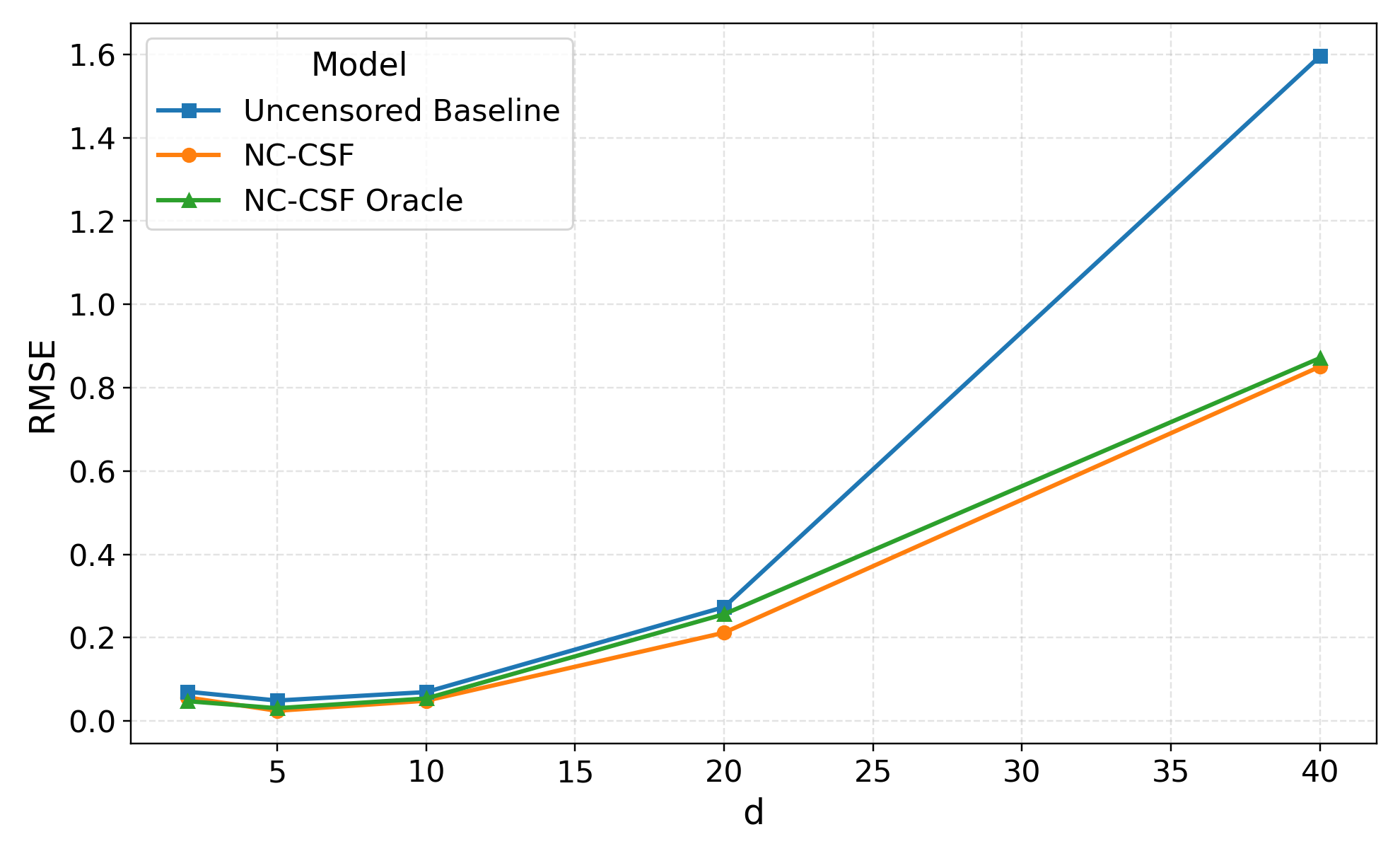}
\end{minipage}
\hfill
\begin{minipage}[t]{0.48\linewidth}
    \raggedright
    \includegraphics[width=\linewidth]{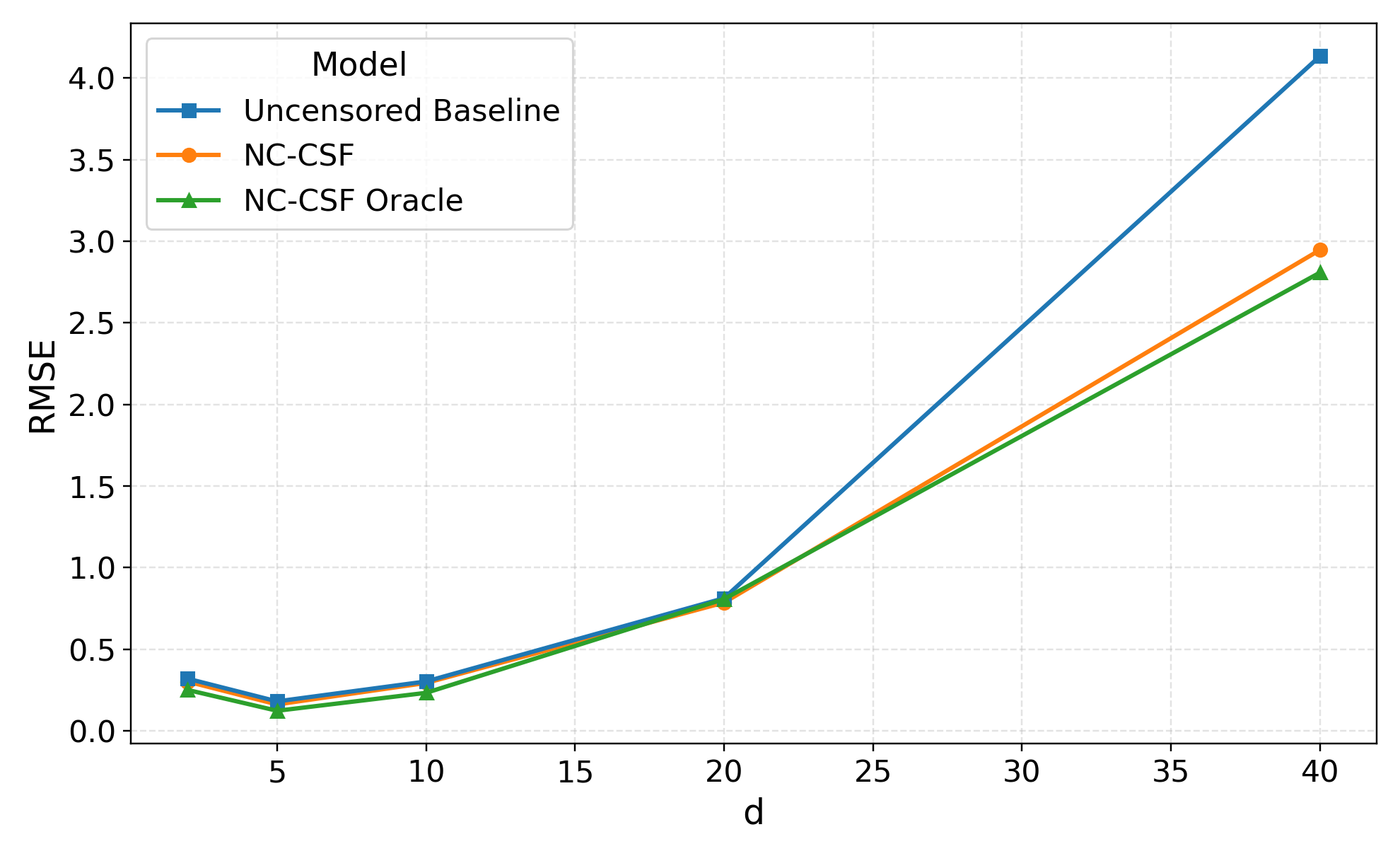}
\end{minipage}

\caption{
\textbf{Left:} RMSE as a function of covariate dimensionality $d$ under Case 10.
\textbf{Right:} RMSE as a function of $d$ under Case 11.
In both cases, NC-CSF scales more robustly with increasing dimensionality $d$, while the baseline degrades substantially in high-dimensional regimes.
}
\label{fig:appendix-uncensored-d-sweep}
\end{figure}

\begin{figure}[H]
\centering
\includegraphics[width=0.9\linewidth]{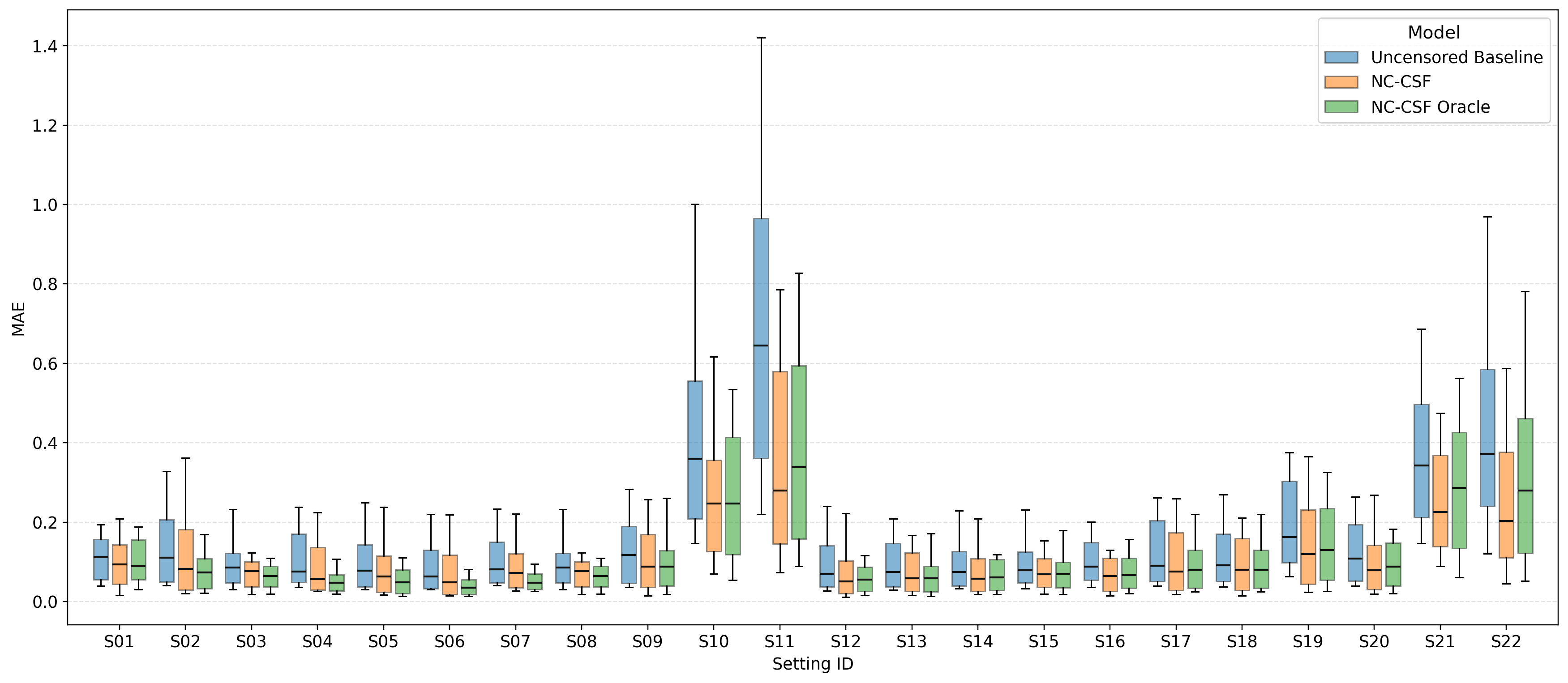}
\caption{MAE of CATE estimation across settings S01--S22 without censoring.
While NC-CSF generally improves over the baseline, the gains are modest.}
\label{fig:appendix_mae_uncensored}
\end{figure}

\begin{figure}[H]
\centering
\includegraphics[width=0.9\linewidth]{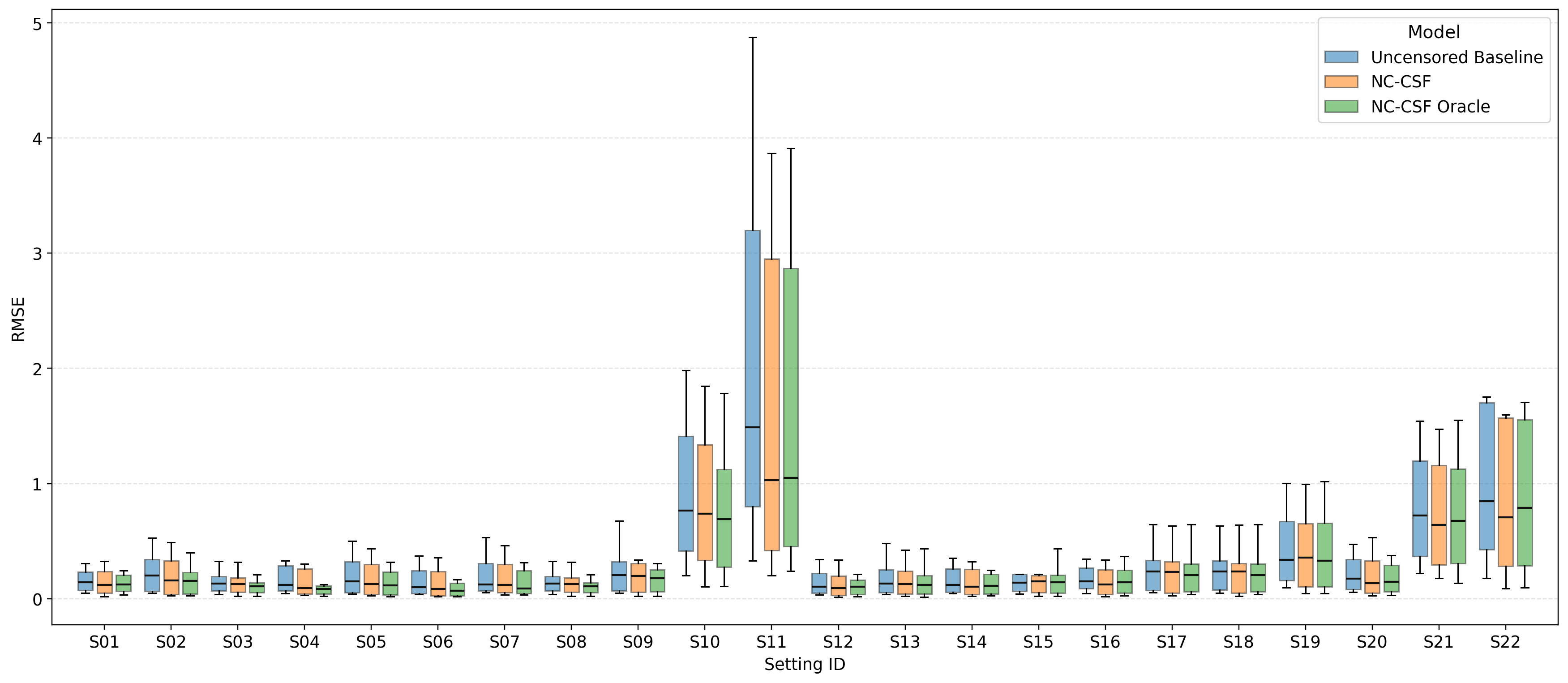}
\caption{RMSE of CATE estimation across settings S01--S22 under uncensoring. While NC-CSF generally improves over the baseline, the gains are modest in the uncensored setting; as shown later, the advantage becomes substantially more pronounced under censoring.}
\label{fig:appendix_rmse_uncensored}
\end{figure}

\subsection{Additional Results and Graphs for Censored Cases}

\begin{figure}[H]
\centering
\includegraphics[width=0.9\linewidth]{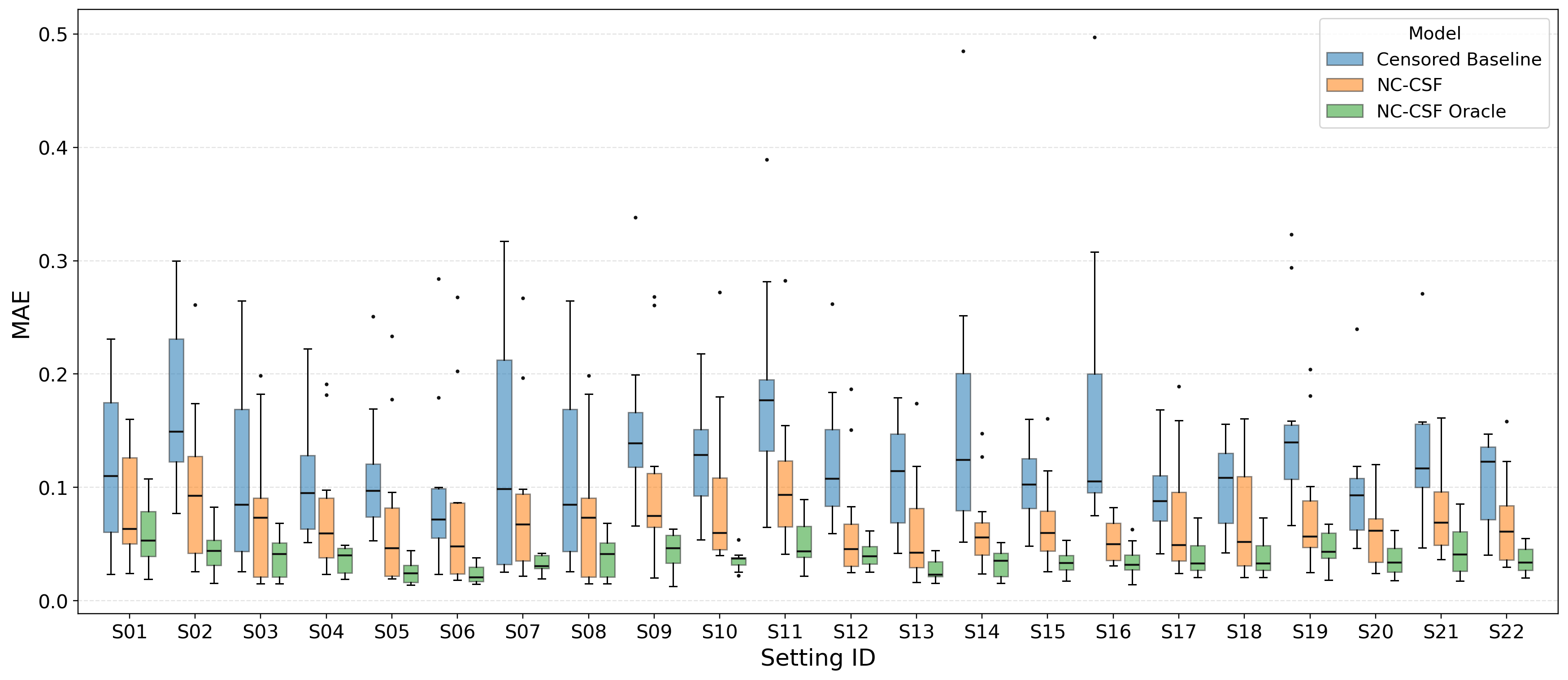}
\caption{MAE for RMST estimation across settings S01--S22 under censoring. NC-CSF achieves substantially lower MAE than the baseline and remains close to the oracle.}
\label{fig:appendix_mae_rmst_censored}
\end{figure}

\begin{figure}[H]
\centering
\includegraphics[width=0.9\linewidth]{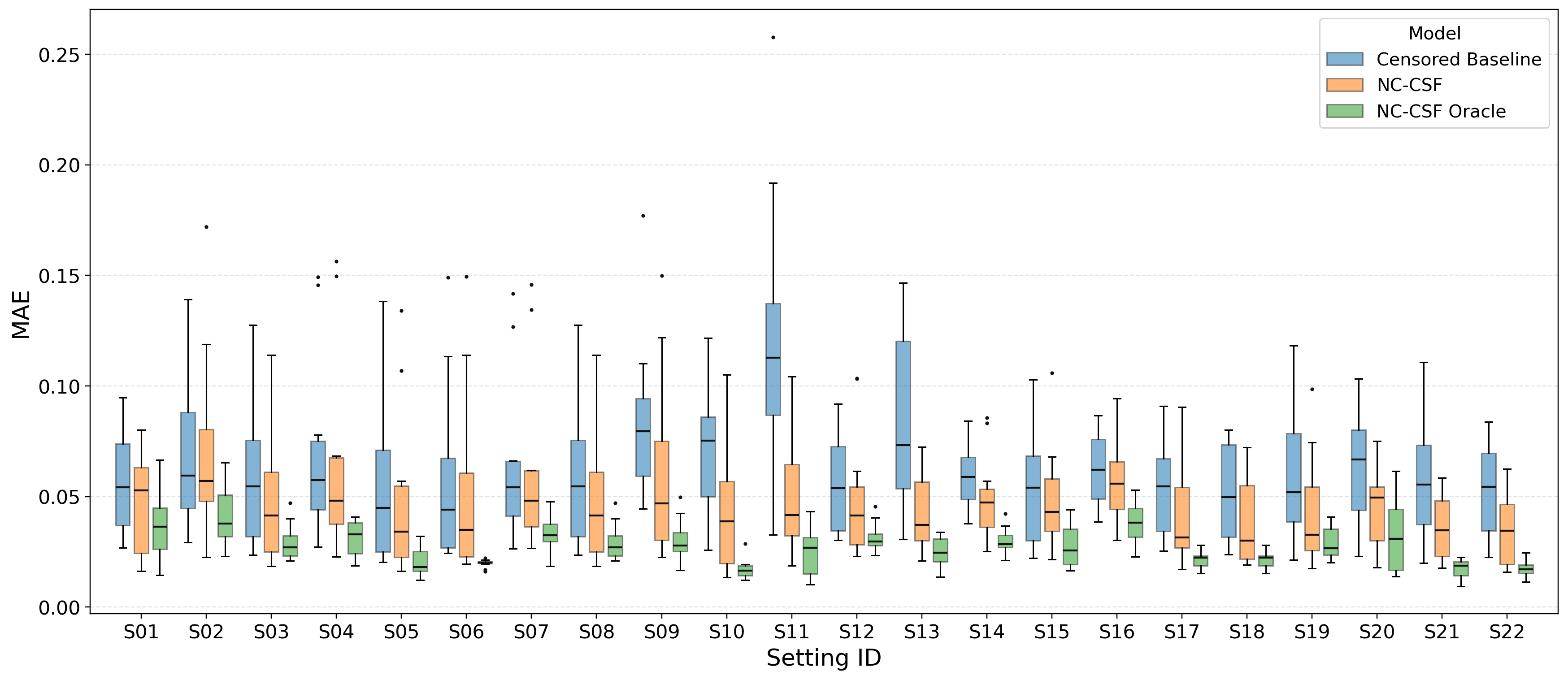}
\caption{MAE for survival probability estimation across settings S01--S22 under censoring. NC-CSF achieves substantially lower MAE than the baseline and remains close to the oracle.}
\label{fig:appendix_mae_survprob_censored}
\end{figure}

\begin{figure}[H]
\centering
\includegraphics[width=0.9\linewidth]{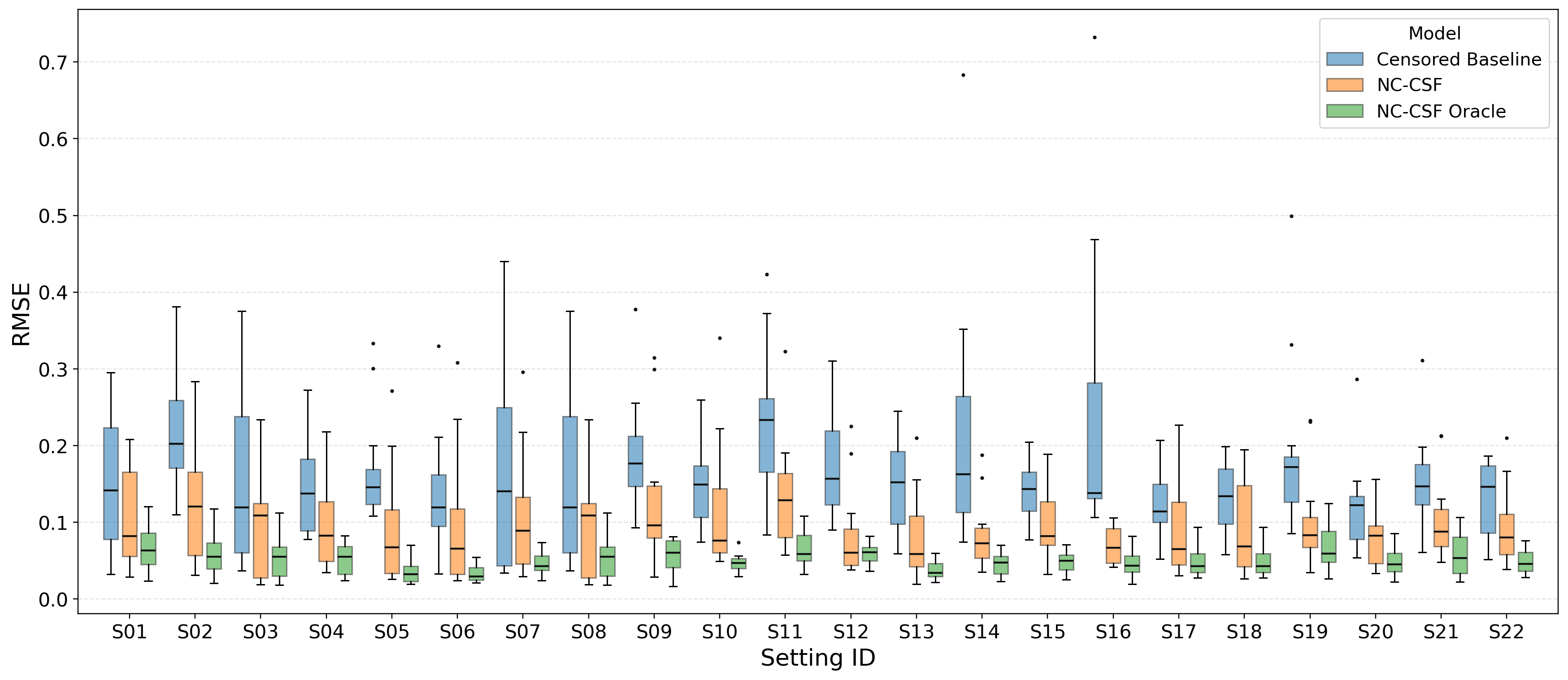}
\caption{RMSE for RMST estimation across settings S01--S22 under censoring. NC-CSF achieves substantially lower RMSE than the baseline and remains close to the oracle.}
\label{fig:appendix_rmse_rmst_censored}
\end{figure}

\begin{figure}[H]
\centering
\includegraphics[width=0.9\linewidth]{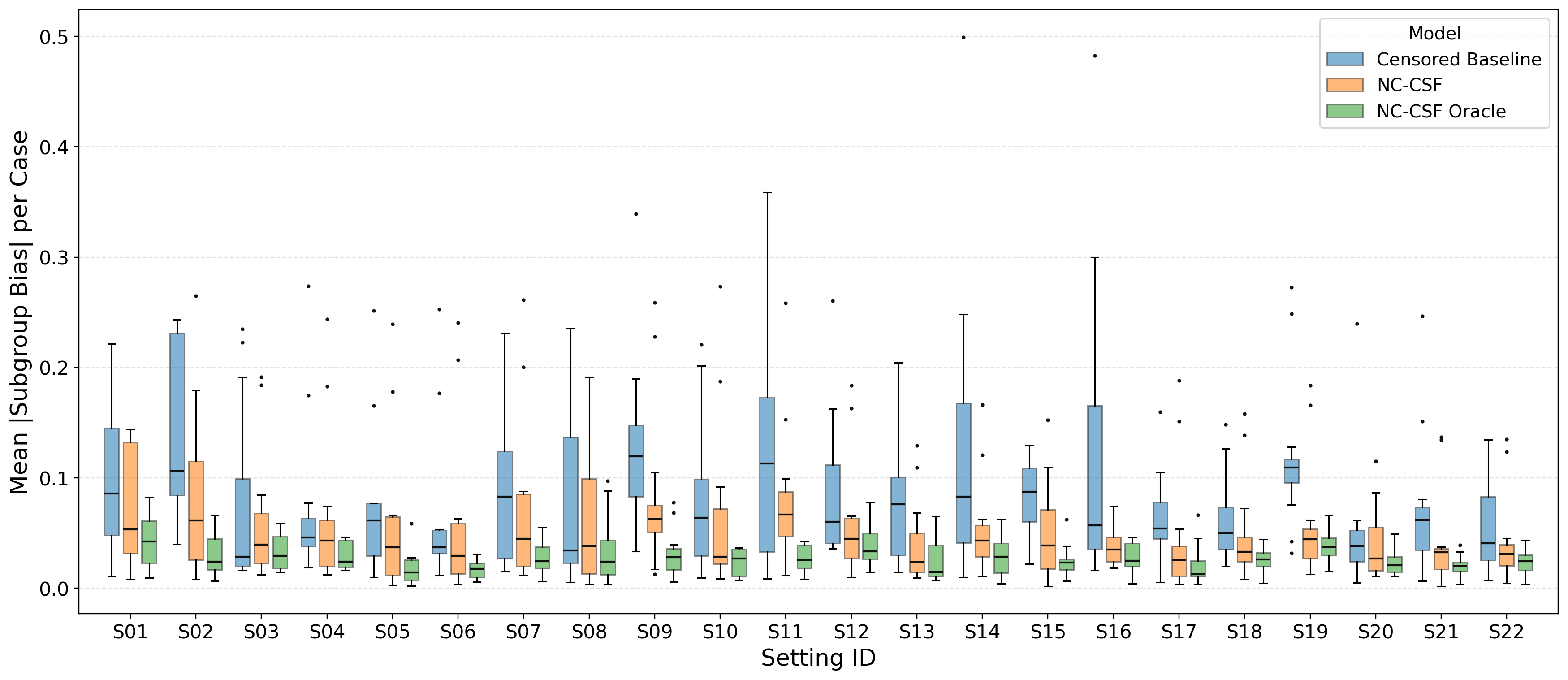}
\caption{Subgroup bias for RMST estimation across settings S01--S22 under censoring, using $X_0$-based subgroups. NC-CSF consistently reduces subgroup bias relative to the baseline and remains closer to the oracle, indicating improved heterogeneous treatment effect estimation.}
\label{fig:appendix_subbias_rmst_x0}
\end{figure}

\begin{figure}[H]
\centering
\includegraphics[width=0.9\linewidth]{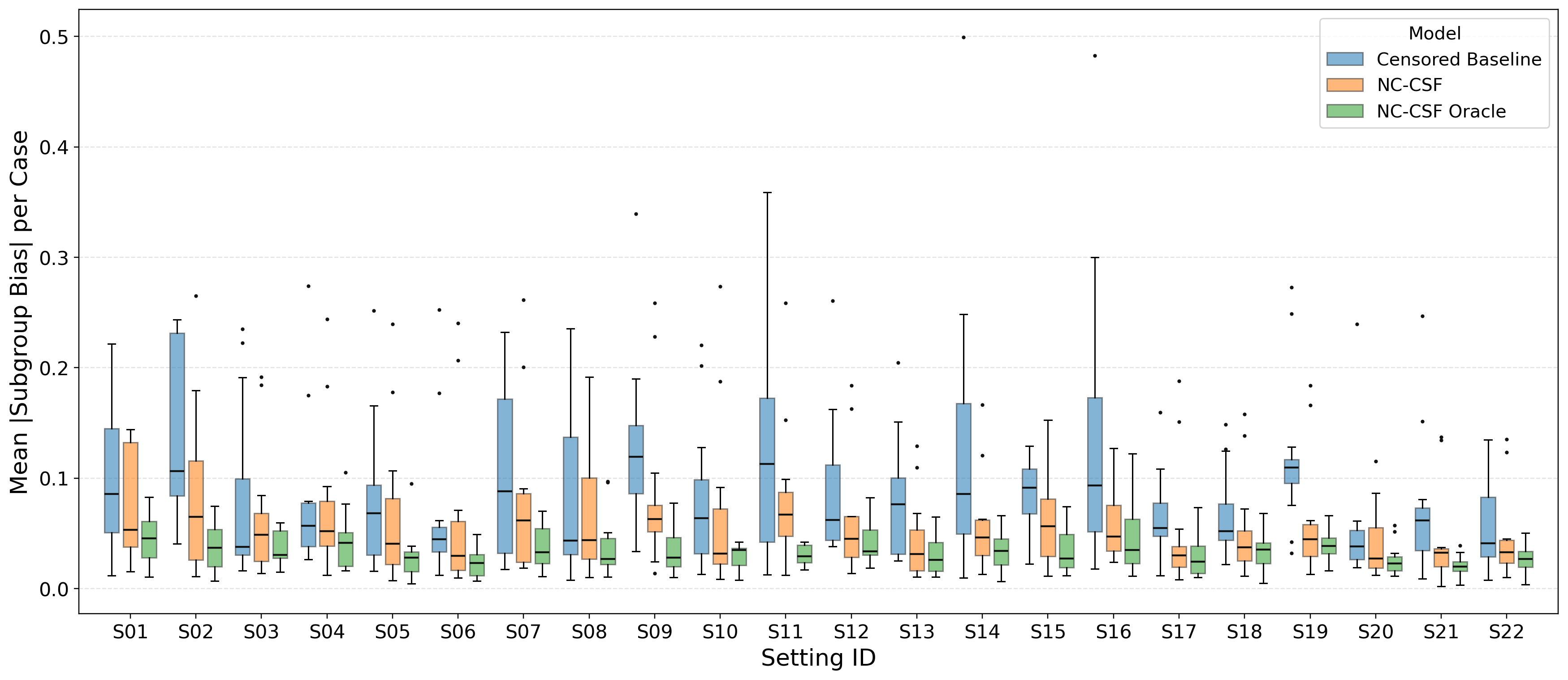}
\caption{Subgroup bias for RMST estimation across settings S01--S22 under censoring, using $X_0X_1$-based subgroups. NC-CSF consistently reduces subgroup bias relative to the baseline and remains closer to the oracle, indicating improved heterogeneous treatment effect estimation.}
\label{fig:appendix_subbias_rmst_x0x1}
\end{figure}

\begin{figure}[H]
\centering
\includegraphics[width=0.9\linewidth]{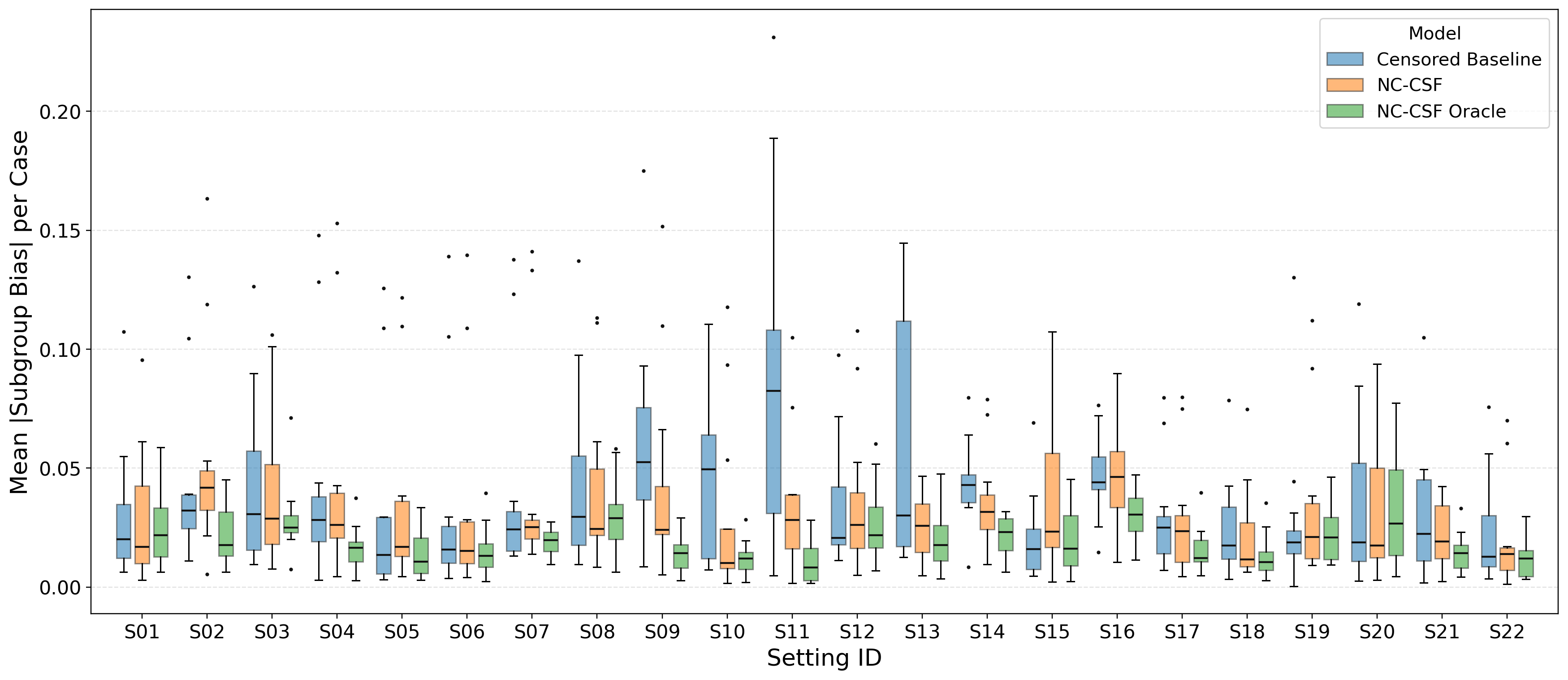}
\caption{Subgroup bias for survival probability estimation across settings S01--S22 under censoring, using $X_0$-based subgroups. NC-CSF consistently reduces subgroup bias relative to the baseline and remains closer to the oracle, indicating improved heterogeneous treatment effect estimation.}
\label{fig:appendix_subbias_survprob_x0}
\end{figure}

\begin{figure}[H]
\centering
\includegraphics[width=0.9\linewidth]{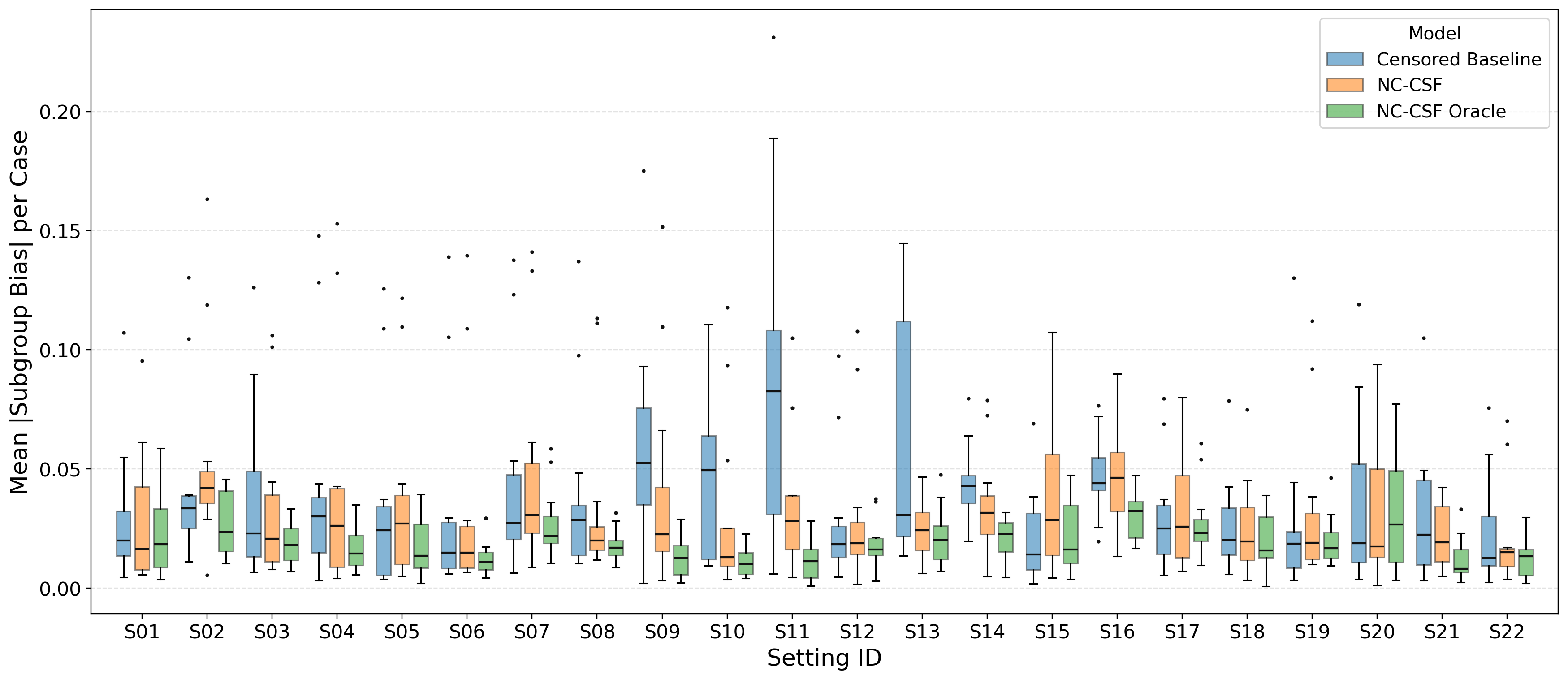}
\caption{Subgroup bias for survival probability estimation across settings S01--S22 under censoring, using $X_1$-based subgroups. NC-CSF consistently reduces subgroup bias relative to the baseline and remains closer to the oracle, indicating improved heterogeneous treatment effect estimation.}
\label{fig:appendix_subbias_survprob_x1}
\end{figure}

\begin{figure}[H]
\centering
\includegraphics[width=0.9\linewidth]{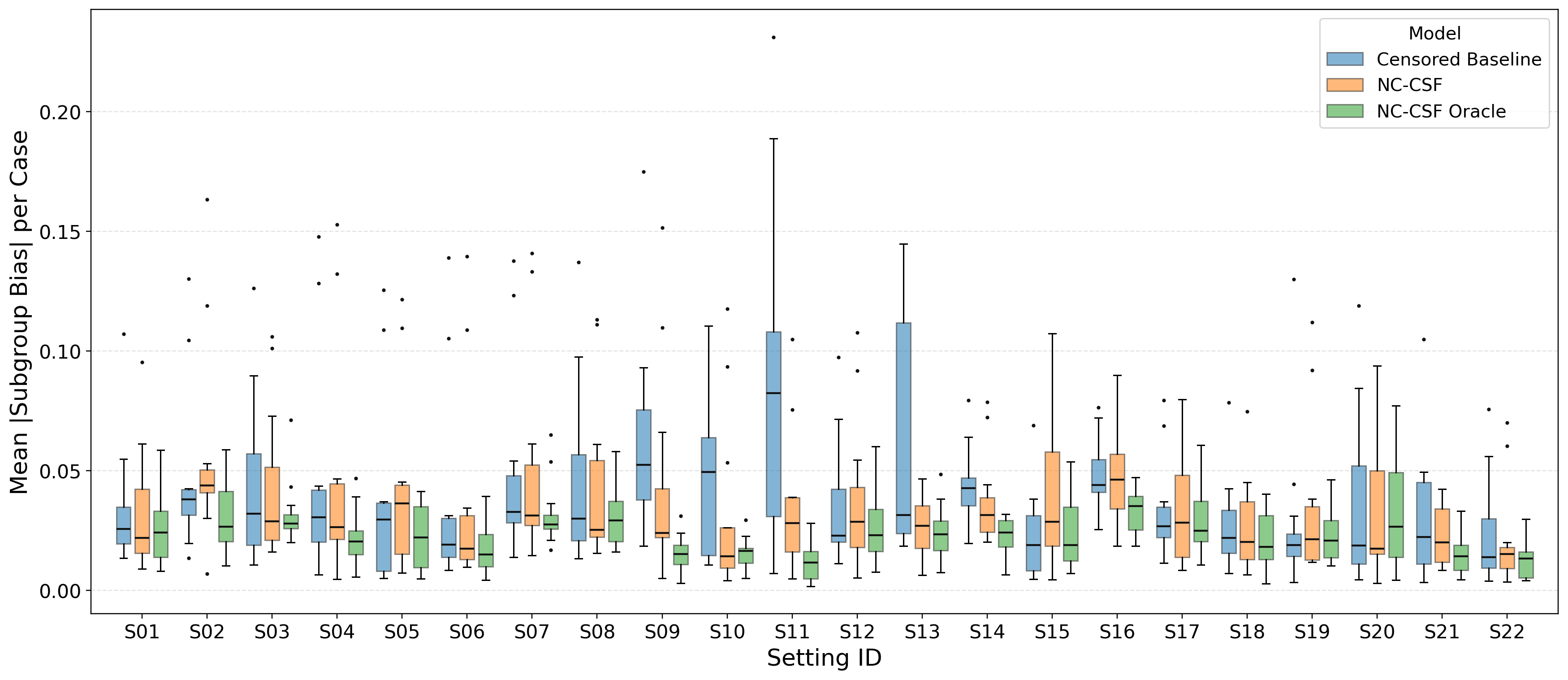}
\caption{Subgroup bias for survival probability estimation across settings S01--S22 under censoring, using $X_0X_1$-based subgroups. NC-CSF consistently reduces subgroup bias relative to the baseline and remains closer to the oracle, indicating improved heterogeneous treatment effect estimation.}
\label{fig:appendix_subbias_survprob_x0x1}
\end{figure}


\end{document}